\documentclass[11pt,reqno,tbtags]{amsart}
\usepackage{amssymb}
\usepackage[utf8]{inputenc}   
\usepackage{url}
\usepackage[square,numbers]{natbib}
\bibpunct[, ]{[}{]}{;}{n}{,}{,}

\title[Thresholds for election methods]
{Thresholds quantifying proportionality criteria for election methods}

\date{12 October, 2018}

\author{Svante Janson}
\address{Department of Mathematics, Uppsala University, PO Box 480,
SE-751~06 Uppsala, Sweden}
\email{svante.janson@math.uu.se}
\newcommand\urladdrx[1]{{\urladdr{\def~{{\tiny$\sim$}}#1}}}
\urladdrx{http://www.math.uu.se/svante-janson/}

\subjclass[2010]{} 

\numberwithin{equation}{section}

\renewcommand\le{\leqslant}
\renewcommand\ge{\geqslant}

\allowdisplaybreaks

\theoremstyle{plain}

\newenvironment{property}[1]%
{\renewcommand\thethxmetod{\textrm{\sc #1}}\begin{thxmetod}}{\end{thxmetod}}

\newtheorem{thxmetod}{}
\newenvironment{metod}[1]%
{\renewcommand\thethxmetod{\textrm{\sc #1}}\begin{thxmetod}}{\end{thxmetod}}

{\begin{metod}{#1}\rm}{\end{metod}}

{\par\smallskip\em\noindent\ignorespaces}{\par\smallskip}

{\par\em}{\par}

\theoremstyle{plain}
\newtheorem{theorem}{Theorem}[section]
\newtheorem{lemma}[theorem]{Lemma}

\newtheorem{corollary}[theorem]{Corollary}
\newtheorem{conjecture}[theorem]{Conjecture}
\newtheorem{theoremx}[theorem]{\Theoremx}
\newcommand\Theoremx{``Theorem''}
\newcommand\Theoremsx{``Theorems''}
\newtheorem{definition}[theorem]{Definition}

\theoremstyle{definition}
\newtheorem{example}[theorem]{Example}
\newtheorem{problem}[theorem]{Problem}
\newtheorem{remark}[theorem]{Remark}

\theoremstyle{remark}

\newenvironment{acks}{%
\section*{Acknowledgement}
}
{}

\newenvironment{romenumerate}[1][-10pt]{
\addtolength{\leftmargini}{#1}\begin{enumerate}
 }{\end{enumerate}}

\newcounter{oldenumi}
{\setcounter{oldenumi}{\value{enumi}}
\begin{romenumerate} \setcounter{enumi}{\value{oldenumi}}}
{\end{romenumerate}}

\newcounter{thmenumerate}

\newcounter{xenumerate}   

\newenvironment{val}[1][0pt]{
\addtolength{\leftmargini}{#1}
\begin{itemize}%
 }{\end{itemize}}

{\list{}{\rightmargin#2\leftmargin#1}\item[]}
{\endlist }

\newcommand\xfootnote[1]{\unskip\footnote{#1}$ $} 

\newcommand\pfitemx[1]{\par#1:}
\newcommand\pfitemref[1]{\pfitemx{\ref{#1}}}

\newcommand{\refT}[1]{Theorem~\ref{#1}}
\newcommand{\refTx}[1]{\Theoremx~\ref{#1}}
\newcommand{\refTs}[1]{Theorems~\ref{#1}}
\newcommand{\refTxs}[1]{\Theoremsx~\ref{#1}}
\newcommand{\refC}[1]{Corollary~\ref{#1}}
\newcommand{\refCs}[1]{Corollaries~\ref{#1}}
\newcommand{\refL}[1]{Lemma~\ref{#1}}
\newcommand{\refLs}[1]{Lemmas~\ref{#1}}
\newcommand{\refR}[1]{Remark~\ref{#1}}
\newcommand{\refRs}[1]{Remarks~\ref{#1}}
\newcommand{\refS}[1]{Section~\ref{#1}}
\newcommand{\refSs}[1]{Sections~\ref{#1}}
\newcommand{\refSS}[1]{Section~\ref{#1}}

\newcommand{\refD}[1]{Definition~\ref{#1}}
\newcommand{\refDs}[1]{Definitions~\ref{#1}}
\newcommand{\refE}[1]{Example~\ref{#1}}

\newcommand{\refApp}[1]{Appendix~\ref{#1}}
\newcommand{\refApps}[1]{Appendices~\ref{#1}}
\newcommand{\refConj}[1]{Conjecture~\ref{#1}}
\newcommand{\refTab}[1]{Table~\ref{#1}}
\newcommand{\refTabs}[1]{Tables~\ref{#1}}

\newcommand{\refFn}[1]{Footnote~\ref{#1}}

\begingroup
  \divide\count255 by 60
  \multiply\count255 by -60
  \advance\count255 by \time
  \ifnum \count255 < 10 \xdef\klockan{\the\count1.0\the\count255}
  \else\xdef\klockan{\the\count1.\the\count255}\fi
\endgroup

\newcommand\nopf{\qed}   

\newcommand{\sumin}{\sum_{i=1}^n}

\newcommand{\sumkn}{\sum_{k=1}^n}

\newcommand{\sumjS}{\sum_{j=1}^S}

\newcommand\set[1]{\ensuremath{\{#1\}}}
\newcommand\bigset[1]{\ensuremath{\bigl\{#1\bigr\}}}
\newcommand\Bigset[1]{\ensuremath{\Bigl\{#1\Bigr\}}}

\newcommand\xpar[1]{(#1)}
\newcommand\bigpar[1]{\bigl(#1\bigr)}
\newcommand\Bigpar[1]{\Bigl(#1\Bigr)}
\newcommand\biggpar[1]{\biggl(#1\biggr)}

\newcommand\bigabs[1]{\bigl|#1\bigr|}
\newcommand\Bigabs[1]{\Bigl|#1\Bigr|}

\def\rompar(#1){\textup(#1\textup)}    
\newcommand\xfrac[2]{#1/#2}

\newcommand\xqfrac[2]{#1/(#2)}

\def\xexp(#1){e^{#1}}
\newcommand\ceil[1]{\lceil#1\rceil}
\newcommand\floor[1]{\lfloor#1\rfloor}

\newcommand\setn{\set{1,\dots,n}}
\newcommand\nn{[n]}
\newcommand\ntoo{\ensuremath{{n\to\infty}}}
\newcommand\Ntoo{\ensuremath{{N\to\infty}}}
\newcommand\Vtoo{\ensuremath{{V\to\infty}}}

\newcommand\Stoo{\ensuremath{{S\to\infty}}}

\newcommand\punkt{.\spacefactor=1000}    
    
\newcommand\ie{i.e\punkt}
\newcommand\eg{e.g\punkt}
\newcommand\viz{viz\punkt}
\newcommand\cf{cf\punkt}

\newcommand\ii{\mathrm{i}}

\newcommand\bbC{\mathbb C}

\newcounter{CC}
\newcounter{cc}

\newcommand\ga{\alpha}

\newcommand\gd{\delta}
\newcommand\gD{\Delta}

\newcommand\gam{\gamma}

\newcommand\gk{\kappa}

\newcommand\gs{\sigma}

\newcommand\eps{\varepsilon}

\renewcommand\phi{\xxx}  

\newcommand\cA{\mathcal A}
\newcommand\cB{\mathcal B}
\newcommand\cC{\mathcal C}
\newcommand\cD{\mathcal D}
\newcommand\cE{\mathcal E}

\newcommand\cS{{\mathcal S}}

\newcommand\cU{{\mathcal U}}
\newcommand\cV{\mathcal V}
\newcommand\cW{\mathcal W}

\newcommand\qw{^{-1}}

\newcommand\oi{[0,1]}

\newcommand\dd{\,\mathrm{d}}

\newcommand\rhs{right-hand side}

\newcommand\phragmen{Phrag\-m{\'e}n}

\newcommand\xx[1]{^{(#1)}}

\newcommand\MM{S}
\newcommand\ellx{\ell}
\newcommand\lv{L}
\newcommand\pix{\overline\pi}
\newcommand\ls{(\ellx,\MM)}
\newcommand\Ls{(L,\MM)}
\newcommand\is{(1,\MM)}
\newcommand\party{_{\mathsf{party}}}
\newcommand\same{_{\mathsf{same}}}
\newcommand\sameq{_\sameqx}
\newcommand\tactic{_{\mathsf{tactic}}}
\newcommand\partyx{{\mathsf{party}}}
\newcommand\samex{{\mathsf{same}}}
\newcommand\sameqx{\ensuremath{\mathsf{same_{=}}}}
\newcommand\tacticx{{\mathsf{tactic}}}
\newcommand\qEJR{{EJR}}
\newcommand\qPJR{{PJR}}
\newcommand\qJR{{JR}}
\newcommand\DPC{\PSC}
\newcommand\PSCx{\mathsf{PSC}}
\newcommand\PSC{_{\PSCx}}
\newcommand\wPSCx{\mathsf{w}\PSCx}
\newcommand\wPSC{_{\wPSCx}}
\newcommand\EJRx{\mathsf{EJR}}
\newcommand\PJRx{\mathsf{PJR}}
\newcommand\SJRx{\mathsf{SJR}}
\newcommand\SSJRx{\mathsf{SSJR}}
\newcommand\EJR{_{\EJRx}}
\newcommand\PJR{_{\PJRx}}
\newcommand\SJR{_{\SJRx}}
\newcommand\SSJR{_{\SSJRx}}
\newcommand\qSx{\fS}
\newcommand\qS{_{\qSx}}
\newcommand\piS{\pi\qS}

\newcommand\cVW{\cW\compl}
\newcommand\compl{^{\mathsf c}}

\newcommand\fS{{\mathfrak S}}
\newcommand\fE{{\mathfrak E}}

\newcommand\opt{optimization}

\newcommand\MMMx{\mathfrak{M}}
\newcommand\MMM{^{\MMMx}}

\newcommand\Phrux{\ensuremath{\mathsf{Phr\text{-}u}}}
\newcommand\Phrox{\ensuremath{\mathsf{Phr\text{-}o}}}
\newcommand\Thax{\ensuremath{\mathsf{Th\text{-}add}}}
\newcommand\Thawx[1]{\ensuremath{\mathsf{Th\text{-}add}(#1)}}
\newcommand\Thex{\ensuremath{\mathsf{Th\text{-}elim}}}
\newcommand\Thox{\ensuremath{\mathsf{Th\text{-}o}}}
\newcommand\Thoptx{\ensuremath{\mathsf{Th\text{-}opt}}}
\newcommand\Thoptwx[1]{\ensuremath{\mathsf{Th\text{-}opt}(#1)}}
\newcommand\Bordax[1]{\ensuremath{\mathsf{Borda}(#1)}}
\newcommand\Divx[1]{\ensuremath{\mathsf{Div}(#1)}}
\newcommand\Qx[1]{\ensuremath{\mathsf{Q}(#1)}}
\newcommand\DHx{\ensuremath{\mathsf{D'H}}}
\newcommand\StLx{\ensuremath{\mathsf{StL}}}
\newcommand\Adamsx{\ensuremath{\mathsf{Adams}}}
\newcommand\vkx{\ensuremath{\mathsf{LR}}}
\newcommand\Droopx{\ensuremath{\mathsf{Droop}}}
\newcommand\BVx{\ensuremath{\mathsf{BV}}}
\newcommand\AVx{\ensuremath{\mathsf{AV}}}
\newcommand\CVx{\ensuremath{\mathsf{CV}}}
\newcommand\CVqx{\ensuremath{\mathsf{CV^{=}}}}
\newcommand\LVxx{\ensuremath{\mathsf{LV}}}
\newcommand\LVx[1]{\ensuremath{\mathsf{LV}(#1)}}
\newcommand\SNTVx{\ensuremath{\mathsf{SNTV}}}
\newcommand\STVx{\ensuremath{\mathsf{STV}}}

\newcommand\Phru{^{\Phrux}}
\newcommand\Phro{^{\Phrox}}
\newcommand\Tha{^{\Thax}}
\newcommand\Thaw[1]{^{\Thawx{#1}}}
\newcommand\The{^{\Thex}}
\newcommand\Tho{^{\Thox}}
\newcommand\Thopt{^{\Thoptx}}
\newcommand\Thoptw[1]{^{\Thoptwx{#1}}}
\newcommand\Borda[1]{^{\Bordax{#1}}}
\renewcommand\DH{^{\DHx}}
\newcommand\StL{^{\StLx}}
\newcommand\Adams{^{\Adamsx}}

\newcommand\vk{^{\vkx}}
\newcommand\Droop{^{\Droopx}}
\newcommand\Q[1]{^{\Qx{#1}}}
\newcommand\BV{^{\BVx}}
\newcommand\AV{^{\AVx}}
\newcommand\CV{^{\CVx}}
\newcommand\CVq{^{\CVqx}}
\newcommand\LV[1]{^{\LVx{#1}}}
\newcommand\SNTV{^{\SNTVx}}
\newcommand\STV{^{\STVx}}

\newcommand\StLn{Sainte-Lagu\"e}
\newcommand\DHn{D'Hondt}
\newcommand\Vkn{Method of Largest Remainder}
\newcommand\vkn{the method of Largest Remainder}

\newcommand\SNTVn{Single Non-Transferable Vote}
\newcommand\SNTVxn{SNTV}
\newcommand\BVn{Block Vote}
\newcommand\AVn{Approval Vote}
\newcommand\LVn{Limited Vote}
\newcommand\CVn{Cumulative Vote}
\newcommand\Than{Thiele's addition method}
\newcommand\Then{Thiele's elimination method}
\newcommand\Thoptn{Thiele's optimization method}
\newcommand\Thanww{Thiele's addition method with weights}

\newcommand\Thoptweakn{Thiele's weak optimization method}
\newcommand\Thaweakn{Thiele's weak addition method}

\newcommand\WWW{|\cW|}
\newcommand\piinf{\pi^*_*}
\newcommand\abbrev[1]{ \ensuremath{(#1)}}
\newcommand\lele{\ensuremath{1\le\ell\le S}}
\newcommand\ccA{\cC\setminus\cA}
\newcommand\cAx{\cA^*}
\newcommand\setx{\ensuremath}
\newcommand\tha{a}
\newcommand\thax{\widehat a}
\newcommand\thb{b}
\newcommand\thc{c}
\newcommand\bw{\overline w}
\newcommand\ww{\mathbf{w}}
\newcommand\vx{v_*}
\newcommand\citat[1]{``#1''}
\newcommand\sumCE{\sum_{C\in\cE}}
\newcommand\weakw{\mathsf{weak}}
\newcommand\kol{: }
\newcommand\xoo{_0^\infty}
\newcommand\Ex{\hat\cE}

\begin{document}

\begin{abstract} 
We define several different thresholds for election methods by considering
different scenarios, corresponding to different proportionality criteria
that have been proposed by various authors. In particular, we reformulate
the criteria known as DPC, PSC, JR, PJR, EJR in our setting.
We consider multi-winner 
election methods of different types,  using ballots with
unordered lists of candidates or ordered lists, and for
comparison also methods using only party lists.
The thresholds are calculated for
many different election methods. The considered methods include
classical ones such as BV, SNTV and STV (with some results going back to
e.g.\ Droop and Dodgson in the 19th century); we also study in detail several
perhaps lesser known methods by Phragmén and Thiele.
There are also many cases left as open problems.
\end{abstract}

\maketitle

\begingroup
\renewcommand\footnote[1]{\relax}
\tableofcontents
\endgroup

\section{Introduction}\label{S:intro}

We consider election methods where a number of persons are elected from some
set of candidates. For example, this is the case
in a multi-member constituency in
a parliamentary election or a local election, 
but also in many other situations such as
electing a board or a committee in an organization.
We will here use the language of a parliamentary election;
some other authors instead talk about \eg{} \emph{committee voting rules}.
We assume throughout that the number of \emph{seats}, i.e., the number of
elected representatives, is given in advance; we denote the number of seats  
by $S$.

Many different election methods have been suggested for this type of
elections, and many of them are, or have been, used somewhere.
For some important examples and discussions, from both mathematical and
political aspects, see \eg{}
\cite{Farrell,Politics,Pukelsheim};
for the election methods actually used at present
in parliamentary elections around the
world, see
\cite{IPU};
see also my surveys \cite{SJV6} (in Swedish) and \cite[Appendix E]{SJV9}.

One important aspect of election methods is whether they are
\emph{proportional} or not, \ie, whether different groups of voters
(parties)  get numbers of seats that are (more or less) proportional to
their numbers of votes.
Whether an election method is proportional or not is not
  precisely defined, even when there are formal parties.
Obviously, there are necessarily deviations from exact proportionality
since the number of seats is an integer for each party.
Moreover, in practice,
whether the outcome of an election is (approximatively)
proportional depends not only on the method, but also on the number of seats
in the constituency, and on other factors such as the number of parties and
their sizes.

One way that has been used by many authors to study proportionality of
election methods theoretically
is to formulate some criterion, which is supposed to be
a desirable  requirement for a method to be regarded as
proportional, 
and then investigate whether a particular method satisfies this criterion
or not.
A typical example is the 
\emph{Droop Proportionality Criterion}  (DPC)
by \citet{Woodall:properties} (see also \citet[p.~283]{Dummett}); 
this is formulated for STV type elections 
where each ballot contains a ranked list of candidates, 
i.e., election
methods with \emph{ordered ballots} in the notation introduced in \refS{Sprel}.

\begin{property}{Droop Proportionality Criterion}
If\/ $V$ votes are cast in an election to fill $S$ seats, 
and, for some whole numbers $\ell$ and $m$ satisfying $0 < \ell \le m$, more
than   $\ell\cdot V/(S+1)$ 
voters put the same $m$ candidates (not necessarily in the same order) 
as the top $m$ candidates in their preference listings,
then at least $\ell$ of those $m$ candidates should be elected. (In the event
of a tie, this should be interpreted as saying that every outcome that
is chosen with non-zero probability should include at least $\ell$ of these $m$
candidates.)
\end{property}

A related criterion is
\emph{Proportionality for Solid Coalitions} (PSC)
\cite{Tideman}, which differs from the DPC above only in that ``more
than $\ell V/(S+1)$ voters'' is replaced by ``at least $\ell V/S$ voters''.
\xfootnote{The formulation of the PSC in \cite{Tideman} is somewhat
  different, but is easily seen to be equivalent to our version; 
\cf{} \refTx{Tx5}.
}
See further \refS{SPSC}.

Further examples of similar proportionality criteria
are
JR, PJR and EJR discussed in \refS{SJR}.

Several authors study such criteria and whether specific elections methods
satisfy them or not; for example, \citet{Woodall:properties} points out that
any version of STV  (Single Transferable Vote)
satisfies the DPC above, at least
provided the Droop quota $V/(S+1)$ is used without rounding.
See  \refS{SPSC}.
\xfootnote{
Recall that there are many versions of STV, 
see \refApp{ASTV};
recall also that in many versions, the Droop quota is rounded to an integer,
which may lead to minor violations of the DPC as stated above.
}

The purpose of the present paper is to promote an alternative
and more quantitative, but closely
related, way of studying such proportionality properties.
Instead of, as in DPC and PSC above, 
fixing a number of voters and asking whether
a set of voters of this size 
always can get at least $\ell$ candidates elected (under certain
assumptions on how these voters vote), we turn the question around and ask
how large a set of voters has to be in order to be guaranteed to get
at least $\ell$ candidates elected
(again under certain assumptions on how they vote).
Corresponding to the DPC and PSC
criteria above, we thus define, for a given
election method,
the threshold
$\pi\DPC\ls$ as the minimum (or rather infimum) fraction of votes that
guarantees at least $\ell$ elected out of $S$
under the conditions above, see \refS{SPSC} for
details.
The DPC criterion then can be formulated as $\pi\DPC\ls\le \ell/(S+1)$ (for all
$S\ge1$ and $\ell\le S$), and the PSC criterion as 
$\pi\PSC\ls<\ell/S$ (ignoring a possible
minor problem in the case of ties, see
further \refR{Rrefined}).

An advantage  of this quantitative version of proportionality
criteria
is thus that it treats both DPC and PSC at once.
Moreover, it is not necessary to consider only the thresholds $\ell/(S+1)$
and $\ell/S$. In fact, \citet{AzizLee} have defined a property
{$q$-PSC}, where the conclusion above is supposed to hold for every set
of at least $q \ell$ voters; with $q=\gk V$ this is thus the same as
$\pi\PSC\ls<\ell\gk$. Thus, to consider this parametric family of
proportionality criteria is equivalent to considering the quantity
$\pi\PSC$. 
As another example, note that the property \emph{majority} in
\cite{Woodall:properties} in our notation can be written as 
$\pi\DPC(1,S)\le \xfrac12$.

Another advantage  of the quantitative version 
is
that it enables us to see not only whether a criterion hold or
not, but also how badly it fails if it does. 
Conversely, if a
criterion holds, perhaps it
turns out that it holds with some marginal.
A threshold such as $\pi\PSC$ therefore gives more information, and 
 may give more insight in the behaviour of an election method.
Furthermore, even if the goal is to study a specific criterion, 
the arithmetic of the thresholds may simplify or clarify proofs. 
(See \eg{} \refR{RThoptww} and \refE{Eonlyweak} for examples of this.)

We have here used DPC and PSC as examples only. Much of our work below deals
with other similar properties, and we will define a number of thresholds
$\pi\ls$ for different scenarios. 

Our point of view is far from new. In fact, thresholds of the type $\pi\ls$
figure prominently in the 19th century discussion on parliamentary reform in
(for example) Britain, \eg{} in \citet{Droop} and \citet{Dodgson}; 
for party list systems, they were defined  
by \citet{RaeEtal}, \cite{Rae} under the name
\emph{threshold of exclusion}, see also \citet{Lijphart}.
Nevertheless, the idea is still useful and seems to deserve further
attention;
in the present paper we extend the study of such thresholds to further
scenarios and further election methods.

We consider many different election methods.
These include several classical ones such as 
\DHn's method, \StLn's method, \BVn, \LVn, SNTV and STV;
we also study, in a rather large part of the paper, 
several election methods by \phragmen{} and Thiele.
(Nevertheless, many election methods remain to be studied.)
We partly survey know results, but most of the results seem to be new.

\begin{remark}\label{Rquota}
  In the present paper, we aim at computing $\pi\ls$ in many cases, but we
  do not try to give any criteria for what is good and what is not.
Nevertheless, note that historically, the two main thresholds that have been
proposed for various scenarios
are $\ell/S$ and $\ell/(S+1)$; 
in other words that the required number of votes $\cW$ to guarantee 
$\ell$ seats is
$\ell$ times the \emph{Hare quota} $V/S$ or the \emph{Droop quota} $V/(S+1)$,
which both have been put forth as a natural number of votes deserving a seat.
The Hare quota is simpler and perhaps more intuitive; see \eg{} \refTx{Tx5}
below;
the Droop quota was
proposed  by \citet{Droop}
based on the argument that if a candidate has more than
$1/(S+1)$ of the votes, then there cannot be $S$ others with at least as
many votes, so the candidate ought to be among the $S$ elected; 
\xfootnote{This argument holds for  \eg{} SNTV (\refApp{ASNTV}) and
\CVn{} (\refApp{ACV}; the case really discussed by \citet{Droop}). 
}
see also \cite{Dodgson} and \cite{Tideman}.
In our context, $\ell/(S+1)$ is essentially the best (\ie, smallest)
threshold that can be achieved, see \refR{Rbest};
hence $\pi\ls=\ell/(S+1)$ can be regarded as the optimal value.
We will see many cases where this value is achieved, and many where it is not.
(See \refTab{tab:optimal} in \refApp{Anum} for some numerical values.)
\end{remark}

\section{Notations and general definitions}\label{Sprel}
We  consider only election methods where each voter votes by leaving a
ballot in a single round. 
The outcome is then determined by these ballots and some mathematical algorithm.
Furthermore, we assume that either each ballot
contains the name of a party, or that each ballot contains a list of one or
several candidates. (In fact, the first case can be regarded as a special
case of the second, since each party may be regarded as a predefined list of
candidates.) Moreover, we consider both methods where
the order of the names on the ballot matter and those where it does not,
and it is usually important to distinguish between these two cases.
Hence, we consider the following three different types of election methods;
we will see that  these three types often have to be considered separately.
\xfootnote{A more general type of ballot, which includes both unordered
  ballots and ordered ballots as special cases, is a ballot with a 
\emph{weak ordering} of the candidates. Election methods for such ballots are
  sometimes discussed in the literature, and have even been used,
see for example 
\cite[\S18.1-2]{SJV9} and the references there.
We will not consider such ballots in the present paper,
but see
\cite{AzizLee} for versions of PSC (and implicitly of
our $\pi\PSC$) for such ballots and methods. 
}

\begin{description}
\item[party ballots]
The candidates are organized in parties, with separate lists of
candidates. Each ballot contains the name of a party, and seats are
distributed among the parties. In this case we are only interested in the
distribution of seats among the parties. (We thus ignore how seats are
distributed within each party.)
\item[unordered ballots] 
Each ballot contains a set of one or several candidates; their order on the
ballot does not matter.
\item[ordered ballots]  
Each ballot contains a list of one or several candidates, in order of
preference.
\xfootnote{An extreme version, more popular in theoretic work than in
  practical use, is that each voter ranks \emph{all} candidates in a linear
  order.}
\end{description}

\subsection{Some notation}\label{Snot}
Let $\cV$ denote the set of all voters, and $V:=|\cV|$ the number of votes.
Formally, an election may be regarded as 
a family $(\gs_\nu)_{\nu\in\cV}$ of ballots,
where each ballot belongs to the set $\cB$  of possible ballots.
\xfootnote{\label{fnxgs}%
We consider only election methods that are \emph{anonymous}, so the order of
the ballots does not matter. Hence, an
election can equivalently be described by a sequence $(x_\gs)_{\gs\in\cB}$
of integers $\ge0$,
where
$x_\gs$ is the number of votes for $\gs$, i.e., the number of ballots $\gs$.
}
Let $\cC$ be the set of parties in the party ballot case and otherwise the set
of candidates. Then,
in the party ballot case, $\cB=\cC$ (the set of parties); with unordered
ballots, $\cB$ is the set of subsets of $\cC$; 
with ordered ballots, $\cB$ is the set of sequences of distinct elements of 
$\cC$. (We assume tacitly that $\cC$ is finite, but sufficiently large when
needed.) 

\begin{remark}\label{Rsum}
For convenience, we may identify the voters with their ballots,
and write \eg{} $\gs\in\cW\subset\cV$,
thus treating $\cV$ as a multiset of ballots.
Similarly,
  we often abuse notation and write $\sum_{\gs\in\cV}$ instead of 
$\sum_{\nu\in\cV}$ (with summands depending on $\gs$ or $\gs_\nu$, respectively).
\end{remark}

 We consider only election methods where the number of seats,
$S$, is fixed in advance. We are mainly interested in the case $S\ge2$, but
the case $S=1$ (single-member constituency) 
is included for completeness and comparisons.

The \emph{outcome} of the election is the set $\cE$ of elected candidates. 
Thus, for
elections with unordered or ordered ballots, $\cE$ is a subset of $\cC$ with
$|\cE|=S$. (In the case of elections with party ballots, $\cE$ can be
regarded as a multi-set; we will not use the notation $\cE$ in this case.)
Note that there might be ties and therefore several possible outcomes for a
given set of ballots $(\gs_\nu)_{\nu\in\cV}$. Formally, we regard the different
possibilities as different elections; thus an election can formally
be defined as a
pair $\bigpar{(\gs_\nu)_{\nu\in\cV},\cE}$ of a sequence of ballots and a
possible outcome.

Some general mathematical notation:
$x\land y :=\min\set{x,y}$;
$x\lor y :=\max\set{x,y}$;
$\nn:=\setn$;
$|\cS|$ is the number of elements in a set $\cS$;
$H_n:=\sumin 1/i$, the $n$-th harmonic number. 

\begin{remark}\label{Rreal}
It is possible to let different voters have different weights, which in
principle could be any positive real numbers. 
Then, the ``number of voters'' in a subset of $\cV$ has to be
interpreted as their total weight, which thus is a (positive)  
real number, not necessarily an integer; similarly $\sum_{\gs\in\cV}$ (see
\refR{Rsum}) has to be interpreted with weights.
We leave the trivial modifications for this extension to the reader, and
continue to talk about ``numbers'' of votes.
\end{remark}

\begin{remark}\label{Rhomo}
The election methods considered below are almost all \emph{homogeneous},
meaning that the outcome remains the same if the number of votes for each
type of ballot is multiplied by the same number; 
in other words, only the proportions of votes matter. 
\xfootnote{Note that this fails for quota methods and STV if the quota is
rounded (up or down) to an integer, which is often the case in practical
uses. This is mathematically inconvenient, and theoretically bad also for
other reasons, but has very little importance in large elections.
We consider  mainly  mathematically ideal versions without rounding, and
they are homogeneous, as are all other methods considered here.
For an inhomogenous case, see \refR{RQrounding}.
}
Hence the methods are well-defined also when the ``number of votes'' for each
type of ballot is a positive rational number. 
In fact, the homogeneous methods considered here (and all reasonable
homogeneous methods) all apply to
the case when these numbers are arbitrary positive real numbers as in 
\refR{Rreal}.

In some cases, we find it convenient to use this and  
allow vote numbers to be arbitrary
positive real numbers.
This is not essential; real numbers can be approximated by rational numbers
if necessary, and rational numbers can be multiplied by a common denominator
to get an equivalent election with integer numbers of votes. 
\end{remark}

\subsection{Proportionality thresholds}
As explained in the introduction, our goal is to define and study thresholds
of proportionality $\pi(\ellx,S)=\pi\MMM(\ellx,S)$, where $\MMMx$ denotes the
election method. Informally, this is the smallest proportion
of  votes that a set $\cW$ of voters may have in order to be
guaranteed at least $\ellx$ elected candidates out of $S$.
Here, the set $\cW$ can be an organized group (party) or not, and we may
consider several different situations regarding both the votes from $\cW$
and the votes from the other voters $\cVW:=\cV\setminus\cW$. For example,
the voters in $\cW$ may all vote with identical ballots, or according to
some organized tactical scheme, or be unorganized with different ballots
that happen to contain some common candidate(s).
Hence, we will define several such thresholds $\pi$, which we distinguish
by subscripts. (As above, the election method is indicated by superscript.)
Further versions are possible, and are left to the readers imagination.

In general, a subscript $\fS$ defines some  \emph{scenario}, 
\ie, a class of possible elections with some 
restrictions on the votes from the set
$\cW$ and possibly also on the votes from
$\cVW$; furthermore, $\fS$ specifies, for a given $\ellx$ and $S$, what we
mean by a \emph{good} outcome. 
This should mean, in some sense, that $\cW$ gets at
least $\ellx$ candidates elected; however, as is seen below, this can be
made precise in different ways, 
since in some scenarios
the voters in $\cW$ might vote for partly different candidates, and some
voters in $\cVW$ also might  vote for some of these.
We say that an outcome that is not good is \emph{bad}.
(\citat{Good} and \citat{bad} are always from the perspective of the chosen 
set $\cW$ of voters.
$\cW$ will always denote such a set of voters, arbitrary but with the
scenario
specifying some assumption on their votes.)

Given a scenario $\fS$ and an
election method $\MMMx$,
we define
$\pi(\ellx,S)=\piS\MMM(\ellx,S)$ as the smallest proportion of votes that
guarantees a good outcome. More precisely and formally:

\begin{definition}\label{Dpi}
\begin{equation}\label{pi}
  \piS\MMM(\ellx,S):=\sup\Bigset{\frac{|\cW|}{|\cV|}: \text{elections satisfying
        $\fS$ with a bad outcome}}.
\end{equation}
\end{definition}

\begin{remark}\label{Rties}
One reason for the form  \eqref{pi} of the definition is that
when there are ties, the same set of ballots might lead to both a good and a
bad outcome. (In practice, the result is typically decided by lot.)
The formulation \eqref{pi} includes such cases; thus the idea is that 
$|\cW|/|\cV|>\pi(\ellx,S)$ guarantees a good outcome, even in the event of
ties. 
Similarly, there exist election methods that are randomized algorithms, so
the outcome may be random. (E.g.\ some versions of STV, see \refApp{ASTV}.) 
In such cases, any case where there is a
positive probability of a bad outcome is included in \eqref{pi}.
\end{remark}

\begin{remark}
  Of course, the outcome might be good also in some cases where
$|\cW|/|\cV|<\pi(\ellx,S)$ (depending on, e.g., how other voters vote);
to have more that $\pi(\ellx,S)$ of the votes is sufficient but not
necessary for a good outcome.
\end{remark}

\begin{remark}\label{Rspecial}
  The case $\ellx=1$ is of particular interest: $\piS(1,S)$ is the smallest
  proportion that guarantees the election of at least one representative of
  $\cW$. 
This is called the \emph{threshold of exclusion} 
\cite{RaeEtal,Lijphart}.

Also the other extreme case $\ell=S$ is of interest: $\piS(S,S)$ is the
smallest proportion that guarantees $\cW$ to get all seats, excluding all
minorities (assuming $\fS$). 

As another example, if $S$ is odd, then 
$\piS\bigpar{(S+1)/2,S}$ is the proportion of votes that guarantees a
majority of the seats.
\end{remark}

\begin{remark}\label{Rrefined}
  Typically, the supremum in \eqref{pi} is attained because ties might appear
when $|\cW|/|\cV|=\piS(\ellx,S)$, \cf{} \refR{Rties}; in such cases, a
proportion of votes exactly equal to $\piS(\ellx,S)$ is not enough to
guarantee a good outcome.
We may indicate whether this happens or not by a more refined notation:
if $p$ denotes the supremum in \eqref{pi}, we may write
$\piS(\ellx,S)=p+$  or $\piS(\ellx,S)=p-$,
with $p+$ meaning that the supremum is attained, i.e., that there exists a
bad outcome with a proportion of the votes exactly $p$, and $p-$ meaning
that there is no such bad outcome.
In other words, if $\piS(\ellx,S)=p-$, then a proportion of votes equal to
$p$ guarantees a good outcome, while if $\piS(\ellx,S)=p+$, then we need
strictly more than $p$ in order to be sure.

For example, the formulation of the DPC in \refS{S:intro} uses ``more than'',
and is thus equivalent to $\pi\DPC(\ellx,S)\le \frac{\ellx}{S+1}+$,
while PSC  uses ``at least'',
and is thus equivalent to $\pi\PSC(\ellx,S)\le \frac{\ellx}{S}-$.

We will occasionally use the refined notation, but usually not; this is then
left to the reader.

Note that almost all values calculated below are of the \citat+ type,
because of the possibility of ties at the threshold, but we note
a few, more or less trivial, exceptions with a \citat{$-$} type:
\begin{romenumerate}
\item 
When $\piS\ls=1$:
$\piS\ls=1+$ means that even if $\cW$ comprises all voters, this does not
guarantee all seats, which does not happen for reasonable methods and
scenarios, while
$\piS\ls=1-$ is perfectly possible, also for $\ell<S$, see 
\eg{} \refT{TPhrPSCbad}.
\item \label{-irrational}
When $\piS\ls$ is irrational. 
(Unless we allow real numbers of votes, see \refR{Rreal}.)
This is usually not the case, but can
happen \eg{} for the divisor method $\Divx\gam$ in \refT{Tdiv} 
with irrational $\gam$
(never used in practice as far as I know) or for 
the Estonian method (\refFn{fnEst} in \refApp{Alist}).
\item 
Some cases for
 election methods with special
  tie-breaking rules.
For example, for \DHn's method with ties always broken in favour of the
largest party, 
$\pi\DH\party\ls=\frac{\ell}{S+1}+$  for $\ell\le(S+1)/2$, but
$\pi\DH\party\ls=\frac{\ell}{S+1}-$ for $\ell>(S+1)/2$, \cf{} \eqref{party-DH}.
Another example is given in \refR{RAVtie}.
\end{romenumerate}

Hence, the property that some given proportion $p\in\oi$ of votes is enough to
guarantee $\ell$ seats, which really means $\pi\ls\le p-$, is 
in practice 
equivalent to $\pi\ls<p$.
(Note that this holds even in the exceptional case \ref{-irrational}, 
since then $\pi\ls=\WWW/V$ cannot occur.)
For example, 
in practice, 
PSC can be written as 
$\pi\PSC(\ellx,S)< \frac{\ellx}{S}$, with strict inequality,  when
$\ell<S$. 
\end{remark}

\begin{remark}
  Formally, the refined notation in \refR{Rrefined}
means that we regard $\piS(\ellx,S)$ as a
  member of the \emph{split interval} (also called \emph{two arrow space}),
  which contains two elements $x+$ and $x-$ for each $x\in\oi$,
 and that \eqref{pi} is interpreted in this space
  with its natural order, and a real $x$ identified with $x+$.
The split interval, with its order topology, has interesting topological
properties and is a standard example in topology, but we do not need any of
this here.
\end{remark}

In some cases, $\piS(\ellx,S)$ may be determined by some bad outcome with
small $V=|\cV|$, while a smaller proportion of votes suffices to get a good
outcome when $V$ is large. 
This may happen, for example, becuse of rounding effects in
the election method
as in \refR{RQrounding}. 
It may also happen because some tactical voting scheme requires votes
to be split \eg{} equally between two candidates, which is impossible when
$V$ is odd, see \refE{EpixSNTV}.
In such cases it is more interesting to consider only large $V$, or
more precisely, the limit as $V\to\infty$, and define
\begin{equation}\label{pix}
  \pix_{\fS}(\ellx,S):=\limsup_{|\cV|\to\infty}
 \Bigset{\frac{|\cW|}{|\cV|}: \text{elections satisfying
        $\fS$ with a bad outcome}}.
\end{equation}

Note that by the definitions \eqref{pi} and \eqref{pix},
always
\begin{align}\label{pixpi}
  \pix_{\fS}\ls\le\piS\ls.
\end{align}

\begin{remark}  \label{Rpix}
For an homogeneous election method, and a scenario that does not assume that
the voters in $\cW$ agree to split their votes on different lists according
to some strategy, it is easy to see that $\pix\MMM\qS\ls=\pi\qS\MMM\ls$.
In fact, consider any election with a bad outcome, and replace each ballot
by $N$ identical ballots, for some large integer $N$.
This gives a new election with 
(by our assumptions on $\MMMx$ and $\qSx$)
the same elected set $\cE$ and thus a bad outcome with the same proportion
$\WWW/V$, but $V$ replaced by $NV$. 
Letting \Ntoo, we obtain 
$\pix\MMM\qS\ls\ge\WWW/V$, and thus
$\pix\MMM\qS\ls\ge\pi\qS\MMM\ls$, showing that  the general inequality
\eqref{pixpi} is an equality in this case.
This applies to all
election methods considered in the present paper except versions of quota
methods and STV with rounding of the quota,
and to all scenarios defined below except $\tacticx$.
\end{remark}

In most cases studied below, $\pix=\pi$
is obvious from \refR{Rpix}, or follows from the given proofs;
we will usually not comment on this or mention $\pix$ at all in such cases,
leaving this to the reader. 
We will thus discuss $\pix$ mainly in cases where it differs from $\pi$.
(In such cases we  instead often leave $\pi$ to the reader.)
\smallskip

In the next section we give a few general properties of the definition.
In the following sections we put life in the general definitions above
by specifying some scenarios $\fS$, and calculating $\piS\MMM\ls$ for some
methods $\MMMx$.
We usually treat the three different types of ballots separately.
The case of party ballots in \refS{Sparty} is straightforward, and can be seen
as a warm-up to the remaining sections with ordered and unordered ballots,
where there
are many reasonable choices of definitions of $\fS$.
Our goal  is to investigate some of these choices of $\fS$,
and to calculate $\piS\MMM\ls$ for various election methods $\MMMx$.
(In a few cases we give only some bounds.)
Note that lower bounds always are found by giving examples, where the
election method $\MMMx$ gives (or may give, in case of ties) a bad
outcome for the chosen $\fS$.

There are many election methods, including many  not mentioned in the present
paper,
and we define thresholds $\piS$ for several
different scenarios (and further may be constructed), so
there are many potential combinations and we will only give complete results
for some of them. There are lots of cases remaining; these should be
regarded as open problems (some are stated explicitly for emphasis), 
and we invite other researchers to continue this work.

The scenarios and thresholds that we consider are defined in
Definitions 
\ref{Dparty} and  \ref{Dparty2} ($\pi\party$),
\ref{Dsame} ($\pi\same$),
\ref{Dtactic} ($\pi\tactic$),
\ref{DPJR} ($\pi\PJR$),
\ref{DEJR} ($\pi\EJR$),
\ref{DPSC} ($\pi\PSC$),
\ref{DwPSC} ($\pi\wPSC$).

The election methods that are studied in the paper are the following;
each is listed after the abbreviation that we use.
Brief definitions, and some alternative names, 
are given in \refApp{Amethods}.

\begin{romenumerate}
\item Party ballots
  \begin{enumerate}
  \item $\DHx$: D'Hondt's method.
  \item $\StLx$: Sainte-Lagu\"e's method.
  \item $\Adamsx$: Adams's method. 
 \item $\Divx\gam$: Divisor method with divisors $d(n)=n-1+\gam$, $\gam\in\oi$. 
  \item $\vkx$: \Vkn.  
  \item $\Droopx$: Droop's method. 
  \item $\Qx\gd$:
  Quota method with quota $V/(S+\gd)$, $\gd\in\oi$. 
  \end{enumerate}

\item Unordered ballots
  \begin{enumerate}
  \item $\BVx$: \BVn. %
  \item \AVx: \AVn. 
  \item \SNTVx: \SNTVn. 
  \item \LVx{L}: \LVn{}, with $L$ votes per voter. 
  \item \CVx: \CVn. 
  \item \CVqx: Equal and Even \CVn. 
  \item $\Phrux$: \phragmen's unordered method. 
  \item $\Thoptx$, $\Thoptwx{\ww}$: Thiele's optimization method. 
  \item $\Thax$, $\Thawx{\ww}$: Thiele's addition method. 
  \item $\Thex$: \Then. 
  \end{enumerate}

\item Ordered ballots
  \begin{enumerate}
  \item \STVx: Single Transferable Vote.  
        (This is really a family of methods, see \refApp{ASTV}.)
  \item $\Phrox$: \phragmen's ordered method. 
  \item $\Thox$: Thiele's ordered method.
  \item $\Bordax{\ww}$: Borda method with weights $\ww$.
  \end{enumerate}
\end{romenumerate}

Some numerical values for small $S$ are given as illustrations  and for
comparisons in \refApp{Anum}.

For methods with ordered or unordered ballots, there are no formal parties
in the definition of the methods.
Nevertheless, we will sometimes, in particular in examples, talk about
parties also when discussing such methods; we then mean a (formal or
informal) group of voters with a common interest, 
typically voting with
identical ballots, or possibly according to an organized strategy.

\section{General properties}\label{Sgen}

We assume always $1\le \ell\le S$.
By the definitions \eqref{pi} and \eqref{pix},
\begin{equation}
  \label{pipix}
 0\le \pix\ls\le\pi\ls\le 1. 
\end{equation}

We state some further simple general results. These hold for 
any reasonable scenario $\fS$ and
any reasonable election method, but since we do not formally  define what we
mean by 'reasonable', and we have not made any formal restrictions on the
scenarios $\fS$ and what a bad outcome might be, we call these results
\Theoremx. 
They are real Theorems for the scenarios and the election methods
considered in the present paper.
\xfootnote{A scenario with stricter assumptions on the votes from $\cVW$
than on the votes from $\cW$
might fail some of these \Theoremsx, but we do not consider any such case.}

\begin{theoremx} \label{Tx1}
  \begin{equation}\label{tx1}
\pi(1,S)
\ge\frac{1}{S+1}.
  \end{equation}
\end{theoremx}
\begin{proof}
  Suppose that there are $S+1$ parties with one candidate each,
and that each voter votes for  only one of the candidates, 
with equally many voters ($V/(S+1)$) for each.
By symmetry, none of the candidates is guaranteed a seat. (A typical election
method would select $S$ of the $S+1$ candidates by lot.) 
Hence, if $\cW$ is the set of voters voting for candidate $A$, 
so $|\cW|/|\cV|=1/(S+1)$,
then a bad event can occur.
\end{proof}

\begin{theoremx} \label{Tx2}
If\/ $(S+1)/\ell$ is an integer, then
  \begin{equation}\label{tx2}
\pi\ls
\ge\frac{\ell}{S+1}.
  \end{equation}
\end{theoremx}
\begin{proof}
  By the same proof as \refTx{Tx1}, now taking $(S+1)/\ell$ parties with exactly
  $\ell$ candidates each (the last requirement disappears for party ballots
  methods), 
and equally many voters each.
\end{proof}

\begin{remark}\label{Rlower}
Note that \eqref{tx2} is not valid in general for arbitrary $\ell$ and $S$, at
least not in the case $\ell>(S+1)/2$. 
  As is well-known, with the Block Vote
method, a group of voters with a majority of the votes 
will get all seats if they vote for the same candidates.
Hence $\piS\BV\ls=1/2$ for every $\ell\le S$ and typical $\fS$, see \refT{TBV}.
\end{remark}

\begin{theoremx} \label{Tx4}
If\/ $1\le \ell\le S$, then
  \begin{equation}\label{tx4}
\pi\ls + \pi(S+1-\ell,S)\ge1.
  \end{equation}
 More generally,  for any $\ell_1,\dots,\ell_m$
with $\ell_1+\dots+\ell_m>S$, 
\begin{equation}
  \label{tx4b}
  \pi(\ell_1,S)+\dots+\pi(\ell_m,S)\ge1.
\end{equation}
\end{theoremx}
\begin{proof}
  Suppose that \eqref{tx4} fails. 
Then we can find a (large) $V=|\cV|$ and a set of voters
  $\cW\subset\cV$ such that
$|\cW|/V>\pi\ls$ and $|\cVW|/V>\pi(S+\ell-1,S)$. 
Thus $\cW$ can get some candidates
$A_1,\dots,A_\ell$ elected, and $\cVW$ some $S+1-\ell$ others
$B_1,\dots,B_{S+1-\ell}$.
But then at least $S+1$ candidates will be elected, a contradiction.

The same argument shows \eqref{tx4b}.
\end{proof}

Note that \eqref{tx4b} implies both \eqref{tx1} and \eqref{tx2}, by taking 
all $\ell_i:=\ell$ and $m:=(S+1)/\ell$.

\begin{remark}
The \Theoremsx{} above hold for $\pix$ too.
\end{remark}

\begin{remark}\label{Rbest}
Recall that the threshold $\pi$ is the infimum that guarantees a good
outcome. Thus, a small $\pi$ is better (from the point of view of
proportionality and representability), and 
proportionality criteria are always of the form $\pi\ls\le p$ 
for some $p$ (possibly with strict inequality, see \refR{Rrefined}).
Conversely, a large $\pi$ is bad, and the extreme $\pi\ls=1$ means that we
need (almost) all votes in order to be sure to avoid a bad outcome.

Hence, the lower bounds in \refTxs{Tx2} and \ref{Tx4} give a bound on what
can be achieved by any election method, and they support the idea that 
$\ell/(S+1)$ is the optimum, and therefore perhaps desirable; 
\cf{} (in related terms) \refR{Rquota}.
(As shown in \refR{Rlower},  $\pi\ls<\ell/(S+1)$ is possible
for some $\ell$, but \refTx{Tx4} implies that then $\pi(\ell',S)>\ell'/(S+1)$
for $\ell'=S+1-\ell$, so it is impossible to have $\pi\ls<\ell/(S+1)$ for
all $\ell$.)
\end{remark}

\begin{problem}
  What is $\piinf\ls:=\inf\piS\MMM\ls$ for any given
$\ell$ and $S$ (with $1\le\ell\le S$),
  taking the infimum over all reasonable scenarios $\fS$ and election
  methods $\MMMx$?
What is the infimum over $\MMMx$ for a given $\fS$?
\end{problem}
The results above, together with simple examples,  
answer this problem in some cases.
Note that, for all $\ell$ and $S$, using \eqref{party-DH} below,
\begin{equation}
  \label{piinf1}
  \piinf\ls\le\pi\party\DH\ls =\frac{\ell}{S+1}.
\end{equation}
Hence, \refTx{Tx2} shows that equality holds in \eqref{piinf1} when
 $(S+1)/\ell$ is an integer.
Furthermore, if  $\ell\ge(S+1)/2$, then \refTx{Tx4} implies
$\pi\ls\ge1/2$, which is attained by $\pi\BV\party$ (\refT{TBV}), and thus
\begin{equation}
  \label{piinf2}
  \piinf\ls=\frac12, 
\qquad \ell\ge(S+1)/2.
\end{equation}

\begin{remark}\label{Rnot<}
Let $\ell<m$.
It might be argued that
if $\cW$ gets less than $\ell$ seats, then $\cW$ gets less than $m$, and 
therefore $\pi\ls\le\pi(m,S)$ ought to hold. However, this is not always
true, since our scenarios often include restrictions on the votes
that depend on $\ell$.
In fact,  \refT{TEJRBV} below shows
that $\pi\BV\EJR$ is a counterexample, see \refR{REJRBV} and 
\refE{EEJRBV}.
\end{remark}

Some proposed proportionality criteria state (for a specified scenario $\fS$)
that a set of voters that
form a fraction $p$ of all voters should get at least the same proportion of
seats, rounded down, i.e., at least $\floor{pS}$ seats. 
This simply means that $\piS(\ell,S)<\ell/S$; we state this formally for
easy reference.

\begin{theoremx}\label{Tx5}
  The property that a set of voters $\cW$ with a fraction $p=|\cW|/V$ of
  the votes always gets at least $\floor{pS}$ seats is equivalent to
  \begin{equation}\label{tx5-}
    \piS(\ell,S) \le \frac{\ell}{S}-, \qquad 1\le\ell\le S,
  \end{equation}
and to
  \begin{equation}\label{tx5<}
    \piS(\ell,S) < \frac{\ell}{S}, \qquad 1\le\ell< S.
  \end{equation}
\end{theoremx}

\begin{proof}
The property is equivalent to saying that for each integer $\ell\le S$, if
$|\cW|/V\ge \ell/S$, then $\cW$ will get at least $\ell$ seats.
This is trivial for $\ell=0$, and otherwise equivalent to \eqref{tx5-}.

Furthermore, \eqref{tx5-} is, in reasonable cases, trivial for $\ell=S$, and
otherwise equivalent to \eqref{tx5<}, see \refR{Rrefined}.
\end{proof}

\section{Party ballots}\label{Sparty}
Proportional election methods are used for political elections
in many countries,
in most (but not all) cases using a list method with
party lists. 
Thus, each ballot contains the name of one party, i.e., each voter votes for
a party, and each party is given a number of seats decided by the numbers of
votes for the parties.

\begin{definition}[Party ballots, $\pi\party$]\label{Dparty}
Suppose that all voters in $\cW$ vote for the same party $A$.
A good outcome is when at least $\ellx$ candidates from $A$ are elected.  
\end{definition}

$\pi\party$ is the same as the
\emph{threshold of exclusion} defined by \citet{RaeEtal} ($\ell=1$)
and \cite{Rae} ($\ell>1$), see also \citet{Lijphart}.
\smallskip

There are two main classes of proportional list methods: divisor methods 
and 
quota methods, 
see \refApp{Alist} for definitions.

For divisor methods, the threshold $\pi\party$
was found for \DHn's method by 
\citeauthor{RaeEtal} \cite{RaeEtal,Rae}, 
and for \StLn's method by \citet{Lijphart}.
General divisor methods
were treated by
\citet{Palomares}, with further results by \citet{JonesW}.
In particular, 
\cite{Palomares} showed the following explicit result
for linear divisor methods with divisors $d(n)=n-1+\gam$, $\gam\in\oi$; 
recall that these include the important D'Hondt ($\DHx$)
and \StLn{} ($\StLx$) methods, 
with $\gam=1$ and $\gam=1/2$, respectively, as well
as several less important methods.
This result is also implicit in \citet[Satz 6.2.3]{Kopfermann}.

\begin{theorem}[\citet{Palomares}; \citet{Kopfermann}]\label{Tdiv}
Let\/ $\Divx{\gam}$ be the divisor method with divisors $d(n)=n-1+\gam$ for some
$\gam\in\oi$. Then, for $1\le \ell\le S$,
\begin{equation}
  \label{party-div}
  \pi\party^{\Divx{\gam}}\ls=\frac{\ell-1+\gam}{\ell-1+\gam\xpar{S+2-\ell}}.
\end{equation}
In particular, for any $\gam\in\oi$,
\begin{equation}
  \label{party-div1}
  \pi\party^{\Divx{\gam}}(1,S)=\frac{1}{S+1}.
\end{equation}

As special cases \cite{RaeEtal,Rae,Lijphart}, we have, for $1\le \ell\le S$,
\begin{align}
  \label{party-DH}
  \pi\party\DH\ls&=\frac{\ell}{S+1},
\\
  \label{party-StL}
  \pi\party\StL\ls&=\frac{2\ell-1}{S+\ell}.
\end{align}
\end{theorem}

Some numerical values of $\pi\party\DH\ls$ and $\pi\party\StL\ls$
are given in \refTabs{tab:optimal} and \ref{tab:StL} in \refApp{Anum}.

\begin{proof}
  $\pi\party\ls$ equals $S_{\ell-1}$ in \cite{Palomares}. 
\end{proof}

\begin{example}
Another special case is Adams's method  $\Adamsx$, which is the case $\gam=0$,
\ie, $d(n)=n-1$.
In this case, for $\ell\ge2$,
\eqref{party-div} simply yields
\begin{equation}
  \label{party-adams}
  \pi\party\Adams\ls=1,
\qquad 2\le \ell\le S.
\end{equation}
For $\ell=1$, however, \eqref{party-div} yields the indeterminate $0/0$.
The reason is that Adams's method gives every party at least one seat.
Hence, if the number of parties is at most the number of seats, then the
threshold $\pi\is=0$ (and \eqref{party-adams} is obvious with $S$ parties),
but if there are more parties, then the method is ill-defined; for example,
if the seats are distributed by lot then $\pi\party\Adams(1,S)=1$, but if
they are distributed to the $S$ largest parties, then 
$\pi\party\Adams(1,S)=1/(S+1)$ (in accordance with \eqref{party-div1}).
\xfootnote{
The same applies to any method guaranteeing every party a seat, for example
\emph{Huntington--Hill's method} used for allocation in the US House of
Representatives.
}
\end{example}

\begin{remark}\label{Rjamkad}
\citet{Lijphart}
consider also the modified \StLn{} method, with 
$d(1)=0.7$ and $d(n)=n-0.5$ for $n\ge1$. (Or equivalently, and
traditionally,
$d(1)=1.4$ and $d(n)=2n-1$ for $n\ge1$.) Then they find
\begin{equation}
  \pi\party\ls=
  \begin{cases}
    \frac{1}{S+1}, & \ell=1,
\\
\frac{2\ell-1}{1.4S+0.6\ell+0.4}, & 2\le \ell\le S.
  \end{cases}
\end{equation}
For the modification with $d(1)=0.6$ instead
(or, equivalently, the sequence of divisors $1.2, 3, 5, \dots$), 
used in Sweden since 2018, 
similar calculations yield
\begin{equation}
  \pi\party\ls=
  \begin{cases}
    \frac{1}{S+1}, & \ell=1,
\\
\frac{2\ell-1}{1.2S+0.8\ell+0.2}, & 2\le \ell\le S.
  \end{cases}
\end{equation}
\end{remark}

\begin{remark}
\cite{Lijphart},  \cite{Palomares} and \cite{Kopfermann}
more generally calculate the thresholds when the number
  of parties is given. 
For simplicity, we do not consider that restriction in the present
  paper; we let the number of parties be arbitrary. (A fixed number of
  parties may be covered by our formalism too, if desired,
by including it in the  definition of the scenario $\fS$.)
\end{remark}

We next compute $\pi\party$ for quota methods, 
including \vkn{} and Droop's method,
at least assuming that
the quota is defined without rounding.
For \vkn, this was done by \citet{Lijphart}. 
Presumably, the general case too has been done previously, and the result is 
implicit in \citet[Satz 6.2.11]{Kopfermann}, 
see also \cite[Sats 8.8]{SJV6},
but we do not know any reference to an explicit formula. 
The general formula is a little involved, and for convenience we give three
equivalent versions of it.

\begin{theorem}[\citet{Kopfermann}]\label{Tquot}
  Let $\Qx\gd$ be the quota method with quota $V/(S+\gd)$, where $\gd\in\oi$.
Then, for $1\le \ell\le S$,
\begin{align}\notag
  \pi\party\Q\gd\ls &
= \frac{\ell(S+2-\ell)-1+\gd}{(S+\gd)(S+2-\ell)}
=\frac{\ell}{S+\gd}-\frac{1-\gd}{(S+\gd)(S+2-\ell)}
\\&
=\frac{\ell}{S+1}+\frac{(1-\gd)(\ell-1)(S+1-\ell)}{(S+\gd)(S+1)(S+2-\ell)}
\label{tquot}.
\end{align}
In particular, 
\begin{align}\label{tquot1}
  \pi\party\Q\gd(1,S)
=\frac{1}{S+1}.
\end{align}
As special cases, for $\gd=0$ \cite{Lijphart} and $\gd=1$, respectively,
\begin{align}\label{vk}
  \pi\party\vk\ls
&= \frac{\ell(S+2-\ell)-1}{S(S+2-\ell)}
=\frac{\ell}{S}-\frac{1}{S(S+2-\ell)},
\\
  \pi\party\Droop\ls
&
=\frac{\ell}{S+1}.
\end{align}
\end{theorem}

Some numerical values of $\pi\party\Droop\ls$ and $\pi\party\vk\ls$
are given in \refTabs{tab:optimal} and \ref{tab:LR} in \refApp{Anum}.

\begin{proof}
  Suppose that party $i$ gets $v_i$ votes and $s_i$ seats, with $\sum_i
  v_i=V$ and $\sum_i s_i=S$.
Let $x_i:=v_i/Q=\frac{v_i}{V}(S+\gd)$.
Thus
\begin{align}
  \label{axa}
\sum_i x_i 
=\frac{\sum_i v_i}{V}(S+\gd)
= S+\gd.
\end{align}
The quota method means that for some $t\in\oi$, see \eqref{qm}, 
\begin{align}
  \label{bxa}
s_i+t-1 \le x_i \le s_i+t.
\end{align}
In particular, for a party $i$ with $s_i>0$, 
\begin{align}
x_i\ge s_i+t-1\ge ts_i,  
\end{align}
and trivially $x_i\ge ts_i$ also if $s_i=0$.
Hence, using also \eqref{axa} and $x_1-s_1\le t$ from \eqref{bxa},
\begin{align}
  S+\gd 
&= x_1+\sum_{i\neq1} x_i
\ge x_1+\sum_{i\neq1}t s_i
=x_1+t(S-s_1)
\\&
\ge x_1+(x_1-s_1)(S-s_1)
=  x_1(S+1-s_1)-s_1(S-s_1).
\end{align}
Consequently,
\begin{align}\label{cac}
  x_1 \le \frac{S+\gd+s_1(S-s_1)}{S-s_1+1}
=s_1+1-\frac{1-\gd}{S-s_1+1}.
\end{align}
Furthermore, we may have equality in \eqref{cac}, for any given $S$ and
$s_1\le S$,
by taking 
\begin{align}
t=x_1-s_1 = \frac{S-s_1+\gd}{S-s_1+1}\in\oi  
\end{align}
and $S-s_1$ other parties with $x_i=t$ and $s_i=1$;
note that this satisfies \eqref{axa} and \eqref{bxa}.
We may then take $v_i:=Nx_i$ for a suitable integer $N$
in order to get integer numbers of votes. (Provided $\gd$ is rational;
otherwise we take $N$ large and approximate by rounding $v_i$ suitably.)

Taking $s_1=\ell-1$, and noting that the bound in \eqref{cac} is increasing in
$s_1$, we see that if \eqref{cac} does not hold, then party 1 has to
get at least $\ell$ seats, and that this is best possible.
Since the proportion of votes is $v_1/V=x_1/(S+\gd)$, cf.\ \eqref{axa}, 
it thus follows from \eqref{cac} that
\begin{align}\label{cacc}
\pi\Q\gd\party\ls
=\frac{1}{S+\gd}\Bigpar{\ell-\frac{1-\gd}{S-\ell+2}}.
\end{align}
This is equivalent to the three different forms in 
\eqref{tquot} by elementary algebraic manipulations.
\end{proof}

\begin{remark}\label{RQrounding}
  If the quota is defined by rounding $V/(S+\gd)$ up or
down to an integer, the result still holds for $\pix\party$, see \eqref{pix}.
\end{remark}

\begin{remark}\label{Rquota+}
  The assumption $\gd\in\oi$ means that the standard definition of the quota
  method works. The results extends (with the same proof) to any real
  $\gd\in(-S,1)$, \ie{} quota $Q\ge V/(S+1)$,
with the interpretation in \refApp{Aquota}; note that then
  still $t\le1$ in the proof.
We have not investigated the case $\gd>1$, \eg{} the Imperiali quota $V/(S+2)$.
\end{remark}


\begin{remark}
  Note that for $\ell=1$,
all linear divisor methods and all quotient methods in
  \refTs{Tdiv} and \ref{Tquot} yield equality in \eqref{tx1}.
\end{remark}

\begin{remark}
  Comparing \eqref{party-DH} and \eqref{party-StL}, we see that
  \begin{equation}
\pi\DH\party\ls\le \pi\StL\party\ls,     
\qquad\lele,
  \end{equation}
with strict inequality for $\ell>1$.
This can be interpreted as meaning that \DHn's method has better proportionality 
than \StLn, in a certain sense.
On the other hand, it is well-known that \StLn's method is (asymptotically)
unbiased, in an average sense, while \DHn's method has a bias in favour of
larger 
parties, see \eg{} \cite{SJ262}, so in this sense
\StLn's method is more proportional than \DHn.
(The same results hold when comparing \vkn{} and Droop's method,
see again \eg{} \cite{SJ262}.)

This just stresses that there is no single criterion that defines
``proportionality''. 
\end{remark}

\subsection{Unordered and ordered ballots}\label{SSpartylist}
We have here considered election methods with party ballots.
Also in an election using unordered or ordered ballots, it may happen that
the voters are so polarised that they do not use their freedom to mix
candidates; instead they all belong to different parties (organised or not),
with all voters in each party voting for the same party list. 
(These naturally being disjoint.)
The election then really becomes an election
with party ballots. 
We formally define, for use in later sections, 
a scenario for this special (but not unrealistic) case, 
using the same notation as in \refD{Dparty}.

\begin{definition}[Party lists for unordered or ordered ballots, $\pi\party$]
\label{Dparty2}
Suppose that all voters vote for disjoint party lists, and
that all voters in $\cW$ vote for the same list $\gs$, containing at least
$\ellx$ candidates.
A good outcome is when at least $\ellx$ candidates from $\gs$ are elected.
\end{definition}

\section{Unordered ballots: Block Vote, SNTV, Limited Vote, \dots}
\label{SBV}

In this section we begin the study of election methods using unordered ballots
by studying  some simple, classical election methods, \viz{} 
\BVn, \AVn, \SNTVn{} (SNTV), \LVn, \CVn.
The results below for them are not new; on the contrary, they were known
already in the 19th century.

We begin by defining scenarios for elections with unordered ballots.
For later use, we note that these definitions apply also to elections with
ordered ballots.
One interesting scenario that we consider
is the party list case in \refD{Dparty2}.
Another interesting scenario is the less restrictive assumption that the
voters in $\cW$ vote on a common list of candidates, while the other voters
may be less organized and we do not assume anyting about their votes.

\begin{definition}[Ordered or unordered ballots, $\pi\same$]\label{Dsame}
Suppose that all voters in $\cW$ vote for the same list of candidates $\cA$.
(In the ordered case, all vote on the same ordered list.)
The other voters may vote arbitrarily.
We assume $|\cA|\ge\ellx$.
A good outcome is when at least $\ellx$ candidates from $\cA$ are elected.
\end{definition}

It is obvious that for some election methods (\eg{} SNTV), it is sometimes a
bad strategy for a group of voters to vote on the same list, and that a
perty in order to be successful is forced to use more elaborate strategies.
We therefore consider also a scenario where we assume 
that the set of voters $\cW$ is organized in some sense, so that the voters in
$\cW$ vote as directed by a ``leadership'' or ``party organization''; 
thus making  tactical voting feasible.
We also assume that the leadership is intelligent  and uses the 
mathematically optimal strategy.
\xfootnote{
To find the optimal strategy might be a more or less difficult and perhaps
interesting problem in game theory. 
We assume that the size $|\cW|$ of $\cW$
as well as the total number of voters $V$ are known; otherwise further game
theoretical complications arise.
Note also that since the definition \eqref{pi} consider the worst cases, we may
assume that the strategy is known to an omniscient adversary.
}

\begin{definition}[Ordered or unordered ballots, $\pi\tactic$]\label{Dtactic}
Let a set $\cA$ of candidates be given, with $|\cA|\ge\ellx$,
and suppose that all voters in $\cW$ vote as instructed (by some intelligent
leader). 
A good outcome is when at least $\ellx$ candidates from $\cA$ are elected.
\end{definition}

Before considering specific election methods, we note two simple relations.

\begin{theorem}
  \label{Tps}
For any election method with unordered or ordered  ballots, and $1\le\ell\le S$,
\begin{align}
  \pi\party\ls &\le \pi\same\ls, \label{party-same}
\\
  \pi\tactic\ls &\le \pi\same\ls. \label{tactic-same}
\end{align}
\end{theorem}

\begin{proof}
  Each instance of the scenario $\partyx$ 
(\refD{Dparty2})
is also an instance of the  scenario  $\samex$ (\refD{Dsame}), 
and the definition of good outcome is the same.
Hence, a bad outcome for $\partyx$ is also a bad outcome for $\samex$,
and thus \eqref{pi} implies \eqref{party-same}, since the supremum is taken
over a larger set for $\pi\same$.

Next,
since $\pi\tactic$ assumes that $\cW$ uses an optimal strategy, a bad
outcome for $\pi\tactic$ with a given set $\cW$ means that $\cW$ has no
strategy that guarantees a good outcome. In particular, the strategy that
all voters in $\cW$ vote for the same list $\cA$ may fail, 
so there is a bad outcome for
$\pi\same$ too for this $\cW$. Hence, \eqref{pi} yields \eqref{tactic-same}.
\end{proof}

Consider first the simple and important 
\emph{\BVn} method $\BVx$ 
(\emph{multi-member plurality}; \emph{plurality-at-large});
recall that here every voter votes for (at most) $S$ candidates.
In this case, the simple strategy to vote on the same list is the best strategy.
As is well-known, this method gives a majority all seats, so the next result
is rather obvious.

\begin{theorem}\label{TBV} 
For \BVn\kol
  \begin{align}
\pi\party\BV\ls= \pi\same\BV\ls=\pi\tactic\BV\ls=\frac12,
\qquad\lele.
  \end{align}
\end{theorem}

\begin{proof}
  If $|\cW|>\frac12V$, so $|\cW|>|\cVW|$, then $\cW$ will get all $S$ seats by
  voting on the same list. Conversely, if $|\cW|<\frac12V$, then $\cW$ will
  not get any seat, regardless of how they vote, 
if all other voters vote on a common list (with at least $S$ names) 
disjoint from  $\cA$. 
This argument shows all three cases.
\end{proof}

\begin{remark}\label{Rdummy}
  For our purposes, it does not matter whether a voter has to vote for
  exactly $S$ candidates, or whether it only has to be at most $S$. This is
  seen from the proof above, but also because we (implicitly) assume that
  there is an unlimited number of potential candidates, so that a voter
  always has the option of throwing away some votes on dummy candidates with
  no hope of being elected. The same applies to Limited Vote below.
\end{remark}

\emph{\AVn} (\AVx) differs mathematically from \BVn{} only in that
each ballot may contain an arbitrary number of candidates. The proof above
applies to this method to, giving the following result.

\begin{theorem}\label{TAV} 
For \AVn\kol 
  \begin{align}
\pi\party\AV\ls=     \pi\same\AV\ls=\pi\tactic\AV\ls=\frac12,
\qquad\lele.
  \end{align}
\end{theorem}

Another related simple election method is \emph{\SNTVn} ($\SNTVx$), 
where each voter only
votes for one candidate. In this case,
if the objective is to get more than one candidate elected,  
it is obviously a stupid strategy to let all voters in $\cW$ vote in the same
way, so more elaborate strategies are required,
and $\pi\party$ and $\pi\same$ are not relevant.
The following result was essentially shown by \citet[p.~174]{Droop} (and may
have been known earlier);
it is a special case of \refT{TLV} below, shown by \citet{Dodgson}.
\begin{theorem}[\citet{Droop}]\label{TSNTV} 
For \SNTVn\kol 
  \begin{align}\label{tsntv}
      \pix\tactic\SNTV\ls=\frac{\ell}{S+1},
\qquad\lele.
  \end{align}
\end{theorem}

\begin{proof}
  Let $\cW$ use the strategy to divide its votes equally between $\ell$
  candidates $A_1,\dots,A_\ell$. Of course, an exactly equal split is possible
  only if $|\cW|$ is divisible by $\ell$, but $\cW$ can always assure that each
  candidate gets al least $\floor{\WWW/\ell}>\WWW/\ell-1$ votes. 

If the outcome is bad, then not all these $\ell$ candidates get elected; say
that $A_\ell$ does not. Then, at
least $S+1-\ell$ other candidates are elected. Each of these must have at least
as many votes as the unsuccessful $A_\ell$, and thus $>\WWW/\ell-1$ votes.
Hence, there are (at least) $S+1$ candidates with $>\WWW/\ell-1$ votes, and
consequently,
\begin{align}
  V > (S+1)\Bigpar{\frac{\WWW}{\ell}-1}
\end{align}
or
\begin{align}
\WWW <  \frac{\ell}{S+1}V + \ell
\end{align}
and
\begin{align}
\frac{\WWW}{V} <  \frac{\ell}{S+1} + \frac{\ell}{V}.
\end{align}
Hence, the definition \eqref{pix} yields $\pix\tactic\SNTV\ls\le \ell/(S+1)$.

Conversely, let $V$ be divisible by $S+1$ and write $w:=V/(S+1)$.
If $\WWW=\frac{\ell}{S+1}V=\ell w$, then $|\cVW|=V-\WWW=V-\ell w=(S+1-\ell)w$.
Hence, $\cVW$ may split their votes on $S+1-\ell$ candidates
$B_1,\dots,B_{S+1-\ell}$ with $w$ votes each. On the other hand,
for any strategy used by $\cW$, at least one of $A_1,\dots,A_\ell$ gets at most
$w$ votes. Hence a bad outcome is possible, 
and \eqref{pix} yields $\pix\tactic\SNTV\ls\ge \WWW/V=\ell/(S+1)$.
\end{proof}

  In \refT{TSNTV}, we use $\pix$ and not $\pi$, \ie, we consider the limit
  as $V\to\infty$. The reason is that the strategy used in the proof above, of
  splitting the votes equally between some candidates, in general can be
  followed only approximately due to divisibility issues. 
The following example shows that this really matters when $V$ is small.

\begin{example}\label{EpixSNTV}
  Let $\ell=2$, $S=3$, $\WWW=3$ and $V=5$.
Suppose that $\cW$ wants $A_1$ and $A_2$ to be elected.
Regardless of the strategy used by $\cW$, either $A_1$ or $A_2$ will get only
one vote, say $A_2$. Hence, if the two voters in $\cVW$ vote with one vote
each for $B_1$ and $B_2$, then there is a tie and it is possible that $B_1$
and $B_2$ are elected together with $A_1$, a bad outcome.
Consequently,
\begin{align}
  \pi\tactic\SNTV(2,3) \ge \frac{3}{5} > \frac12 = \pix\tactic\SNTV(2,3).
\end{align}
\end{example}

\begin{remark}
  \refT{TSNTV} shows that SNTV under ideal conditions can be regarded as
  a proportional method. However, 
as is well-known,
there are obvious practical problems with
  using  this strategy in, say, a general election. Moreover, the strategy
  assumes that the size of $\cW$ and of its opponents are known;
  misjudgement can lead to disastrous results.
The same applies to Limited Vote and \CVn{} below. 
(These and other problems were discussed already by \citet{Droop}.)
\end{remark}

More generally, \emph{Limited Vote} is a version of Block Vote, where each
ballot contains only (at most) $\lv$ candidates, for some fixed $\lv\le
S$; we use the notation $\LVx\lv$.
Thus, \BVn{} is the special case $\lv=S$, and \SNTVxn{} is the special case
$\lv=1$. 

\refTs{TBV} and \ref{TSNTV} generalize to Limited Vote as follows. Again we
consider $\pix$ for the same reason as for \SNTVxn.
This result was proved by Charles Dodgson
\xfootnote{Better known by his pseudonym Lewis Carroll used when writing
  fiction.}
\cite{Dodgson}, who also gave a table of
numerical values for all cases with $1\le S\le 6$.
(The result may have been known earlier; \citet{Droop} mentions a few
values in his discussions, so it is perhaps likely that he knew the general
formula.)

\begin{theorem}[\citet{Dodgson}]\label{TLV}
For \LVn:
Let\/ $1\le\lv\le S$. Then, for $1\le \ell\le S$,
\begin{align}
  \pix\tactic\LV{\lv}\ls 
&= \frac{\min\bigset{1,\xqfrac\lv{S+1-\ellx}}}
{\min\bigset{1,\xqfrac\lv{S+1-\ellx}}+\min\bigset{1,\xfrac\lv\ellx}}
\label{tlv1}\\&
= \frac{\ell\min\bigset{\lv,{S+1-\ellx}}}
{\ell\min\bigset{\lv,{S+1-\ellx}}+(S+1-\ell)\min\bigset{\lv,\ellx}}.
\label{tlv2}
\end{align}
  Hence, if\/ $1\le\lv\le (S+1)/2$, then
  \begin{align}\label{tlv3}
      \pix\tactic\LV{\lv}\ls 
=
\begin{cases}
\frac{\lv}{S+1+\lv-\ell},
  & 1\le \ell\le \lv,\\
\frac{\ell}{S+1},
  & \lv\le \ell\le S+1-\lv,\\
\frac{\ell}{\ell+\lv},
  & S+1-\lv\le \ell\le S,
\end{cases}
  \end{align}
and  if\/ $(S+1)/2\le\lv\le S$, then
  \begin{align}\label{tlv4}
      \pix\tactic\LV{\lv}\ls 
=
\begin{cases}
\frac{\lv}{S+1+\lv-\ell},
  & 1\le \ell\le S+1-\lv,\\
\frac12,
  &  S+1-\lv\le \ell\le \lv,\\
\frac{\ell}{L+\ell},
  & \lv\le \ell\le S,
\end{cases}
  \end{align}
\end{theorem}

\begin{proof}
  As for \SNTVxn, the strategy of $\cW$ is to divide its votes equally
  between its $\ell$ favoured candidates $A_1,\dots,A_\ell$.
If $\ell\le\lv$, then all voters in $\cW$ vote for the same list 
$\set{A_1,\dots,A_\ell}$. (If required to vote for exactly $\lv>\ell$ candidates,
they further vote for $\lv-\ell$ dummy candidates each; these have to come
from at least two different sets for different voters. See \refR{Rdummy}.)

On the other hand, if $\ell>\lv$, then the votes are split so that each
of the $\ell$ candidates gets $\cW \lv/\ell+O(1)$ votes. (E.g.{} by splitting the
votes as equally as possible between all $\lv$-subsets of
\set{A_1,\dots,A_\ell}, or, more practically, between 
the $\ell$ sets $\set{A_i,\dots,A_{i+\lv-1}}$, $ 1\le i\le \ell$,
with indices taken modulo $\ell$.)

This strategy thus gives each candidate $A_i$ at least
\begin{align}\label{lva}
v_A:=  \min\Bigpar{1,\frac{\lv}{\ell}}\WWW + O(1)
\end{align}
votes. Conversely, for any strategy, at least one of $A_1,\dots,A_\ell$ gets at
most $v_A$ votes. (Assuming that they get no votes from $\cVW$.)

Similarly, $\cVW$ can give each of $S-\ell+1$ candidates $B_1,\dots,B_{S-\ell+1}$
at least
\begin{align}\label{lvb}
v_B:=  \min\Bigpar{1,\frac{\lv}{S+1-\ell}}(V-\WWW) + O(1)
\end{align}
votes, but they cannot give all $S-\ell+1$ candidates more than $v_B$ votes.

It follows that the outcome is good if $v_A>v_B$, but bad, for any strategy of
$\cW$, if $v_A<v_B$. It follows that $\pix$ satisfies the equation
\begin{align}
\pix  \min\Bigpar{1,\frac{\lv}{\ell}}
=\bigpar{1-\pix}\min\Bigpar{1,\frac{\lv}{S+1-\ell}}.
\end{align}
This gives \eqref{tlv1}, and \eqref{tlv2}--\eqref{tlv3} follow.
\end{proof}

For $\pi\same\LV{L}$, only $\ell\le L$ are relevant.
In this case, the strategy in the proof  of \refT{TLV}
is to vote on the same list, and the proof (without $O(1)$)
shows also the following.

\begin{theorem}
\label{TRLV}
For \LVn:
  If\/ $1\le\ell\le L\le S$, then
  \begin{equation}
\pi\LV{\lv}\same\ls
=\pi\LV{\lv}\tactic\ls
=\pix\LV{\lv}\tactic\ls,
  \end{equation}
  given in \eqref{tlv1}--\eqref{tlv4}.
\qed
\end{theorem}

\begin{example} 
  \label{ELV1}
For $\ell=1$, \refTs{TLV} and \ref{TRLV} yield
  \begin{equation}
\pi\LV{\lv}\same\is=\pi\LV{\lv}\tactic\is
=\frac{L}{S+L}. 
  \end{equation}
Furthermore, taking $L=1$, the proof above applies also to
$\pi\party\is$;  thus,
  \begin{equation}\label{sntv1}
    \pi\SNTV\party\is=\pi\SNTV\same\is=\pi\SNTV\tactic\is
=\frac{1}{S+1}. 
  \end{equation}
\end{example}

It is easy to see that the result in \refT{TSNTV}
above for \SNTVxn{} applies also to \CVn{}
($\CVx$), see \refApp{ACV},
with the same proof. However, \CVn{} has the practical
advantage that the strategy of
splitting the vote equally between $\ell$ candidates may be easier to organize;
in particular, if the version of \CVn{} used allows each voter to
split the vote equally between $\ell$ candidates, then all voters in $\cW$ can
vote the same way. Hence we obtain the following result, again essentially
shown by \citet[p.~174]{Droop}. 

\begin{theorem}[\citet{Droop}]\label{TCV} 
For \CVn\kol
  \begin{align}
\pix\tactic\CV\ls = \pix\tactic\SNTV\ls=\frac{\ell}{S+1},
\qquad \lele.
  \end{align}
Furthermore,
in an ideal case where a voter may split the vote equally between $\ell$
candidates, 
  \begin{align}
\pi\tactic\CV\ls=\frac{\ell}{S+1}.
  \end{align}
\end{theorem}

The scenarios $\partyx$ and $\samex$ are not very interesting
for \CVn, since
they allow the voters in $\cW$ to spread their votes too thinly over too
many candidates. This is seen formally in the following theorem, 
where we consider the ideal Equal and Even version $\CVqx$, allowing a vote
to be split (equally) on an arbitrary number of candidates.
\xfootnote{
Also called \emph{Satisfaction Approval Voting (SAV)} \cite{BramsKilgour}.
}

\begin{theorem}\label{TCVqsame}
  For (ideal) Equal and Even \CVn\kol
  \begin{align}
    \pi\CVq\party\ls=\pi\CVq\same\ls=1,
\qquad \lele.
  \end{align}
\end{theorem}
\begin{proof}
Let $\cV:=\cW\cup\cU$, where $W:=\WWW\ge\ell$ and $\cU$ is disjoint from $\cW$
with $|\cU|=S$. Futhermore, let each voter be a candidate, let each voter in
$\cW$ vote for $\cW$, and let each voter in $\cVW=\cU$ vote only for
herself.
This is an instance of both $\partyx$ and $\samex$.
Furthermore, there is a tie between all candidates, and it is possible that the
outcome is
$\cE=\cU$, which is a bad outcome for both $\partyx$ and $\samex$.
Hence,
$\pi\party\CVq\ls\ge (W+S)/W$, for any $W\ge\ell$, and thus
$\pi\party\CVq\ls=1$, and similarly,  or by \eqref{party-same}, 
$\pi\same\CVq\ls=1$.
\end{proof}

\begin{remark}
  We might define a scenario $\sameqx$ to be as $\samex$, but with the further
  restriction that $|\cA|=\ell$. Then we would have
  $\pi\sameq\CVq\ls=l/(S+1)$.
We leave this to the reader to explore further.
\end{remark}

\begin{problem}
  Another possibility is to consider the version of \CVn{} where each voter
  may split her vote equally on at most $S$ candidates. 
What is   $\pi\same\CV\ls$ for this version? 
\end{problem}

Finally we note that the strategy for SNTV in \refT{TSNTV}, \ie, to split
$\cW$ into subsets voting for one candidate each, can be used (more or less
successfully) for any election method,
This yields the following general results, which will be used later.
(Note that the results include a rigorous version of \refTx{Tx4} for
$\pi\tactic$.) 
\begin{theorem}
  \label{Tsplit}
Let $\MMMx$ be any election method for ordered or unordered ballots.
\begin{romenumerate}
\item \label{Tsplit1}
For $\lele$,
  \begin{equation}\label{tsplit1}
    \pix\MMM\tactic\ls\le \ell\pix\MMM\tactic\is \le \ell\pi\MMM\tactic\is.
  \end{equation}
\item\label{Tsplit2} 
If\/ $\ell+m\le S$, then
  \begin{equation}\label{tsplit2}
    \pix\MMM\tactic(\ell+m,S)\le \pix\MMM\tactic\ls+\pix\MMM\tactic(m,S).
  \end{equation}
\item\label{Tsplitc} 
If\/ $\lele$, then
  \begin{equation}\label{tsplitc}
\pi\MMM\tactic\ls+\pi\MMM\tactic(S+1-\ell,S)
\ge
\pix\MMM\tactic\ls+\pix\MMM\tactic(S+1-\ell,S)\ge1.
  \end{equation}
\item\label{Tsplit=x}
If\/ $\pix\MMM\tactic\is\le 1/(S+1)$,
then
  \begin{equation}\label{tsplit=x}
    \pix\MMM\tactic\ls=\frac{\ell}{S+1},
\qquad\lele.
  \end{equation}
In particular, this holds 
if\/ $\pi\MMM\tactic\is\le 1/(S+1)$.
\item\label{Tsplit=}
If\/ $\pi\MMM\tactic\ls\le \ell/(S+1)$ for $\lele$,
then this is an equality:
  \begin{equation}\label{tsplit=}
    \pi\MMM\tactic\ls=\frac{\ell}{S+1},
\qquad\lele.
  \end{equation}
\end{romenumerate}
\end{theorem}

\begin{proof}
\pfitemref{Tsplit1}
The second inequality is \eqref{pixpi}.
The first inequality follows from \ref{Tsplit2} and induction, but we prefer
to give a direct proof.

  Fix $p>\pix\MMM\tactic\is$.
If $\WWW>p\ell V+\ell$, split $\cW$, as equally as possible, 
into $\ell$ subsets $\cW_i$ of sizes
\begin{equation}\label{jullan}
  |\cW_i|\ge\floor{\cW/\ell}>\cW/\ell-1>pV.
\end{equation}
Let $A_1,\dots,A_\ell\in\cA$ be $\ell$ desired candidates, and let
$\cW_i$
vote with the aim of electing $A_i$.
By \eqref{jullan} and our choice of $p$, if $V$ is
large enough, there is a strategy for $\cW_i$ that always will succeed
to get $A_i$ elected. Hence, if each $\cW_i$ uses such a strategy, then
$A_1,\dots,A_\ell$ are all elected, a good outcome for $\cW$.
Consequently,
\begin{equation}
  \pix\MMM\tactic\ls \le \limsup_{\Vtoo}\frac{p\ell V+\ell}{V}=p\ell,
\end{equation}
and \eqref{tsplit1} follows because $p$ is arbitrary with
$p>\pix\tactic\MMM\is$.

\pfitemx{\ref{Tsplit2},\ref{Tsplitc}}
For \ref{Tsplitc}, let $m:=S+1-\ell$. 
Note that the first inequality in \eqref{tsplitc} follows from \eqref{pixpi}.

Let 
$p:=\pix\MMM\tactic\ls+\pix\MMM\tactic(m,S)+2\eps$ with $\eps>0$.
If $\WWW>pV$, split $\cW$ into two sets $\cW_1$ and $\cW_2$ with
$|\cW_1|>\bigpar{\pix\MMM\tactic\ls+\eps}V-1$ 
and
$|\cW_2|>\bigpar{\pix\MMM\tactic(m,S)+\eps}V-1$. 
Then, if $V$ is large enough, given $\ell+m$ candidates
$A_1,\dots,A_{\ell+m}$,
there exists a strategy for $\cW_1$ to get $A_1,\dots,A_\ell$ elected,
and a strategy for $\cW_2$ to get $A_{\ell+1},\dots,A_{\ell+m}$ elected.
Combining these, we have a strategy that guarantees that
$A_1,\dots,A_{\ell+m}$ are elected.

In \ref{Tsplitc}, with $\ell+m=S+1$, this is impossible. Thus no  set
$\cW\subseteq\cV$ 
with $\WWW>pV$ can exist, and thus $p\ge1$, which yields \eqref{tsplitc}
since $\eps$ is arbitrary.

In \ref{Tsplit2}, we have shown that, for large $V$, we have a good outcome
(with $\ell+m$ elected) whenever $\WWW>pV$, and thus 
$\pix\MMM\tactic(\ell+m,S)\le p$. This yields \eqref{tsplit2}, since $\eps$
is arbitrary.

\pfitemref{Tsplit=x}
Part \ref{Tsplit1} yields the inequality
\begin{align}\label{kuku}
\pix\MMM\tactic\ls\le\frac{\ell}{S+1}.
\end{align}
Conversely, 
by \eqref{tsplitc} and \eqref{kuku} with $\ell$ replaced by $S+1-\ell$,
\begin{align}\label{koko}
\pix\MMM\tactic\ls\ge
1-\pix\MMM\tactic(S+1-\ell,S)
\ge
1-\frac{S+1-\ell}{S+1}
=
\frac{\ell}{S+1}.
\end{align}
The final sentence follows by \eqref{pixpi}.

\pfitemref{Tsplit=}
By the assumption, $\pi\MMM\tactic\is\le1/(S+1)$, and thus
\ref{Tsplit=x} applies and yields \eqref{tsplit=x}.
Hence, using the assumption again and \eqref{pixpi},
\begin{align}\label{kvkv}
\frac{\ell}{S+1}
=
\pix\MMM\tactic\ls
\le
\pi\MMM\tactic\ls
\le
\frac{\ell}{S+1},
\end{align}
showing \eqref{tsplit=}.
\end{proof}

\section{JR, PJR, EJR}\label{SJR}

We turn to properties and thresholds intended for situations without
organised parties, where a group of voters have similar opinions but do not
necessarily vote identically.

For election methods with unordered ballots, 
\citet{EJR} defined two properties \emph{\qJR{} (justified representation)}
and (stronger) \emph{\qEJR{} (extended justified representation)};
\citet{PJR2016,PJR2017} then defined a related property
\emph{\qPJR{} (proportional justified representation)}
such that \qEJR$\implies$\qPJR$\implies$\qJR.
Inspired by their definitions, we define more generally 
the corresponding thresholds.
Recall that $\cE$ denotes the set of elected candidates.

\begin{definition}[unordered ballots, $\pi\PJR$]\label{DPJR}
  Let $\cA$ be a set of at least $\ell$ candidates, and assume that every
  voter in $\cW$ votes for a set $\gs\supseteq\cA$, i.e.,
for all candidates in $\cA$ and possibly also for some others.
A good outcome is when at least $\ell$ candidates are elected that someone in
$\cW$ has voted for, \ie{} $\Bigabs{\cE\cap\bigcup_{\gs\in\cW}\gs}\ge\ellx$.
\end{definition}

\begin{definition}[unordered ballots, $\pi\EJR$]\label{DEJR}
  Let $\cA$ be a set of at least $\ell$ candidates, and assume that every
  voter in $\cW$ votes for a set $\gs\supseteq\cA$, i.e.,
for all candidates in $\cA$ and possibly also for some others.
A good outcome is when there exists a voter in $\cW$ that has voted for at
least $\ell$ candidates that are elected,
\ie{} $\bigabs{\cE\cap\gs}\ge\ellx$ for some $\gs\in\cW$.
\end{definition}

Note that the difference between \qPJR{} and \qEJR{}
disappears for $\ell=1$: 
\begin{align}\label{pjr-ejr1}
\pi\EJR(1,S)=\pi\PJR(1,S).
\end{align}

For larger $\ell$, 
the scenarios $\EJRx$ and $\PJRx$ have the same instances; the only
difference is the definition of a good outcome. 
This leads to the following inequality.
\begin{theorem} \label{Tpjr-ejr}
For any election method with unordered ballots\kol
  \begin{align}
\pi\PJR\ls\le \pi\EJR\ls, \qquad 1\le \ell\le S,
\label{pjr-ejr}
\end{align}
with equality when $\ell=1$.
\end{theorem}
\begin{proof}
A good outcome for the scenario $\EJRx$ is a good outcome for $\PJRx$ too.
Hence, an instance  with a bad outcome for $\PJRx$
is a bad outcome for $\EJRx$ too, 
and thus the supremum in \eqref{pi}
is taken over a larger set of instances for $\pi\EJR$; hence
\eqref{pjr-ejr} follows. The case $\ell=1$ is \eqref{pjr-ejr1}.
\end{proof}

We also have simple relations with the thresholds defined in earlier for
scenarios with more organized voters.

\begin{theorem} \label{Tsame-pjr}
For any election method with unordered ballots\kol
  \begin{align}
\pi\party\ls
\le\pi\same\ls
\le\pi\PJR\ls\le \pi\EJR\ls, 
\qquad 1\le \ell\le S.
\label{same-pjr}
\end{align}
\end{theorem}
\begin{proof}
  An instance of the scenario $\samex$ is also an instance of the
  scenario  $\PJRx$, and the outcome is good for one scenario if it is
  for the other.
Hence, an instance with a bad outcome for $\pi\same$ is also a bad outcome
for $\pi\PJR$, and thus \eqref{pi} implies $\pi\same\ls\le\pi\PJR\ls$.
The other inequalities in \eqref{same-pjr} are repeated from \refTs{Tps} and
\ref{Tpjr-ejr}. 
\end{proof}

\begin{remark}
The conditions \qEJR{} and \qPJR{} defined in \cite{EJR} and
  \cite{PJR2016,PJR2017} require a good outcome, 
in the sense of our definitions above,
  for any $\cW$ with $|\cW|\ge \ell V/S$, for any $\ell\le S$.
The condition \qJR{} \cite{EJR} is the special case  $\ell=1$ of both.
Consequently, for any reasonable election method, 
see \refR{Rrefined} and \refTx{Tx5},
  \begin{align}
\text{PJR} &\iff    \pi\PJR\ls< \frac{\ell}{S}, 
\qquad 1\le\ell< S,\label{PJR}
\\
\text{EJR} &\iff    \pi\EJR\ls< \frac{\ell}{S}, 
\qquad 1\le\ell< S,\label{EJR}
\\
\text{JR} &\iff    \pi\PJR\is< \frac{1}{S},
\qquad S>1.
\label{JR}
  \end{align}
\end{remark}

Let us first consider the classical election method discussed in \refS{SBV}.
It is obvious that the thresholds above are relevant for SNTV only when
$\ell=1$. In this case we have the following. 
(In particular, SNTV satisfies JR by \eqref{JR}.)
\begin{theorem}\label{TJRSNTV}
For \SNTVn\kol
    \begin{equation}
    \pi\SNTV\PJR\is=
    \pi\SNTV\EJR\is=
\pi\SNTV\same\is
=\frac{1}{S+1},
\qquad S\ge1. 
  \end{equation}
\end{theorem}
\begin{proof}
  For SNTV and $\ell=1$, 
the scenarios $\PJRx$ and $\EJRx$ require that all voters in $\cW$ vote
  on the same candidate $A$, and a good outcome is when $A$ is elected. 
Thus, $\pi\SNTV\PJR\is=\pi\EJR\SNTV\is=\pi\SNTV\same\is$, and the result
follows from \eqref{sntv1}.
\end{proof}

For BV{} and AV,
the results in \refTs{TBV} and \ref{TAV} extend to $\pi\PJR$,
but not to $\pi\EJR$.
The same holds for LV($L$), but in this case, similarly to SNTV above,  
only the case $\ell\le\lv$ is relevant.
We consider $\pi\PJR$ first, and begin with a separate treatment of BV and AV, 
although this result also can be obtained as a corollary of the more
complicated result for LV in \refT{TPJRLV}.

\begin{theorem}
  \label{TPJRBV}
For \BVn{} and \AVn\kol 
  \begin{align}\label{tpjrbv}
    \pi\PJR\BV\ls 
=     \pi\PJR\AV\ls 
&= \frac12,
\qquad 1\le \ell\le S.
  \end{align}
\end{theorem}

\begin{proof}
  The lower bound follows by \refTs{TBV}, \ref{TAV} and \ref{Tsame-pjr}.

To bound $\pi\PJR\BV\ls$ from above,
suppose that $|\cW|>\frac12V$.
Then the candidates in $\cA$ have at least $|\cW|$ votes each, while the
candidates not in $\bigcup_{\gs\in\cW}\gs$ have at most $|\cVW|=V-|\cW|<\WWW$
votes each. Hence either no candidate outside $\bigcup_{\gs\in\cW}\gs$ is elected,
or all candidates in $\cA$ are; in both cases the outcome is good.
Consequently, $\pi\PJR\BV\ls\le \xfrac12$.
The same argument applies to AV.
\end{proof}

\begin{remark}\label{RAVtie}
\citet{EJR} showed that
\AVn{} does not satisfy JR when $S\ge3$, and that for $S=2$, the answer
depends on the tie-breaking rule. 
By \eqref{JR}, the negative result for $S\ge3$ is a consequence of
\eqref{tpjrbv}. 
For $S=2$, we interpret their result using our refined notation in
\refR{Rrefined}: it is easy to see that 
$\pi\PJR\AV(1,2)=\frac12+$ (JR does not hold) with standard random tie-breaking,
but $\pi\PJR\AV(1,2)=\frac12-$ (JR holds) if we use a tie-breaking rule that 
in the case of several candidates with exactly $V/2$ votes each,
gives preference to a pair of candidates with disjoint
voter support before a pair of candidates supported by the same voters.
\end{remark}

\begin{theorem}
  \label{TPJRLV}
For Limited Vote:
If\/ $1\le\ell\le L$, then
  \begin{align}
    \pi\PJR\LV{L}\ls
=     \pi\same\LV{L}\ls
=\pi\tactic\LV{L}\ls
=\pix\tactic\LV{L}\ls,
  \end{align}
given in \refT{TLV}.
\end{theorem}

\begin{proof}
Consider a bad outcome for a set of voters $\cW$ as in \refD{DPJR}.
Let $\cAx:=\bigcup_{\gs\in\cW}\gs\supseteq\cA$, the set of candidates voted
for by some voter in $\cW$. Since the outcome is bad,
$|\cAx\cap\cE|\le\ell-1$, and thus there are at least $S-\ell+1$ candidates
$B_1,\dots,B_{S+1-\ell}$ not in $\cAx$ that are elected.
For a candidate $C$, let $v(C)$ be her number of votes, and let
$\vx:=\min_j v(B_j)$.
Thus $v(B_j)\ge\vx$ for $1\le j\le S+1-\ell$.
Furthermore, any candidate $C$ with $v(C)>\vx$ is elected; in particular,
there are at most $\ell-1$ such candidates in $\cAx$.

If we modify the election by eliminating all votes on any candidate in
$\cAx\setminus\cA$, then each voter in $\cW$ votes for $\cA$, so the new
election is an instance of $\samex$.
Furthermore, we still have $v(B_j)\ge\vx$ for $S+1-\ell$ candidates
$B_j\notin\cA$, and $v(A_i)>\vx$ for at most $\ell-1$ candidates
$A_i\in\cA$.
Hence, even if there are ties, a possible outcome is that
$B_1,\dots,B_{S+1-\ell}$ are elected, and thus at most $\ell-1$ from $\cA$,
a bad outcome for $\samex$.

We have shown that for every bad outcome for $\PJRx$, there is an election
with the same $|\cW|$ and $V$ and a bad outcome for $\samex$.
Hence, \eqref{pi} implies $\pi\PJR\LV{L}\ls\le\pi\same\LV{L}\ls$.
\refT{Tsame-pjr} provides the opposite inequality, and thus equality holds.
The proof is completed by \refT{TRLV}.
\end{proof}

The results for EJR are more complicated. We state the results as three
separate theorems, but prove them together.

\begin{theorem}\label{TEJRBV}
For \BVn\kol
  \begin{align}\label{tejrbv}
\pi\EJR\BV\ls
& = 
\begin{cases}
\frac{S}{2S+1-\ell},  
& 1\le\ell\le(S+1)/2,
\\[3pt]
 \frac{2S+1-2\ell}{3S+2-3\ell},
& (S+1)/2\le\ell\le S.
\end{cases}
  \end{align}
\end{theorem}

\begin{theorem}\label{TEJRAV}
For \AVn\kol
  \begin{align}
\pi\EJR\AV\ls
& = \frac{S}{2S+1-\ell},
\qquad 1\le \ell\le S.
  \end{align}
\end{theorem}

\begin{theorem}
  \label{TEJRLV}
For \LVn:
Let $1\le L\le S$. Then, for $1\le\ell\le L$,
  \begin{align}
\pi\EJR\LV\lv\ls
= 
\begin{cases}
\frac{L}{S+L+1-\ell},  
& 1\le\ell\le(S+1)/2,
\\[3pt]
 \frac{S+L+1-2\ell}{2S+L+2-3\ell},
& (S+1)/2\le\ell\le L.
\end{cases}
  \end{align}
\end{theorem}

Some numerical values of $\pi\EJR\BV\ls$ and $\pi\EJR\AV\ls$
are given in \refTabs{tab:BVEJR} and \ref{tab:AVEJR} in \refApp{Anum}.

\begin{proof}[Proof of \refTs{TEJRBV}--\ref{TEJRLV}]
 Take $L:=S$ for BV, and $L:=\infty$ for AV.
Thus we can treat BV, AV and LV($L$) 
together.

Consider a bad outcome, and let 
$a:=\WWW/V$ and
$k:=|\cE\cap\cA|$, the number of elected
  from $\cA$.
Note that $k<\ell$, since otherwise the outcome is good, as witnessed by any
$\gs\in\cW$.
Similarly,
no voter in $\cW$ votes for more than $\ell-k-1$ candidates in
$\cE\setminus\cA$.
Furthermore, no voter votes for more than $L$ candidates, so a voter in
$\cW$ votes for at most $L-\ell$ candidates outside $\cA$.
Thus, let
\begin{equation}\label{mk}
  m(k):=(\ell-k-1)\land (L-\ell).
\end{equation}
Then no voter in $\cW$ votes for more than $m(k)$ 
candidates in $\cE\setminus\cA$.
Furthermore, no voter in $\cVW$ votes for more than
\begin{equation}\label{m'k}
  m'(k):=|\cE\setminus\cA|\land\lv = (S-k)\land \lv
\end{equation}
candidates in $\cE\setminus\cA$.

Hence, if $v$ is the total number of votes for the $S-k$ candidates in
$\cE\setminus\cA$, then
\begin{align}\label{alseda}
v\le
  |\cW|m(k) + |\cVW|m'(k)
=m(k)aV+m'(k)(1-a)V.
\end{align}
On the other hand, there is at least one non-elected candidate in $\cA$;
she has at least $\WWW=aV$ votes, and thus each elected candidate has at
least $aV$ votes, so $v \ge (S-k)aV$.
Consequently, using \eqref{alseda},
\begin{align}\label{alg0}
 ( S-k)a \le v/V 
\le
m(k)a+m'(k)(1-a)
\end{align}
and thus
\begin{equation}\label{alga}
\frac{1}{a}-1=\frac{1-a}{a} \ge \frac{S-k-m(k)}{m'(k)}  .
\end{equation}
Let $k_1:=(2\ell-L-1)\lor0$, and $k_2:=S-L$.
Then \eqref{mk}--\eqref{m'k} yield
\begin{align}
      m(k)&=
  \begin{cases}
L-\ell, & 0\le k< k_1,
\\
\ell-k-1, & k\ge k_1,
  \end{cases}
\label{algb}
\intertext{and}
    m'(k)&=
  \begin{cases}
L, & k\le k_2,
\\
S-k, & k\ge k_2.
  \end{cases}
\label{algc}
\end{align}
Using \eqref{algb}--\eqref{algc}, 
it is easily verified that, regardless of the value of $k_2$, the \rhs{} of
\eqref{alga} is, as a function of $k$, (weakly) decreasing on $[0,k_1]$ and
(weakly) increasing on $[k_1,\infty)$; hence it has a minimum at $k=k_1$ and
\eqref{alga} implies, using $m(k_1)=\ell-k_1-1$ from \eqref{algb},
\begin{align}\label{algd}
  \frac{1}{a}\ge 1 + \frac{S-k_1-m(k_1)}{m'(k_1)}
=\frac{S+1-\ell+m'(k_1)}{m'(k_1)}
\end{align}
and thus
\begin{align}\label{abv}
a\le
\frac{m'(k_1)}{S+1-\ell+m'(k_1)}.
\end{align}

Conversely, we construct an example with equality in \eqref{abv}.
Suppose that $a$ is such that equality holds in \eqref{abv}, and thus in
\eqref{algd}. 
 Let $\cW$ be a set of voters with $\WWW=aV$, for a suitable $V$ 
(so that $aV$ and other numbers in the construction are integers),
let $k:=k_1$ and let $\cC_1,\cC_2,\cC_3$ be disjoint sets of candidates with
$|\cC_1|=k$, $|\cC_2|=\ell-k$, $|\cC_3|=S-k$.
Let each voter in $\cW$ vote on $\cA:=\cC_1\cup\cC_2$, and on $m(k)$
candidates from $\cC_3$, in an organised way such that each candidate in
$\cC_3$ gets equally many votes from $\cW$, \viz{} $\WWW m(k)/(S-k)$;
furthermore, let each voter in $\cVW$ vote for $m'(k)$ candidates from
$\cC_3$, again with evenly spead votes.
Then there is equality in \eqref{alseda} and
\eqref{alga},
 and thus in \eqref{alg0}, which implies that each candidate in $\cC_3$ gets
 exactly $aV$ votes. Thus all candidates in $\cC_1\cup\cC_2\cup\cC_3$ tie
 with $aV=\WWW$ votes each, and a possible outcome is $\cE=\cC_1\cup\cC_3$,
which is a bad outcome for $\pi\EJR$.
Consequently, \eqref{abv} is best possible, \ie,
\begin{align}\label{abw}
\pi\EJR\LV{\lv}\ls=
\frac{m'(k_1)}{S+1-\ell+m'(k_1)}.
\end{align}
The result follows from \eqref{abw} and \eqref{algb}--\eqref{algc}.
\end{proof}

\begin{remark}\label{REJRBV}
  Note that $\pi\EJR\BV\ls$ is not monotone in $\ell$, \cf{} \refR{Rnot<}.
In fact, \eqref{tejrbv} shows that $\pi\BV\EJR\ls$, as a function of $\ell$,
increases from $\pi\EJR\BV\is=\frac12$ to a maximum at
$\ell=\ceil{(S+1)/2}$, and then decreases back to 
$\pi\EJR\BV(S,S)=\frac12$. The maximum value is $\approx \frac{2}3$ for
large $S$.

The behaviour of $\pi\LV{L}\PJR\ls$ is similar when $L>\ceil{(S+1)/2}$.
\end{remark}

\begin{example}\label{EEJRBV}
  For $S=3$,
$\pi\EJR\BV(1,3)=\frac12$,
$\pi\EJR\BV(2,3)=\frac{3}{5}$,
$\pi\EJR\BV(3,3)=\frac12$.
\end{example}

\begin{remark}
  As noted above, $\pi\BV\EJR(S,S)=\frac12=\pi\same\BV(S,S)$.
More generally, 
\begin{align}
   \pi\LV{\lv}\PJR\Ls=\pi\EJR\LV\lv\Ls=\pi\LV\lv\same\Ls,
\end{align}
which follows directly from \refDs{DPJR} and \ref{DEJR}, 
since when $\ell=L$, the
assumptions require every voter in $\cW$ to vote for $\cA$.
\end{remark}

For \CVn, the scenarios above are not very interesting, for the same reason
as $\samex$;
they allow the voters in $\cW$ to spread their votes over too
many candidates. 
We state this for the ideal Equal and Even version $\CVqx$, 
as in \refT{TCVqsame}.

\begin{theorem}\label{TCVqJR}
  For (ideal) Equal and Even \CVn\kol
  \begin{align}
    \pi\CVq\PJR\ls=\pi\CVq\EJR\ls=1,
\qquad \lele.
  \end{align}
\end{theorem}
\begin{proof}
An immediate consequence of \refTs{TCVqsame} and \ref{Tsame-pjr}.
\end{proof}

\begin{remark}\label{RSJR}
In \refD{DEJR}, it is not required that $\ell$ candidates
from the common set $\bigcap_{\gs\in\cW}\gs\supseteq\cA$ are elected. 
That requirement for $\ell=1$ is used in the definition of
\emph{Strong Justified Representation} in \cite{EJR};
similarly, the  weaker version that $|\gs\cap\cE|\ge\ell$ for every
$\gs\in\cW$
is used in the definition of
\emph{Semi-Strong Justified Representation}.
However, as remarked in \cite{EJR}, these conditions are too strong,
and they cannot be required in general.
In our setting,
we might define  scenarios $\SJRx$ and $\SSJRx$ as in \refD{DEJR} but
using these requirements;
however, then
$\pi\SJR\MMM\ls=\pi\SSJR\MMM\ls=1$ for every election method $\MMMx$ and
$\lele$. 
To see this, let
$n\ge1$ and consider
an election with $S+1$ candidates $C_{1},\dots,C_{S+1}$
and
$V=n+S+1$
votes: $n$ votes on all candidates and 1 vote on 
$\cA_i:=\set{C_{i},\dots,C_{i+\ell-1}}$ (with indices mod $S+1$)
for each $i\in[S+1]$.
Let $\cW_i:=\set{\gs:\gs\supseteq\cA_i}$.
If, for $\SJRx$ or $\SSJRx$, 
the outcome is good for $\cW_i$, then $\cE\supseteq\cA_i$.
This cannot hold for all $i$, and thus 
 the outcome is bad for at least one set of voters $\cW_i$.
Hence, 
\begin{equation}
\pi\SJR\MMM\ls, \pi\SSJR\MMM\ls\ge\frac{|\cW_i|}{V}=\frac{n+1}{n+S+1},
\end{equation}
and the result follows since $n$ is arbitrary.
We therefore do not consider these scenarios.
\end{remark}

\section{\phragmen's and Thiele's  unordered methods}
\label{SPhrThu}
In this section, we continue the study of 
the thresholds 
$ \pi\party,
\pi\tactic,\allowbreak
\pi\same,\allowbreak
 \pi\PJR,
 \pi\EJR
$
for unordered ballots; we now consider
\phragmen's and Thiele's election methods, defined in
\refApps{APhru}--\ref{AThe}.

\begin{problem}
  Further election methods for unordered ballots are described and studied in
  \eg{} \cite{Kilgour,EJR,PJR2017,SJ321,FSST}.
Study   
$ \pi\party,
\pi\tactic,
\pi\same,
 \pi\PJR,
 \pi\EJR
$
for them!
\xfootnote{
Some inequalities for $\pi\PJR$ and $\pi\EJR$
follow by \eqref{PJR}--\eqref{JR} from results in 
\cite{EJR,PJR2017,SJ321} showing whether or not 
certain methods satisfy JR, PJR or EJR.
}
\end{problem}

\subsection{The party list case}
We begin with a simple result, 
presumably known already to \phragmen{} and Thiele.

\begin{theorem}
  \label{TPhruparty}
For \phragmen's and Thiele's 
unordered methods\kol
\begin{align}\label{tPhruparty}
\pi\party\Phru\ls
&=\pi\party\Thopt\ls
=\pi\party\Tha\ls
=\pi\party\The\ls
\notag\\&
=\pi\party\DH\ls
=\frac{\ell}{S+1},
\qquad\lele.
\end{align}
\end{theorem}
\begin{proof}
  It is easy to see that in the party list case, all
four methods reduce to D'Hondt's method 
\cite[Theorem 11.1]{SJV9}, \cite{Phragmen1895}, \cite{Thiele}. 
Hence, the result follows from \eqref{party-DH}.
\end{proof}

\subsection{\phragmen's method}\label{SPhru}

For \phragmen's unordered  method, we have a simple result.
The result for $\pi\same\Phru$ is shown in \cite{SJV9}, and its
extension to $\pi\PJR\Phru$ in \cite{SJ321}.
\begin{theorem}[\cite{SJ321}]  \label{TsamePhru}
For \phragmen's unordered methods\kol
\begin{equation}
\pi\tactic\Phru\ls
=\pi\same\Phru\ls
=\pi\PJR\Phru\ls
=\frac{\ell}{S+1},
\qquad\lele.
\end{equation}
\end{theorem}

\begin{proof}
  Consider first $\PJRx$.
Let $\cW$ be a set of 
voters and $\cA$ a set of candidates as in \refD{DPJR},
and suppose that the outcome is bad.
Let $\cAx:=\bigcup_{\gs\in\cW}\gs\supseteq\cA$.
Thus each voter in $\cW$ votes for a set $\gs$ with
$\cA\subseteq\gs\subseteq\cAx$, where $|\cA|\ge\ell$,
but $k:=|\cAx\cap\cE|<\ell$.
In particular, at least one candidate in $\cA$ is not elected.

We use the formulation with loads
in  \refApp{APhru},
and let $t=t\xx \MM$ be the final
maximum load of a ballot.
Say that a ballot with load $u$ has \emph{free voting power} $t-u$; this is the
additional load that the ballot may accept without raising the maximum load $t$.

The $k$ elected candidates in $\cAx\cap\cE$
together give load $k$, 
of which some part may fall on
voters not in $\cW$. 
Thus, if ballot $i$ has final load $x_i$,
\begin{align}\label{aga}
  \sum_{i\in\cW}x_i \le k\le \ell-1.
\end{align}
Furthermore, the total free voting power of the ballots in
$\cW$
is at most 1, since otherwise another candidate from $\cA\setminus\cE$ 
would have been elected to the last place, with a smaller $t\xx\MM$. 
(This total free voting power may equal 1 if there is a
tie for the last seat.)
Thus, 
\begin{align}\label{agb}
  \sum_{i\in\cW}(t-x_i) \le 1.
\end{align}
Consequently, combining \eqref{aga} and \eqref{agb},
\begin{align}\label{agc}
\WWW t=\sum_{i\in\cW}x_i+  \sum_{i\in\cW}(t-x_i) \le \ell.
\end{align}

On the other hand, $\MM-k$ candidates not in $\cAx$ have been
elected, incurring a total load $\MM-k\ge \MM+1-\ell$.
Since the ballots in $\cW$ do not get any of this load,
all of it falls on the $V-|\cW|$ other ballots, and thus
\begin{equation}\label{agd}
  \MM+1-\ell \le (V-|\cW|)t.
\end{equation}
Combining \eqref{agc} and \eqref{agd} we find
\begin{align}
  \frac{\WWW}{V-\WWW} =   \frac{\WWW t}{(V-\WWW)t}
\le \frac{\ell}{\MM+1-\ell},
\end{align}
which is equivalent to 
\begin{align}\label{agh}
  \frac{\WWW}{V}\le\frac{\ell}{\MM+1}.
\end{align}
Thus,
\begin{align}\label{age}
  \pi\PJR\Phru\ls\le\frac{\ell}{S+1}.
\end{align}

\refTs{Tps} and \ref{Tsame-pjr} together with \eqref{age} yield
\begin{align}\label{agf}
  \pi\tactic\Phru\ls
\le  \pi\same\Phru\ls
\le  \pi\PJR\Phru\ls
\le\frac{\ell}{S+1}.
\end{align}
Hence,
\refT{Tsplit}\ref{Tsplit=} applies and yields
\begin{align}\label{agge}
\pi\tactic\Phru\ls
= \frac{\ell}{S+1}.
\end{align}
The result follows by \eqref{agf} and \eqref{agge}.
\end{proof}

However, it was shown in \cite[Example 5]{SJ321} that \phragmen's unordered
method does 
not satisfy EJR, \ie,
$\pi\Phru\EJR\ge \ell/S>\ell/(S+1)$ is possible.
We may tweak that example a little to the following (found by computer
experiment rather than an analysis).

\begin{example}[based on \cite{SJ321}]
  \label{EBrill5XX}
Consider an election by \phragmen's unordered method with 
$\MM=12$ seats, 
candidates $\cC=\set{A,B,\allowbreak C_1,\dots,C_{12}}$, and
2409 voters voting
\begin{val}
\item [200] \setx{A,B,C_1}
\item [209] \setx{A,B,C_2}
\item [600] \setx{C_1,C_2,C_3,\dots,C_{12}}
\item [500] \setx{C_2,C_3,\dots,C_{12}}
\item [900] \setx{C_3,\dots,C_{12}}
\end{val}
Then, a computer assisted calculation shows that the elected are, in order,
$C_3, C_4, C_5, C_6, C_2, C_7, C_8, C_1, C_9, C_{10}, C_{11}, C_{12}$
(with $C_3,\dots,C_{12}$ tying throughout; here one possibility is chosen).
Hence, if $\cW$ is the set of the 409 voters on the first two lines, 
the conditions of $\EJRx$ are satisfied for $\ell=2$ with $\cA=\set{A,B}$,
but the outcome is bad. Consequently,
\begin{align}
  \pi\EJR\Phru(2,12)\ge \frac{409}{2409}>\frac{2}{12}
=\frac{\ell}{S}>\frac{\ell}{S+1}.
\end{align}
(This is not sharp, and can be improved at least a little.)
\end{example}

\begin{problem}
  What is $\pi\Phru\EJR(\ell,S)$ in general?
Even the case $\ell=2$ and a given (small) $S$ seems far from trivial.
\end{problem}

\subsection{Thiele's \opt{} method}\label{SThopt}

\citet{EJR} showed that Thiele's \opt{} method satisfies EJR;
\ie, $\pi\EJR\Thopt\ls<\ell/S$ by \eqref{EJR}.
In fact, the proof can be improved to show the optimal bound $\ell/(S+1)$.

\begin{theorem}[\cite{EJR}, improved]
  \label{TEJRTh}
For \Thoptn\kol
For $\lele$,
\begin{equation}
\pi\tactic\Thopt\ls
=\pi\same\Thopt\ls
=\pi\PJR\Thopt\ls
=\pi\EJR\Thopt\ls
=\frac{\ell}{S+1}
.
\end{equation}
\end{theorem}

\begin{proof}
Consider an election with a bad outcome $\cE$ for the scenario $\EJRx$ in
\refD{DEJR}. 
Then not every candidate in $\cA$ is elected, since otherwise the outcome
would be good, witnessed by any $\gs\in\cW$.
Fix some $A\in\cA\setminus\cE$.

For every ballot $\gs$ and every elected candidate $C\in\cE$, let
$\gd(\gs,C)$ be the change in the satisfaction of $\gs$ 
(see \refApp{AThopt})
if $A$ is elected instead of $C$,
i.e., 
$\cE$ is replaced by $\cE\cup\set{A}\setminus\set{C}$.
Let  further $\gD(\gs):=\sumCE\gd(\gs,C)$.

Since $\cE$ maximizes the total satisfaction,
\begin{equation}
  \sum_{\gs\in\cV}\gd(\gs,C)\le0
\end{equation}
for every $C\in\cE$, and thus
\begin{equation}\label{ua}
\sum_{\gs\in\cV}\gD(\gs)=\sumCE  \sum_{\gs\in\cV}\gd(\gs,C)\le0.
\end{equation}

Consider a ballot $\gs$, and let $k:=|\gs\cap\cE|$.
Then $\gd(\gs,C)\ge -1/k$ if $C\in\gs\cap\cE$, and
$\gd(\gs,C)\ge 0$ if $C\notin\gs\cap\cE$. (In both cases with equality if
$A\notin\gs$.)
Hence,
\begin{equation}\label{ub}
  \gD(\gs)\ge k\cdot\Bigpar{-\frac{1}k}=-1.
\end{equation}
Moreover, if $\gs\in\cW$, then $A\in\gs$.
Hence, $\gd(\gs,C)=0$ if $C\in\gs\cap\cE$, and
$\gd(\gs,C)=1/(k+1)$ if $C\notin\gs\cap\cE$.
Furthermore, $k\le\ell-1$, since otherwise
the outcome would be good.
Hence,
\begin{equation}\label{uc}
  \gD(\gs)=\frac{S-k}{k+1}\ge \frac{S+1-\ell}{\ell},
\qquad \gs\in\cW.
\end{equation}
Consequently, using \eqref{uc} for $\gs\in\cW$ and \eqref{ub} for
$\gs\in\cVW$,
\begin{equation}\label{ud}
  \sum_{\gs\in\cV}\gD(\gs)
\ge |\cW|\frac{S+1-\ell}{\ell}-|\cVW|
= \WWW\frac{S+1}{\ell}-V.
\end{equation}
Combining \eqref{ua} and \eqref{ud} we find
\begin{equation}
0\ge \WWW\frac{S+1}{\ell}-V
\end{equation}
and thus
\begin{equation}
  \frac{\WWW}{V}\le\frac{\ell}{S+1}.
\end{equation}
This holds for every bad outcome, and thus
\begin{equation}\label{ue}
\pi\EJR\Thopt\ls\le\frac{\ell}{S+1}.
\end{equation}
Combining \eqref{ue} with \refTs{Tsame-pjr} and \ref{TPhruparty}
yields
\begin{equation}\label{uf}
\pi\party\Thopt\ls
=\pi\same\Thopt\ls
=\pi\PJR\Thopt\ls
=\pi\EJR\Thopt\ls
=\frac{\ell}{S+1}
.
\end{equation}

Finally, \eqref{uf} and \eqref{tactic-same} yield
\begin{equation}\label{uff}
\pi\tactic\Thopt\ls
\le\pi\same\Thopt\ls
=\frac{\ell}{S+1}
.
\end{equation}
Hence,
\refT{Tsplit}\ref{Tsplit=} applies and shows that
the inequality \eqref{uff} is an equality.
\end{proof}

\subsection{Thiele's addition method}\label{SSTha}

The results in \refTs{TsamePhru} and \ref{TEJRTh} do not hold
for Thiele's addition method. We first give a concrete example 
by \citet{Tenow1912}.

\begin{example}\label{ETenow1912}
Consider an election by Thiele's addition method with $S=3$ seats
and 50 voters, divided into two parties.
Suppose first that all vote along party lines:
\begin{val}
\item [37] $ABC$
\item [13]$KLM$
\end{val}
Then Thiele's method reduces to D'Hondt's method
and the larger party gets 2 seats and the smaller 1.

However,
the larger party may cunningly split their votes on five different
lists as follows:
\begin{val}
\item [1]$A$
\item [9]$AB$
\item [9]$AC$
\item [9]$B$
\item [9]$C$
\item [13]$KLM$
\end{val}
Then $A$ gets the first seat (19 votes), 
and the next two go to $B$ and $C$ (in some order)
with 13.5 votes each, beating $KLM$ with 13. Thus the large party gets all
seats.
(See \cite[Example 13.13]{SJV9} for a further discussion.)

For $\samex$ or $\PJRx$ with $\ell=1$, 
this is a bad outcome for the $KLM$ party, and thus
\begin{equation}
\pi\PJR\Tha(1,3) \ge
\pi\same\Tha(1,3) 
\ge \frac{13}{50}=0.26 > \frac{1}{4}=\frac{\ell}{S+1}. 
\end{equation}
\end{example}

Moreover,
it was shown in \cite{EJR} that \Than{} does not satisfy JR, \ie, by
\eqref{JR}, that $\pi\Tha\PJR\ge 1/S$ for some $S$.
In \cite{EJR}, this was shown for $S\ge10$, by an example;
\cite{PJR2016,PJR2017} then found the sharp range to be $S\ge6$, by solving
a linear programming problem.
The analysis in \cite{PJR2016,PJR2017} 
is easily modified to give a method for calculating $\pi\Tha\PJR\is$ for any
given $S$.

For convenience, we let in the remainder of this subsection 
the ``number of votes'' be arbitrary positive real (or at least rational)
numbers; see \refRs{Rreal} and \ref{Rhomo},
and note that this does not affect the results since 
\Than{} is homogeneous. 
This really means that we allow ballots with weights; 
in such cases all counts and sums over ballots should be interpreted
accordingly; for convenience we omit this from the notation.

First, consider the problem of electing $n$ candidates $C_1,\dots,C_n$ with
a score of at least 1 each, and as few votes as possible, assuming that
there are no other candidates. 
Denote the minimum total number of votes by $\ga_n$.

We may compute $\ga_n$ as follows.
We may assume that $C_1,\dots,C_n$ are elected in this order.
For notational convenience, 
we identify $\set{C_1,\dots,C_n}$ with $\nn$; we thus regard the ballots as 
subsets of $\nn$, with $\gs\subseteq\nn$ interpreted as a vote
on $\set{C_i:i\in\gs}$. 
Let $x_\gs$ be the number of votes $\gs$.
The condition that the $n$ elected are $C_1,\dots,C_n$ in this order, with
scores at least 1, can  then be written as a number of linear inequalities
in the $2^n-1$ variables $x_\gs$, $\gs\neq\emptyset$.
(We may ignore $\gs=\emptyset$, which counts blank votes, since these can be
deleted without affecting the outcome.)
For example, for $n=3$, we obtain the system
\begin{align}
x_1,x_2,x_3,x_{12},x_{13},x_{23},x_{123}&\ge0  \label{X0} \\
x_1+x_{12}+x_{13}+x_{123}&\ge x_2+x_{12}+x_{23}+x_{123} \label{X12}\\
x_1+x_{12}+x_{13}+x_{123}&\ge x_3+x_{13}+x_{23}+x_{123} \label{X13}\\
x_2+\tfrac12x_{12}+x_{23}+\tfrac12x_{123}&\ge
  x_3+\tfrac12x_{13}+x_{23}+\tfrac12x_{123} \label{X23}\\
x_1+x_{12}+x_{13}+x_{123}&\ge1 \label{X1}\\
x_2+\tfrac12x_{12}+x_{23}+\tfrac12x_{123}&\ge1 \label{X2}\\
x_3+\tfrac12x_{13}+\tfrac12x_{23}+\tfrac13x_{123}&\ge1 \label{X3}
\end{align}
where \eqref{X12}--\eqref{X13} say that $C_1$ wins over $C_2$ and $C_3$ and
thus is elected first. (At least, in case of a tie, this is possible.)
Similarly, \eqref{X23} says that $C_2$ is elected before $C_3$.
Finally, \eqref{X1}--\eqref{X3} say that $C_1, C_2, C_3$ all are elected
with scores $\ge1$.
\xfootnote{\label{fredundant}%
In fact, it is easily seen that \eqref{X1} and \eqref{X2} are redundant,
since the scores of the elected always are non-increasing.
}

This leads to the linear programming problem
\begin{align}\label{LP}
    &\text{Minimize } x:=\sum_{\gs}x_\gs
\notag\\&
\text{subject to \eqref{X0}--\eqref{X3}, generalized to $n$ candidates}.
\end{align}
Thus, $\ga_n$ equals the minimum in this linear programming problem.
\xfootnote{
\citet{PJR2016,PJR2017} give an equivalent linear programming problem, with
  $\sum_\gs x_\gs=1$ and maximizing the score of $C_n$ when elected. Their
  problem has the maximum $\ga_n\qw$.
}

\begin{theorem}\label{TThaJR}
  For \Than{} and $\ell=1$\kol
  \begin{align}\label{tthajr}
    \pi\Tha\tactic\is
=   \pi\Tha\same\is
= \pi\Tha\PJR\is
=  \pi\Tha\EJR\is
=\frac1{1+\ga_S},
  \end{align}
where $\ga_S$ is given by \eqref{LP}.
\end{theorem}

\begin{proof}
  This time we show the lower bound first.
Let $(x_\gs)_\gs$ be a vector yielding the minimum $\ga_S$
in \eqref{LP}, with $n=S$.
Consider an election with candidates $A,B_1,\dots,B_S$, 
1 vote on $A$ (this is $\cW$), and 
for each $\gs\subseteq[S]$,
$x_\gs$ votes on the corresponding set of candidates $\set{B_i:i\in\gs}$.
Thus the total number of votes is $V=1+\sum_\gs x_\gs=1+\ga_S$.
A possible outcome is \set{B_1,\dots,B_S}, which is a bad outcome for
$\cW$ for any of the scenarios $\samex$, $\PJRx$, $\EJRx$.
Thus, using \eqref{same-pjr},
\begin{align}\label{kula}
  \pi\Tha\EJR\is
\ge   \pi\Tha\PJR\is
\ge  \pi\Tha\same\is
\ge\frac{\WWW}{V}=\frac{1}{1+\ga_S}.
\end{align}

To see that this is also a lower bound for $\pi\Tha\tactic$, regard $\cW$
as a set of voters with total weight 1. They may vote in any way, but if we
add $\ga_S$ further votes as above, choosing $\cB:=\set{B_1,\dots,B_S}$
disjoint from $\cAx:=\bigcup_{\gs\in\cW}\gs$, then the candidates in $\cAx$
will have scores $\le1$, and again the bad outcome $\set{B_1,\dots,B_S}$
  is possible. Hence,
\begin{align}\label{kulb}
  \pi\Tha\tactic\is
\ge\frac{\WWW}{V}=\frac{1}{1+\ga_S}.
\end{align}

Conversely, consider any election with bad outcome for $\EJRx$ with $\ell=1$
for a set $\cW$ of voters. (Recall that $\EJRx$ and $\PJRx$ are the same for
$\ell=1$.)
This means that there exists some candidate $A$ that everyone in $\cW$ has
voted for.
Furthermore, 
if $\cAx:=\bigcup_{\gs\in\cW}\gs$, then 
$\cAx\cap\cE=\emptyset$, \ie, none from $\cAx$ is elected.
Hence, if $\gs\in\cW$, then no candidate in $\gs$ is ever elected, so $\gs$
is counted with full value throughout the counting.
Thus $A$ has score at least $\WWW$ throughout the counting.

Let the elected, in order, be $B_1,\dots,B_S$. Since  $B_i\in\cE$, we have
$B_i\notin\cAx$ for every $i$, \ie, no voter in $\cW$ votes for any $B_i$.
Furthermore, we may delete all votes from voters in $\cVW$ on unelected
candidates; this will not affect the scores of $B_1,\dots,B_S$, and not
increase the score of anyone else, so the outcome 
$\cB:=\set{B_1,\dots,B_S}$ is still
possible. Hence we may assume that if $\gs\in\cVW$, then $\gs\subseteq\cB$.
Moreover, each $B_i$ is elected with score at least the current score of
$A$, and thus at least $\WWW$.

Consequently, considering only the votes from $\cVW$, we have after scaling
all votes by $1/\WWW$ an election of the type described by \eqref{LP},
and thus
\begin{align}\label{kulc}
  \frac{|\cVW|}{\WWW}\ge\ga_S.
\end{align}
This yields 
\begin{align}
  \frac{\WWW}{V} \le \frac{1}{1+\ga_S}
\end{align}
for any bad outcome, and thus $\pi\EJR\Tha\is\le 1/(1+\ga_S)$, which
together with \eqref{kula}, \eqref{kulb} and \eqref{tactic-same} 
completes the proof.
\end{proof}

Before considering $\ell>1$, let us study the sequence $\ga_n$. We begin
with the first values.

\begin{example}
  \label{Ega1}
For $n=1$, the system \eqref{X0}--\eqref{X3} becomes the trivial single
equation $x_1\ge1$, and thus
\begin{equation}
  \label{ga1}
  \ga_1=1.
\end{equation}
\end{example}

\begin{example}
  \label{Ega2}
For $n=2$, the system \eqref{X0}--\eqref{X3} becomes 
\begin{align}
x_1,x_2,x_{12}&\ge0  \label{2X0} \\
x_1+x_{12}&\ge x_2+x_{12} \label{2X12}\\
x_1+x_{12}&\ge1 \label{2X1}\\
x_2+\tfrac12x_{12}&\ge1 \label{2X2}.
\end{align}
It follows from \eqref{2X12} and \eqref{2X2} that
$x_1+x_2+x_{12}\ge 2x_2+x_{12}\ge2$, and this is attained by, for example,
$x_1=x_2=1$, $x_{12}=0$ or $x_1=x_2=0$, $x_{12}=2$.
(Consequently, there is no way to split votes tactically when $n=2$.)
Hence,
\begin{equation}
  \label{ga2}
  \ga_2=2.
\end{equation}
\end{example}

\begin{example}
  \label{Ega3}
For $n=3$, the linear programming problem \eqref{LP} has (by Maple)
the minimum
\begin{equation}
  \label{ga3}
  \ga_3=8/3\doteq2.667,
\end{equation}
with a solution $x_{12}=x_{13}=x_2=x_3=2/3$ and all other $x_\gs=0$.
\xfootnote{In fact, this is the unique solution. This is easily seen by
  solving the dual problem, which leads to considering the linear
  combination
$\frac{1}{3}\cdot\eqref{X12}+
\frac{1}{3}\cdot\eqref{X13}+
\frac{4}{3}\cdot\eqref{X23}+
\frac{8}{3}\cdot\eqref{X3}$; this also verifies \eqref{ga3}. 
We omit the details. 
}
Consequently,
\begin{equation}
  \pi\Tha\same(1,3)=\pi\Tha\PJR(1,3)
=\pi\Tha\EJR(1,3)=\frac{1}{1+8/3}=\frac{3}{11}.
\end{equation}
Note that the solution to the linear programming problem corresponds to the
following election. (\Than, 3 seats.)
\begin{val}
\item [3] $A$
\item [2]$B_1B_2$
\item [2]$B_1B_3$
\item [2]$B_2$
\item [2]$B_3$
\end{val}
$B_1$, $B_2$, $B_3$ tie for the first seat, and if it goes to $B_1$, then 
$A$, $B_2$ and $B_3$ tie for the remaining two seats, and it is possible
that $A$ is not elected.
(Note that the second election in \refE{ETenow1912} is an approximation of this
example, with votes multiplied by $4.5$, 
modified to avoid ties.)
\end{example}

\begin{example}\label{Ega4}
For $n=4$, the linear programming problem \eqref{LP} has (by Maple)
the minimum
\begin{equation}
  \label{ga4}
  \ga_4=24/7\doteq3.4286,
\end{equation}
with one solution $x_{123}=x_{124}=x_{13}=x_{14}=x_{23}=x_{24}=x_3=x_4=3/7$
and all other $x_\gs=0$. 
\xfootnote{This solution is not unique. 
For example, another is given by
$x_{124}=x_{13}=6/7$,
$x_{23}=x_{4}=4/7$,
$x_{24}=x_{3}=2/7$.
}
\end{example}

For $n\ge5$, we do not know the values of $\ga_n$, and thus not of
$\pi\Tha\same(1,n)$; it should be easy to obtain more values by solving
the linear programming problem \eqref{LP} by computer, but we have not done
so.
\xfootnote{\label{fga6}%
\citet{PJR2016} report, in our notation, $\ga_5\qw\doteq0.2389$ and
  $\ga_6\qw\doteq0.204$, 
and thus $\ga_5\doteq4.186$ and $\ga_6\doteq4.90$,
but they do not give exact rational values.
See also the example in Table 1 in the full (arXiv) version of \cite{PJR2017},
which  shows that $\ga_6\le
4992/(6103/6)=29952/6103\doteq4.908$.
(This bound is not sharp and can be improved further.)
} 
We have the following general result.
\begin{theorem}
  \label{Tga}
The sequence $\ga_n$ is weakly increasing and subadditive.
In particular,
\begin{equation}\label{tga}
  \ga_n\le\ga_{n+1}\le\ga_n+1, \qquad n\ge1.
\end{equation}
Furthermore,
\begin{equation}\label{tga2}
  \frac{n}{H_n}\le\ga_n\le n,
\qquad n\ge1,
\end{equation}
and if $n\ge3$, then $\ga_n<n$.
\end{theorem}

\begin{proof}
  Let $\fE_n$ be the set of all elections of the type used to define $\ga_n$
  above. 
Given an election in $\fE_{n+1}$, with $C_1,\dots,C_{n+1}$ elected in this
order,
remove $C_{n+1}$ from all ballots. This gives an election in $\fE_n$, and
it follows that $\ga_n\le\ga_{n+1}$.

Similarly, given two elections in $\fE_{m}$ and $\fE_{n}$, we may relabel
the candidates and assume that the two elections have disjoint sets of
candidates; then their union is an election in $\fE_{m+n}$, and 
the subadditivity 
$\ga_{m+n}\le\ga_m+\ga_n$ follows.

In particular, since $\ga_1=1$, we have $\ga_{n+1}\le\ga_n+1$.
Induction yields $\ga_n\le n$ and, since $\ga_3<3$ by \refE{Ega3}, $\ga_n<n$
for $n\ge3$.

Finally, let $w(\gs,i)\ge0$ be the contribution of ballot $\gs$
to the score of $C_i$ when elected.
Thus $\sum_{\gs}w(\gs,i)\ge1$. On the other hand, if $\gs$ is a ballot with
$k$ names, then $\sumin w(\gs,i)=H_k\le H_n$.
Consequently,
\begin{align}
  n\le \sumin \sum_{\gs}w(\gs,i) 
= \sum_{\gs} \sumin w(\gs,i) 
\le \sum_\gs H_n = H_n|\cV|.
\end{align}
Hence $|\cV|\ge n/H_n$, and thus $\ga_n\ge n/H_n$.
\end{proof}

\begin{problem}
  Find a general formula for $\ga_n$, and thus for $\pi\Tha\same\is$.
\end{problem}

\begin{problem}
  Find an asymptotic formula for $\ga_n$ as \ntoo, and thus for 
$\pi\Tha\same\is$ as \Stoo.
\end{problem}

The limit $\lim_\ntoo \ga_n/n$ exists since $\ga_n$ is subadditive.
We conjecture that the limit is $0$, i.e., $\ga_n=o(n)$;
by \refT{TThaJR}, this is equivalent to $\pi\Tha\same\is\gg1/S$ as \Stoo.

\begin{problem}
  Prove (or disprove) the conjecture $\ga_n/n\to0$ as \ntoo.
\end{problem}

\begin{remark}\label{RThaJR}
It follows from 
\eqref{tthajr} and \eqref{JR} that \Than{} satisfies JR for a given $S$ if
and only if $\ga_S>S-1$. It was shown by \citet{PJR2016,PJR2017} that this
holds if 
and only if $S\le 5$, see \refFn{fga6} and \eqref{tga}.
\end{remark}

We return to $\pi\Tha\ls$ for the scenarios above.
For $\ell>1$, we can only prove an inequality, which we conjecture is sharp.

\begin{theorem}\label{TThaPJR}
  For \Than{} and $\lele$,
  \begin{align}\label{tthapjr}
 \pi\Tha\EJR\ls
\ge \pi\Tha\PJR\ls
\ge    \pi\Tha\same\ls
\ge\frac\ell{\ell+\ga_{S+1-\ell}},
  \end{align}
where $\ga_n$ is given by \eqref{LP}.
\end{theorem}

\begin{conjecture}\label{ConjTh}
  The inequalities in \eqref{tthapjr} are equalities for $\PJRx$ and
  $\samex$,
\ie, $\pi\Tha\PJR\ls=\pi\Tha\same\ls=\ell/(\ell+\ga_{S+1-\ell})$.
\end{conjecture}

\begin{proof}[Proof and discussion]
We prove \eqref{tthapjr} by a simple extension of the example in the proof
of \refT{TThaJR}.
Let $(x_\gs)_\gs$ be a vector yielding the minimum $\ga_n$
in \eqref{LP}, with $n=S+1-\ell$.
Consider an election with candidates $A_1,\dots,A_\ell,B_1,\dots,B_n$;
let $\cW$ be a set of $\ell$ votes 
on $\cA:=\set{A_1,\dots,A_\ell}$,
and let there be for each $\gs\subseteq[n]$,
$x_\gs$ votes on the corresponding set of candidates $\set{B_i:i\in\gs}$.
Thus the total number of votes is $V=\ell+\sum_\gs x_\gs=\ell+\ga_{S+1-\ell}$.
A possible outcome is \set{A_1,\dots,A_{\ell-1},B_1,\dots,B_{S+1-\ell}}, 
which is a bad outcome for $\cW$ for $\samex$ (and for $\PJRx$, $\EJRx$).
Thus \eqref{tthapjr} follows, using \eqref{same-pjr}.

For the converse, for $\PJRx$ (or $\samex$) we may try to argue in the same
way as for \refT{TThaJR}. Thus, consider an election with a bad outcome.
let again $\cAx:=\bigcup_{\gs\in\cW}\gs$, so $|\cE\cap\cAx|\le\ell-1$.
Let $\cB:=\cE\setminus\cAx$, and 
$n:=|\cB|\ge S+1-\ell$.

Throughout the counting, there is at least one candidate in $\cA$ that is
not elected, and this candidate has score $\ge\WWW/\ell$; hence every
elected candidate is elected with a score $\ge\WWW/\ell$;
in particular, this holds for the $n$ candidates in $\cB$.
We may now delete all candidates not in $\cE$.
By construction, the voters in $\cW$ vote only for candidates in $\cAx$.
The problem is that it is possible that voters in $\cVW$ 
vote not only for candidates in $\cB$ but
also for
one or several candidates in $\cAx$.
If this does not happen, then the votes from $\cVW$ yield, after scaling by
$\ell/\WWW$, an election of the type in \eqref{LP}, and thus
\begin{equation}
  |\cVW|\ge \ga_n\frac{\cW}{\ell} \ge \ga_{S+1-\ell}\frac{\cW}{\ell},
\end{equation}
which yields the desired estimate 
\begin{equation}
\frac{\WWW}{|\cV|}\le\frac{\ell}{\ell+\ga_{S+1-\ell}}.
\end{equation}

However, the voters in $\cVW$ are free to also vote for candidates in
$\cAx$.
This may decrease the score for some of the candidates in $\cB$, which
seems like a bad idea,
and we conjecture that the optimal strategy for $\cVW$ does not use votes 
outside $\cB$.
However, such votes may change the order in which
the candidates from $\cB$ are elected, which may change their scores in a
complicated way, so we have not been able to show this conjecture
rigorously, and thus not \refConj{ConjTh}.
Note that voting on candidates from
another party may be a good strategy in some situations,
see \refE{Ecounter} below, 
so one should be careful, even if that example is of a
different type.
\end{proof}

\begin{problem}
  Prove (or disprove) \refConj{ConjTh}.
\end{problem}

\begin{example}\label{Ega3-5}
  \refT{TThaPJR} and \refE{Ega3} show that
  \begin{equation}
    \pi\Tha\same(3,5)
\ge\frac{3}{3+\ga_3}
= \frac{3}{3+8/3}=\frac{9}{17}>\frac12.
  \end{equation}
Thus, with $S=5$, there are bad outcomes where $\cA$ has a majority of the
votes but does not get a majority of the seats.
We may construct a concrete example from the solution to the linear
programming problem in \refE{Ega3} as follows; \cf{} the election in
\refE{Ega3}.
(\Than, 5 seats.)
\begin{val}
\item [9] $A_1A_2A_3$
\item [2]$B_1B_2$
\item [2]$B_1B_3$
\item [2]$B_2$
\item [2]$B_3$
\end{val}
A possible outcome is that the elected are, in order, $A_1,A_2,B_1,B_2,B_3$.
\end{example}

\begin{example}
If $\ell\le S-2$, the $S+1-\ell\ge3$ and \refT{Tga} yields
$\ga_{S+1-\ell}<S+1-\ell$. Hence,
\refT{TThaPJR} yields
\begin{equation}
  \pi\Tha\same\ls>\frac{\ell}{S+1},
\qquad 1\le\ell\le S-2.
\end{equation}
Similarly, since $\ga_6<5$ as shown by \citet{PJR2016,PJR2017}, see
\refR{RThaJR} and \refFn{fga6},  
it follows that if $S-\ell\ge5$, then $\ga_{S+1-\ell}<S-\ell$, 
and thus \refT{TThaPJR} yields
\begin{equation}
  \pi\Tha\same\ls>\frac{\ell}{S},
\qquad 1\le\ell\le S-5.
\end{equation}
\end{example}

To find $\pi\Tha\EJR$ seems even more complicated than to find
$\pi\Tha\PJR$ or $\pi\Tha\same$, and we have no non-trivial result.
\begin{problem}
Find $\pi\Tha\EJR\ls$ for $\ell\ge2$.
\end{problem}

Finally, consider $\pi\tactic\Tha$.
    The values $x_\gs$ in the solution to \eqref{LP} give an optimal
    strategy for how a well-organized party should distribute its votes in
    order to get $n$ candidates elected, provided the other voters vote for
different candidates.
However, \Than{} is non-monotone, and the strategy is risky and 
may be bad if the other voters vote in some other way, and in
particular  if another party knows the strategy and may counteract as
in the following example.

\begin{example}[\cite{SJV9}]
  \label{Ecounter}
Consider again the second election in \refE{ETenow1912}, and suppose that
two voters in the $KLM$ party also vote for $B$, by chance or by clever
design because they know the plan of the other party:
\begin{val}
\item [1]$A$
\item [9]$AB$
\item [9]$AC$
\item [9]$B$
\item [9]$C$
\item [11]$KLM$
\item [2]$BKLM$
\end{val}
Then $B$ is elected first, followed by $C$ and finally $K$ (or $L$ or $M$).
Hence the $ABC$ party gets only 2 seats.
\end{example}

\refE{Ecounter} suggests that the strategy in the linear programming problem
\eqref{LP} is not the best strategy for the scenario $\tacticx$, where we
consider the worst-case behaviour of the voters not in $\cW$, and thus need
a fool-proof strategy. We do not know the best such strategy, and leave it
an open problem to find it, and the resulting threshold.

\begin{problem}
  Find $\pi\Tha\tactic\ls$.
\end{problem}

\subsection{\Then}\label{SThe}

We consider briefly also \Then.

\begin{theorem}
  \label{TThe}
For \Then:
\begin{align}\label{tthe}
  \pi\tactic\The\ls=\pi\same\The\ls=\frac{\ell}{S+1},
\qquad\lele.
\end{align}
\end{theorem}

\begin{proof}
We first give an upper bound for $\pi\same\The\ls$.
Thus suppose that $\cW\subseteq\cV$, and that every voter in $\cW$ votes
for the same list $\cA$ with $|\cA|\ge\ell$. Suppose that the outcome is
bad for $\samex$. This 
means that less than $\ell$ candidates in $\cA$ remain at the end. Hence, in
some round $\cA$ has $\ell$ remaining candidates and one of them, say $A$, is
eliminated. 

Consider this round.
The eliminated candidate $A$ has score $\ge\WWW/\ell$, and thus every other
remaining candidate has at least this score.
Let $m$ be the number of remaining candidates in $\ccA$.
Since the total number of remaining candidates is $>S$
(otherwise there would be no more elimination), 
$m\ge  S+1-\ell$. 
Hence, the total score of the candidates in $\ccA$, $T$ say, 
is $\ge (S+1-\ell)\WWW/\ell$.

 A ballot with $k$ remaining candidates
contributes $1/k$ to the score of each of them, so the total contribution of
the ballot is 1, for any $k\ge1$. (And $0$ if no candidate on the ballot
remains.) Since no voter in $\cW$ votes for any candidate in $\ccA$, it
follows that
\begin{align}
V-\WWW=  |\cVW| \ge T\ge \frac{S+1-\ell}{\ell}\WWW,
\end{align}
which implies 
\begin{align}
\frac{\WWW}{V} \le \frac{\ell}{S+1}.
\end{align}
Consequently,
\begin{align}\label{thesame}
\pi\same\The\ls\le\frac{\ell}{S+1}.  
\end{align}

Finally, \eqref{thesame} and \eqref{tactic-same} yield
$\pi\tactic\The\ls\le\pi\same\The\ls\le\frac{\ell}{S+1}$, and
thus \refT{Tsplit}\ref{Tsplit=} yields the equalities \eqref{tthe}.
\end{proof}

For $\PJRx$ we give only an example, giving a lower bound for $\pi\PJR\is$ (and
thus for $\pi\EJR\is$).
We do not believe that this example is sharp. In fact, we do not know
whether $\pi\PJR\The\ls<1$ or not, even in the case $\pi\PJR\The(1,1)$.

\begin{example}
  \label{EThe}
Let $S,m,n\ge1$. Consider an election with $1+mn+S$ candidates
$A$, $C_{ij}$ and $B_k$, for $i\in[m]$, $j\in[n]$, $k\in[S]$, and the
foillowing votes:
\begin{val}
\item [${n+1}$] $\set{A}\cup\set{C_{ij}:j\in[n]}$ for every $i\in[m]$
\item [${m-1}$] $\set{C_{ij}}$ for every $i\in[m]$, $j\in[n]$
\item [$m+n$] \set{B_k} for every $k\in[S]$.
\end{val}
The total number of votes is thus
\begin{align}
  V=m\xpar{n+1}+mn\xpar{m-1}+S(m+n)
=m\xpar{mn+1}+S(m+n).
\end{align}

Initially, $A$ has score $\frac{m(n+1)}{n+1}=m$, 
each $C_{ij}$ has score 
$\frac{n+1}{n+1}+{m-1}={m}$
and each $B_k$ has score $m+n$. Thus either $A$ or some $C_{ij}$ is
eliminated; suppose that $A$ is.
In the sequel, every remaining $C_{ij}$ has score $\le
{n+1}+{m-1}=m+n$, and each remaining $B_k$ still has score $m+n$.
We may thus assume that in each round (after the first) some $C_{ij}$ is
eliminated, until only $B_1,\dots,B_S$ remain and are elected.

Consider the scenario $\PJRx$ with $\ell=1$,
let $\cA=\set{A}$, and let the voters $\cW$ be the $m\xpar{n+1}$ voters
in the first line above. Then the election above is a bad outcome.
Hence, for any $S,m,n\ge1$,
\begin{align}\label{elim1}
  \pi\PJR\The(1,S)\ge\frac{\WWW}{V}
=\frac{m\xpar{n+1}}{m\xpar{mn+1}+\xpar{m+n}S}.
\end{align}
Letting $n\to\infty$, we obtain
\begin{align}\label{elim2}
  \pi\PJR\The(1,S)
\ge
\frac{m}{m^2+S}, 
\qquad m\ge1.
\end{align}
For a numerical example, with $S=3$ and $m=2$, and comparing with \eqref{tthe},
\begin{align}
\pi\PJR\The(1,3)\ge \frac{2}{7}>\pi\same\The(1,3)=\frac{1}{4}.  
\end{align}

Given $S\ge1$, we may thus maximize the \rhs{} of \eqref{elim2} over
$m\ge1$; the maximum is obtained for $m=\floor{\sqrt{S}}$ or
$m=\ceil{\sqrt{S}}$.
Taking $m=\floor{\sqrt S}$ yields the lower bound
\begin{align}\label{elimSoo}
  \pi\PJR\The(1,S)
\ge
\frac{\floor{\sqrt{S}}}{2S}\sim \frac{1}{2\sqrt S},
\qquad\text{as \Stoo}.
\end{align}
In particular, by \eqref{JR}, \Then{} does not satisfy JR.
\end{example}

\begin{problem}
  What is $\pi\The\PJR\ls$?
\end{problem}

\subsection{\Thoptn{} with general weights}

We have in \refSs{SThopt}--\ref{SSTha} studied Thiele's optimization and
addition methods with the standard weights $w_k=1/k$.
In this and the next subsection, we extend  
some of the results to general weights $\ww=(w_k)_0^\infty$.
Our results are inspired by, and extend considerably, results 
by \citet{EJR} 
on the criteria JR and EJR, see \refRs{RThoptww} and \ref{RThawJR}.

We assume throughout that $w_1=1$ and that $w_1\ge w_2\ge \dots \ge0$.

\begin{theorem}
  \label{Toptw}
For \Thoptn{} with general weights $\ww=(w_k)_0^\infty$\kol
For $S\ge1$, 
\begin{equation}\label{toptw}
  \pi\Thoptw{\ww}\same\is
=
  \pi\Thoptw{\ww}\PJR\is
=
  \pi\Thoptw{\ww}\EJR\is
=\frac{1}{1+\xfrac{S}{\max_{k\le S} (kw_k)}}
\end{equation}
\end{theorem}
\begin{proof}
  Consider an election with a bad outcome for $\EJRx$.
Let $\cAx:=\bigcup_{\gs\in\cW}\gs$.
Then $\cE\cap\cAx=\emptyset$; furthermore, there exists a candidate
$A\in\cAx$ such that $A\in\gs$ for every $\gs\in\cW$.

Let the elected be $B_1,\dots,B_S\notin\cAx$.
We may eliminate all candidates except $A$ and $B_1,\dots,B_S$; this yields
an election with the same bad outcome $\cE=\cB:=\set{B_1,\dots,B_S}$.
Furthermore, we may delete $A$ from any ballot not in $\cW$; this will not
change the total satisfaction
$\Psi(\cB)$,
while the total satisfaction $\Psi(\cC')$ for every other set $\cC'\subseteq
\set{A}\cup\cB$ is decreased or remains the same.
Hence, $\cB$ still maximizes the total satisfaction, and is thus a possible
outcome which is bad.
Note that each voter in $\cW$ now votes for \set{A}, and each voter in
$\cVW$ for some subset of $\cB$.

Suppose that $A$ is elected instead of $B_j$. This increases the
satisfaction of every $\gs\in\cW$ from 0 to $w_1=1$.
If $\gs\in\cVW$ and $B_j\in\gs$, then the satisfaction of
$\gs$ decreases from $\psi(|\gs|)$ to $\psi(|\gs|-1)$, \ie, by $w_{|\gs|}$.
For all other $\gs\in\cVW$, the satisfaction remains the same.
Hence, the change in total satisfaction \eqref{sat} is
\begin{equation}
  \Psi\bigpar{\cB\setminus\set{B_j}\cup\set{A}}
-\Psi(\cB)
=
\WWW-\sum_{\gs\ni B_j} w_{|\gs|}.
\end{equation}
This change is $\le0$, since $\cB$ is elected and thus maximizes the total
satisfaction.
Summing over $j\in[S]$ we thus obtain
\begin{equation}
0
\ge
\sumjS\Bigpar{\WWW-\sum_{\gs\ni B_j} w_{|\gs|}}
=S\WWW - \sum_{\gs\in\cVW}|\gs|w_{|\gs|}.
\end{equation}
Thus,
\begin{equation}
  \WWW \le\frac{1}{S} \sum_{\gs\in\cVW}|\gs|w_{|\gs|}
\le\frac{|\cVW|}{S}\max_{1\le k\le S}\bigpar{|k|w_{k}}.
\end{equation}
which yields
\begin{equation}
\frac{\WWW}{V}\le\frac{1}{1+{S}/\max_{1\le k\le S}\bigpar{|k|w_{k}}}
\end{equation}
and thus an upper bound ``$\le$'' in \eqref{toptw}.

Conversely, for any $k\in[S]$,
consider an election with $S+1$ candidates 
$A,B_1,\dots,B_S$. Let $W\ge2$ 
and let there be $V=W+S$ votes:
$W$ votes on $\set{A}$ (this is the set $\cW$) and 1 vote on each set
$\set{B_i,\dots,B_{i+k-1}}$ (indices mod $S$) for each $i\in[S]$.

The sets of $S$ candidates are $\cB:=\set{B_1,\dots,B_S}$ and
$\cA_i:=\set{A}\cup\cB\setminus\set{B_i}$.
The total satisfaction if these sets are elected are
\begin{align}
  \Psi(\cB)&=S\psi(k),
\\
\Psi(\cA_i)&=W+(S-k)\psi(k)+k\psi(k-1)
=\Psi(\cB)+W-k w_k.
\end{align}
Hence, if $W<kw_k$, then $\cB$ is elected, a bad outcome for $\samex$.
As usual, we may here let $W$ be a rational number, since we may multiply
the votes above by the denominator of $W$. Hence, we obtain
\begin{align}
  \pi\same\Thoptw\ww\is\ge\frac{kw_k}{kw_k+S}=\frac{1}{1+S/(kw_k)}.
\end{align}
We may now take the maximum over $k\in[S]$ to get the inequality ``$\ge$''
for $\pi\same\Thoptw\ww$ in \eqref{toptw}, which by \eqref{same-pjr}
completes the proof. 
\end{proof}

\begin{theorem}\label{Tw1}
    For \Thoptn{} with
weights $(w_k)\xoo$ satifying $w_k\le 1/k$, $k\ge1$\kol
For $S\ge1$, 
\begin{equation}\label{tw1}
  \pi\Thoptw{\ww}\same\is
=
  \pi\Thoptw{\ww}\PJR\is
=
  \pi\Thoptw{\ww}\EJR\is
=\frac{1}{S+1}.
\end{equation}
Furthermore,
\begin{equation}\label{tw1l}
  \pix\Thoptw{\ww}\tactic\ls
=\frac{\ell}{S+1},
\qquad\lele.
\end{equation}
In particular, this applies to \Thoptweakn{}. 
\end{theorem}

\begin{proof}
By assumption, recalling also the standing assumption $w_1=1$, we have
$\max_k(kw_k)=1$. Hence, \eqref{tw1} follows from \refT{Toptw}.
  
This implies, by \eqref{tactic-same},
$\pi\Thoptw\ww\tactic\is\le
\pi\Thoptw\ww\same\is=\xqfrac{1}{S+1}$,
and thus \eqref{tw1l} follows from
\refT{Tsplit}\ref{Tsplit=x}.
\end{proof}

With the weak method, it is for $\ell\ge2$ obviously a bad strategy to vote
on the same list. 
One aspect of this is the following trivial result.
(Recall that the proof of \eqref{tw1l} by \refT{Tsplit}
uses the same strategy as for SNTV in
\refT{TSNTV}, with the votes split on different candidates separately.)
\begin{theorem}\label{Tthoptw2}
    For \Thoptweakn{}\kol
If\/ $2\le\ell\le S$, then 
\begin{equation}\label{toptw2}
  \pi\Thoptw{\weakw}\same\ls
=
  \pi\Thoptw{\weakw}\PJR\ls
=
  \pi\Thoptw{\weakw}\EJR\ls
=1.
\end{equation}
\end{theorem}
\begin{proof}
Let $W\ge1$.
  Consider an election with $S+\ell-1$ candidates
  $A_1,\dots,A_\ell,\allowbreak B_1,\dots,B_{S-1}$ and $V=W+S-1$ votes:
$W$ votes on \set{A_1,\dots,A_\ell} (this is the set $\cW$)
and 1 vote on each $B_i$.
With \Thoptweakn{} \eg{} \set{A_1,B_1,\dots,B_{S-1}} is elected, a bad outcome
for $\samex$.
Consequently, 
$\pi\same\Thoptw\weakw\ls\ge W/(W+S-1)$, and the result follows, using 
\eqref{same-pjr},
since $W$ is arbitrary.
\end{proof}

We can extend \refT{Tthoptw2} to other weights in the form of an inequality.

\begin{theorem}
  \label{Tw3}
    For \Thoptn{} with 
weights $(w_k)\xoo$\kol
\begin{align}\label{tw3}
  \pi\Thoptw{\ww}\EJR\ls
&\ge
  \pi\Thoptw{\ww}\PJR\ls
\ge
  \pi\Thoptw{\ww}\same\ls
\notag\\&
\ge\frac{w_\ell\qw}{w_\ell\qw+S+1-\ell}
= \frac{1}{w_\ell(S+1-\ell)+1}.
\end{align}
\end{theorem}

\begin{proof}
Let $W\ge1$.
  Consider an election with $S+1$ candidates $\cA=\set{A_1,\dots,A_\ell}$
and $\cB=\set{B_1,\dots,B_{S+1-\ell}}$, and $V=W+S+1-\ell$ votes:
$W$ votes on $\cA$ (this is the set $\cW$)
and 1 vote on $\set{B_j}$ for each $j\in[S+1-\ell]$.

Let $\Ex$ be a set of $S$ candidates. 
Then, the total satisfaction  \eqref{sat} if $\Ex$ is elected is
\begin{align}
  \Psi(\Ex)=
  \begin{cases}
    W\psi(\ell)+S-\ell, & \Ex\supseteq\cA,
\\
  W\psi(\ell-1)+S+1-\ell, & \Ex\not\supseteq\cA
  \end{cases}
\end{align}
The difference between the two values is
\begin{align}
\bigpar{W\psi(\ell)+S-\ell}-\bigpar{  W\psi(\ell-1)+S+1-\ell}
  =Ww_\ell-1.
\end{align}
Hence, if $Ww_\ell<1$, then the elected set $\cE\not\supseteq\cA$; this is a bad
outcome for  $\samex$, and thus
\begin{align}\label{kaja}
  \pi\Thoptw\ww\same\ge\frac{\WWW}{V}=\frac{W}{W+S+1-\ell}.
\end{align}
We may here (as usual) 
let $W$ be any rational number $<w_\ell\qw$, since we may
multiply all numbers of votes by its denominator to obtain integer values;
hence, \eqref{kaja} holds also if we replace $W$ by $w_\ell\qw$, which
yields \eqref{tw3}, using also \eqref{same-pjr}.
\end{proof}

\begin{remark}\label{RThoptww}
It follows from \eqref{toptw} that, as \Stoo,
  \begin{align}\label{duva}
    S\pi\Thoptw\ww\PJR\is=
    S\pi\Thoptw\ww\same\is \to\max_{k\ge1}(kw_k).
  \end{align}
In particular, if $w_k>1/k$ for some $k$, then \eqref{toptw} or \eqref{duva}
implies that
$\pi\Thoptw\ww\PJR\is \ge 1/S$ for large $S$; thus, by \eqref{JR}, JR does
not hold.

On the other hand, 
\refT{Tw1} shows that if $w_k\le1/k$ for every $k$, then
$\pi\Thoptw\ww\PJR\is < 1/S$; thus, by \eqref{JR}, JR holds.

Hence, as shown by \citet{EJR}, JR holds for $\Thoptwx\ww$ if and only
if $w_k\le1/k$ for every $k$.

Furthermore, if $w_k<1/k$ for some $k$, then \eqref{tw3} implies 
$\pi\Thoptw\ww\EJR(k,S)\ge\pi\Thoptw\ww\PJR(k,S)>k/S$ for large $S$, and
thus $\Thoptwx\ww$ does not satisfy PJR or EJR.

Consequently, recalling \refT{TEJRTh},
the standard weights $w_k=1/k$ are the only weights that yield
a method satisfying PJR or EJR.
For EJR, this too  was shown in \cite{EJR}; for PJR it was then shown in
\cite{PJR2016,PJR2017}.
\end{remark}

\subsection{\Than{} with general weights}
We next study \Than{} with
general weights $\ww=(w_k)_0^\infty$.
We assume as above that $w_1=1$ and that $w_1\ge w_2\ge \dots \ge0$.

Some of the analysis in \refSS{SSTha} extends to this case.
Define $\ga_n(\ww)$ as in \refSS{SSTha}, \ie, as the minimum number of
(real-valued) votes in an election with (exactly) $n$ candidates such that all
are elected with scores $\ge1$.
Thus, $\ga_n(\ww)$ can be computed by a linear programming problem of the
type in
\eqref{X0}--\eqref{X3}, but with the fractions $1/k$ replaced by the weights
$w_k$ as coefficients. Clearly, for any $\ww$,
\begin{align}
  \ga_1(\ww)=1.
\end{align}

First, we note that \refT{TThaJR} extends, with the same proof.

\begin{theorem}\label{TThaJRw}
  For \Thanww{} $\ww$ and $\ell=1$\kol
  \begin{align}\label{tthajrw}
    \pi\Thaw\ww\tactic\is
&=   \pi\Thaw\ww\same\is
= \pi\Thaw\ww\PJR\is
=  \pi\Thaw\ww\EJR\is
\notag\\&
=\frac1{1+\ga_S(\ww)},
\qquad S\ge1,
  \end{align}
where $\ga_S(\ww)$ is given by the version of \eqref{LP} with weights $\ww$.
\end{theorem}

Also the proof of \refT{Tga} extends, and yields the following.

\begin{theorem}
  \label{Tgaw}
The sequence $\ga_n(\ww)$ is weakly increasing and subadditive.
In particular,
\begin{equation}\label{tgaw}
  \ga_n(\ww)\le\ga_{n+1}(\ww)\le\ga_n(\ww)+1, \qquad n\ge1.
\end{equation}
Furthermore,
\begin{equation}\label{tga2w}
  \frac{n}{\psi(n)}\le\ga_n(\ww)\le n,
\qquad n\ge1,
\end{equation}
where $\psi(n)=\sumkn w_k$.
\nopf
\end{theorem}

\begin{corollary}\label{CThaw}
  For \Thanww{} $\ww$ and $\ell=1$\kol
For any scenario $\fS=\tacticx, \samex, \PJRx, \EJRx$,
  \begin{align}\label{cthaw}
\frac{1}{S+1} \le 
    \piS\Thaw\ww\is
\le\frac1{1+ S/\psi(S)},
\qquad S\ge1.
  \end{align}
\end{corollary}

\begin{proof}
  By \refT{TThaJRw} and the inequalities \eqref{tga2w}.
\end{proof}

For the weak method, we have a simple result.

\begin{theorem}
  \label{TThaweak1}
  For \Thaweakn{} and $\ell=1$\kol
For any scenario $\fS=\tacticx, \samex, \PJRx, \EJRx$,
  \begin{align}\label{tthajrw1}
    \piS\Thaw\weakw\is
=\frac1{S+1},
\qquad S\ge1.
  \end{align}
\end{theorem}

\begin{proof}
  For the weak method, $\psi(n)=1$ for $n\ge1$, and thus \eqref{tga2w}
  yields
  \begin{align}
    \ga_n(\weakw)=n,
\qquad n\ge1.
  \end{align}
Thus the result follows by \refT{TThaJRw} (or \refC{CThaw}).
\end{proof}

\begin{example}
  \label{Ega2w}
For $n=2$, the system \eqref{X0}--\eqref{X3} with weights becomes,  
after deleting a redundant inequality, \cf{} \refFn{fredundant},
\begin{align}
x_1,x_2,x_{12}&\ge0  \label{2X0w} \\
x_1+x_{12}&\ge x_2+x_{12} \label{2X12w}\\
x_2+w_2 x_{12}&\ge1 \label{2X2w}.
\end{align}
Here \eqref{2X12w} simplifies to $x_1\ge x_2$.
It is clear that the minimum $\ga_2(\ww)$ of $x_1+x_2+x_{12}$ is obtained
with $0\le w_2x_{12}\le1$ (otherwise we could decrease $x_{12}$), and that
we have equalities in \eqref{2X12w} and \eqref{2X2w} (otherwise we could
decrease $x_1$ or $x_2$, respectively).
Hence, the optimum is given by $x_1=x_2=1-w_2x_{12}$, and we obtain
\begin{align}\label{ga2w}
  \ga_2(\ww) = \min_{0\le x\le w_2\qw} \bigpar{2(1-w_2x)+x}
= \min\bigpar{2,w_2\qw}
=
\begin{cases}
  2, & 0\le w_2\le\frac12,
\\
w_2\qw,& \frac12\le w_2\le 1.
\end{cases}
\end{align}

\refT{TThaJRw} and \eqref{ga2w} yield,
for $\fS=\tacticx, \samex, \PJRx, \EJRx$,
\begin{align}\label{pi12w}
  \piS\Thaw\ww(1,2)
=
\begin{cases}
  \frac{1}3, &0\le w_2\le\frac12,
\\
\frac{w_2}{1+w_2},& \frac12\le w_2\le 1.
\end{cases}
\end{align}
\end{example}

\begin{remark}
  \label{RThawJR}
\refT{TThaweak1} shows in particular, by \eqref{JR}, that JR holds for
\Thaweakn. This was shown by \citet{EJR}, who also proved that
JR does not hold for any other weights. We show this in \refE{Eonlyweak}
below. 
\end{remark}

\begin{example}\label{Eonlyweak}
Suppose $\ww\neq\weakw$; in other words (by our assumptions on $\ww$) that
$w_2>0$. 

  Consider an election with $n\ge2$ candidates $C_0,\dots,C_{n-1}$, 
$n$  seats, 
and  $(n-1)^2$ votes:
1 vote for \set{C_0,C_i} and $n-2$ votes for \set{C_i} for each $i\in[n-1]$.
All candidates tie with $n-1$ votes each. If $C_0$ is elected first, then
the others are all elected with scores $n-2+w_2$. 
Hence,
\begin{align}
  \ga_n(\ww) \le \frac{(n-1)^2}{n-2+w_2}=n-\frac{nw_2-1}{n-2+w_2}.
\end{align}
In particular, 
$\ga_n(w)<n$ if $n>1/w_2$. (Cf.\ the precise \refE{Ega3} and the final claim
in \refT{Tga} for the standard weights.)

Furthermore, by subadditivity (\refT{Tgaw}), for any $m\ge1$,
\begin{align}
  \ga_{mn}(\ww) 
\le m \ga_{n}(\ww) 
\le mn-m\frac{nw_2-1}{n-2+w_2}.
\end{align}

In particular, for any $\ww\neq\weakw$, we can choose first $n$ such that
$\ga_n(\ww)<n$, and then $m$ such that $\ga_{mn}(\ww)<mn-1$. Then, with
$S:=mn$, \refT{TThaJRw} shows that
\begin{align}
  \pi\PJR\Thaw\ww\is >\frac{1}{S}.
\end{align}
Consequently, by \eqref{JR}, JR does not hold for 
\Thanww{} $\ww\neq\weakw$.
\end{example}

For $\ell>1$, we have even for the standard weights in \refSS{SSTha}
only partial results.
\refT{TThaPJR} extends to general weights.

\begin{theorem}\label{TThaPJRw}
  For \Thanww{} $\ww$ and $\lele$,
  \begin{align}\label{tthapjrw}
 \pi\Thaw\ww\EJR\ls
\ge \pi\Thaw\ww\PJR\ls
\ge    \pi\Thaw\ww\same\ls
\ge\frac{w\qw_\ell}{w\qw_\ell+\ga_{S+1-\ell}(\ww)}.
  \end{align}
If $w_\ell=0$, this is interpreted as
  \begin{align}\label{kajsa}
 \pi\Thaw\ww\EJR\ls
= \pi\Thaw\ww\PJR\ls
=    \pi\Thaw\ww\same\ls
=1.
  \end{align}
\nopf
\end{theorem}
\begin{proof}
The same as for \refT{TThaPJR},
taking $w\qw_\ell$ votes on $\cA$ if $w_\ell>0$, and otherwise an arbitrary
number. 
\end{proof}

In particular, \eqref{kajsa} holds for \Thaweakn.
For the weak method, we have also a simple result for $\pix\tactic$.

\begin{theorem}\label{Tthaweak}
    For \Thaweakn{}\kol
\begin{equation}\label{taw1l}
  \pix\Thaw{\weakw}\tactic\ls
=\frac{\ell}{S+1},
\qquad\lele.
\end{equation}
Furthermore, if\/ $2\le\ell\le S$, then 
\begin{equation}\label{taddw2}
  \pi\Thaw{\weakw}\same\ls
=
  \pi\Thaw{\weakw}\PJR\ls
=
  \pi\Thaw{\weakw}\EJR\ls
=1.
\end{equation}
\end{theorem}
\begin{proof}
First, \eqref{taddw2} is a special case of \eqref{kajsa}.

Furthermore, \eqref{taw1l} follows from \refTs{TThaweak1} and
\ref{Tsplit}\ref{Tsplit=x}.  
\end{proof}

\section{Ordered ballots: PSC and STV}\label{SPSC}

In this section we consider election methods with ordered  ballots.
First, recall that $\pi\party$, $\pi\same$ and $\pi\tactic$ in
\refDs{Dparty2}, \ref{Dsame} and \ref{Dtactic} apply to ordered ballots too,
and that the inequalities in \refT{Tps} hold.

Furthermore, let us
return to DPC and PSC, used as an example in \refS{S:intro}, and define the
corresponding threshold. 
(We use PSC to denote it, since PSC is used in
recent literature, \eg{} \cite{ElkindEtal2017} and \cite{AzizLee}.)

\begin{definition}[PSC, ordered ballots: $\pi\PSC$]\label{DPSC}
Suppose that all voters in $\cW$ put the same $m\ge\ellx$ candidates
(not necessarily in the same order) 
as the top $m$ candidates in their preference listings.
A good outcome is when at least $\ellx$ of these candidates are elected.  
\end{definition}

As said in \refS{S:intro}, the Droop Proportionality Criterion (DPC) 
formulated by
\citet{Woodall:properties} is equivalent to $\pi\DPC\ls\le \ell/(S+1)$.

Furthermore, \citet{Tideman} formulated what he called 
``proportionality for solid coalitions'' (PSC), which easily is seen to be
equivalent to $\pi\DPC\ls\le \ell/S-$, and thus in practice to
$\pi\DPC\ls< \ell/S$ for $\ell<S$, \cf{} \refTx{Tx5}.
\xfootnote{\label{fnDummett}%
\citet{Tideman} refers to \citet[p.~282]{Dummett}, and attributes
both the property and the name ``proportionality for solid coalitions'' to
him. This seems inaccurate.  
\citet[pp.~282--283]{Dummett} discusses this type of property, but he does
not give it a name 
(and does not use the term ``solid coalition'' used by
\citet{Tideman}). Moreover, \citet{Dummett} talks about (in our
notation) a set of $\ell Q$ voters, where $Q$ is the quota, so ignoring
effects of rounding the quota, Dummett's property is the same as
DPC in \citet{Woodall:properties}, and is thus essentially 
equivalent to $\pi\DPC\ls\le \ell/(S+1)$.
See also \refT{TSTV}.
}

DPC and PSC are properties satisfied by STV, see \refT{TSTV} below, but
more or less only by STV; the properties depend on the possibility of
eliminations in the election method, so that, in \refD{DPSC} 
(for sufficiently large $\WWW$), 
even if the votes of $\cW$ are split on
a large number $m$ of candidates, they will by eliminations be concentrated
on (at least) $\ell$ candidates that become elected. 
For election methods without eliminations, for example \phragmen's and
Thiele's, this does not happen, and even a large set of voters $\cW$ may
fail to be represented by splitting their votes on too many candidates, even
in the scenario $\PSCx$ in \refD{DPSC} above, see the trivial result in 
\refT{TPhrPSCbad} below.
Therefore, also a weak version of PSC requiring $m=\ell$ 
has been studied, \eg{} by
\citet{AzizLee}.
This leads to the following definition.

\begin{definition}[weak PSC, ordered ballots: $\pi\wPSC$]\label{DwPSC}
Suppose that all voters in $\cW$ put the same $\ellx$ candidates
(not necessarily in the same order) 
as the top $\ellx$ candidates in their preference listings.
A good outcome is when all these candidates are elected.  
\end{definition}

\begin{remark}  \label{RPSC}
The property called ``weak $q$-PSC'' in \citet{AzizLee} is, with $q=\gk V$,
equivalent to 
$\pi\wPSC\ls<\ell\gk$, or, more precicely, see \refR{Rrefined}, 
$\pi\wPSC\ls\le (\ell\gk)-$.
Furthermore,
the property called ``Proportionality for Solid Coalitions'' in
\cite{ElkindEtal2017} 
uses this (weak) version (and is the same as weak $(V/S)$-PSC in 
\cite{AzizLee});
hence it is in our notation
$\pi\wPSC\ls\le (\ell/S)-$
(or, in practice, $\pi\wPSC\ls<\ell/S$, $\ell<S$).
\xfootnote{
As remarked by \cite{AzizLee},
\cite{ElkindEtal2017} 
inaccurately attributes this weak version of PSC
to \citet{Dummett}; furthermore,  
\cite{ElkindEtal2017} 
says that Dummett's original proposal
is, in our notation, 
$\pi\wPSC\ls\le (\ell/S)-$.
This seems to be doubly inaccurate, since \citet[p.~283]{Dummett} uses the
stronger 
PSC and the threshold $\ell Q/V\approx \ell/(S+1)$, see
\refFn{fnDummett}.
}
Moreover, three further definitions in \cite{ElkindEtal2017} 
also fit the weak version:
``Solid coalitions'' is the special case $\ell=1$:
$\pi\wPSC(1,S)<1/S$;
``Unanimity'' is the case $\ell=S$: $\pi\wPSC(S,S)\le 1-$;
``Fixed Majority'' is $\pi\wPSC(S,S)\le 1/2$.
\end{remark}

\begin{theorem}
  \label{Tsame-psc}
For any election method with ordered  ballots, and $1\le\ell\le S$,
\begin{align}
  \pi\party\ls &\le \pi\same\ls 
\le \pi\wPSC\ls 
\le \pi\PSC\ls 
.\label{same-psc}
\end{align}
\end{theorem}

\begin{proof}
The first inequality is \eqref{party-same} in \refT{Tps}.

The second inequality follows since an instance of the scenario $\samex$ also
is an instance of $\wPSCx$, and if it is bad for $\samex$ then it is bad for
$\wPSCx$. 

Similarly, the third inequality follows because an instance of $\wPSCx$ with
bad outcome is an instance of $\PSCx$ with bad outcome.
\end{proof}

\citet{Dummett} and \citet{Woodall:properties} 
both stress that STV satisfies the DPC type
of property. In our notation, we can state their results as
\begin{equation}\label{dpc}
  \pi\STV\PSC\ls\le\frac{\ell}{S+1},
\end{equation}
if the Droop quota is used and we ignore rounding.
More precisely,
\cite{Woodall:properties} assumes that the unrounded Droop quota $V/(S+1)$
is used, and then \eqref{dpc} holds.
On the other hand, \cite{Dummett} assumes that the quota $Q$ is larger
(as it usually is in practice), and states the result that
$\ell Q$ voters 
are enough to guarantee $\ell$ seats.
This, in fact, holds for any quota $Q>V/(S+1)$, and can be stated as follows.
\begin{theorem}[\citet{Dummett}]\label{TSTVD}
  Any (reasonable) version of STV with quota $Q=r V$, where $r> 1/(S+1)$, 
satisfies
  \begin{align}\label{tstvD}
\pi\STV\PSC\ls<r\ell,
\qquad\lele.
  \end{align}
Thus, a set $\cW$ of at least $\ell Q$ voters forming a ``solid coalition''
in the sense of $\PSCx$ will get at least $\ell$ of its candidates elected. 
\end{theorem}

This is easy to see directly, but we instead prove the following
more precise result, yielding \refT{TSTVD} as a corollary.

\begin{theorem}\label{TSTV}
Consider any (reasonable) version of STV  with quota $Q=V/(S+\gd)$, where
$\gd\in\oi$,
Then, for $1\le \ell\le S$,
\begin{align}
  \pi\party\STV\ls
&
=  \pi\same\STV\ls
=  \pi\wPSC\STV\ls
=  \pi\PSC\STV\ls
\notag\\&
= \frac{\ell(S+2-\ell)-1+\gd}{(S+\gd)(S+2-\ell)}
=\frac{\ell}{S+\gd}-\frac{1-\gd}{(S+\gd)(S+2-\ell)}
\notag\\&
=\frac{\ell}{S+1}+\frac{(1-\gd)(\ell-1)(S+1-\ell)}{(S+\gd)(S+1)(S+2-\ell)},
\label{tstv}
\\\label{tstvtactic}
  \pix\tactic\STV\ls
&= \frac{\ell}{S+1}.
\end{align}
As a special case,
\begin{align}\label{tstv1}
  \pi\party\STV\is
&
=  \pi\tactic\STV\is
=  \pi\same\STV\is
=  \pi\wPSC\STV\is
=  \pi\PSC\STV\is
=\frac{1}{S+1}.
\end{align}
In particular, with the (unrounded) Droop quota $Q=V/(S+1)$,
\begin{align}\label{tstv2}
  \pi\party\STV\ls
&
=  \pi\tactic\STV\ls
=  \pi\same\STV\ls
=  \pi\wPSC\STV\ls
=  \pi\PSC\STV\ls
=\frac{\ell}{S+1}.
\end{align}
Furthermore, also with a rounded Droop quota $Q=V/(S+1)+O(1)$ with 
$Q\ge V/(S+1)$, 
\begin{align}\label{tstvx}
  \pix\party\STV\ls
&
=  \pix\tactic\STV\ls
=  \pix\same\STV\ls
=  \pix\wPSC\STV\ls
=  \pix\PSC\STV\ls
=\frac{\ell}{S+1}.
\end{align}
\end{theorem}

\begin{proof}
  In the party list case, it is easy to see that
STV becomes equivalent to the quota method with
  quota $Q$,  which
  we denote by $\Qx\gd$. Hence
$\pi\party\STV\ls=\pi\party\Q\gd\ls$, which is given by \eqref{tquot}.
Consequently, recalling also the general inequalities \eqref{same-psc}, to
prove \eqref{tstv}, it remains only to show that
$\pi\PSC\STV\ls\le\pi\STV\party\ls$.
In other words, somewhat informally, we want to show that the worst case 
in the party list case (or, equivalently, for the quota method), also is the
worst case for $\PSCx$.

Thus, consider an instance of $\PSCx$ with a bad outcome, 
and let $\cA$ be the set of $m=|\cA|\ge\ell$ common top candidates for $\cW$;
thus, at most
$\ell-1$ of the candidates in $\cA$ are elected. Let $\cD:=\cE\setminus\cA$,
so $|\cD|\ge S+1-\ell$.

Suppose that during some stage of the counting process, there are exactly
$\ell$ candidates in $\cA$ that have not yet been eliminated, 
and that the next event is that one of these, say $A_1$, is eliminated.
Let $k\le\ell-1$ be the number of candidates from $\cA$ elected so far; 
thus there are $\ell-k>0$  candidates from $\cA$ still remaining (including
$A_1$). 
Let further $m$ be the number of candidates in $\ccA$ that have been elected
so far.
There are also at least $S+1-\ell-m$ remaining candidates in $\ccA$, 
since otherwise all remaining candidates would be
elected, and thus $|\cA\cap\cE|=\ell$,
contrary to our assumption.

Let $x$ be the current number of votes counted for $A_1$. 
Thus $x<Q$, since $A_1$ is not elected.
Since there is still at least one remaining candidate from $\cA$, 
the (remaining) ballots in $\cW$ are all counted for some candidate in $\cA$
(by the assumption in $\PSCx$). Each elected candidate accounts for $Q$
votes, and each of the $\ell-k-1$
remaining candidate besides $A_1$ has currently less than $Q$ votes; hence
\begin{align}\label{ap}
  \WWW \le kQ + (\ell-k-1)Q+x = (\ell-1)Q+x
\end{align}
and  thus
\begin{align}\label{apa}
x\ge \WWW- (\ell-1)Q.
\end{align}
Similarly, 
the $m$ elected candidates in $\ccA$ account for $mQ$ votes, and
the at least $S+1-\ell-m$ 
remaining candidates in $\ccA$ have at least $x$ votes each
(otherwise $A_1$ would not be the next to be eliminated). Thus,
since so far only votes from $\cVW$ are counted for candidates in $\ccA$,
\begin{align}\label{apb}
  |\cVW| &\ge mQ + (S+1-\ell-m)x \ge (S+1-\ell)x
\notag\\&
\ge(S+1-\ell)\bigpar{\WWW-(\ell-1)Q}.
\end{align}
Consequently,
\begin{align}\label{apc}
  V=\WWW+|\cVW|
\ge(S+2-\ell)\WWW - (S+1-\ell)(\ell-1)Q.
\end{align}
and thus, recalling $Q=V/(S+\gd)$,
\begin{align}\label{anna}
  \frac{\WWW}{V} 
\le \frac{1}{S+2-\ell} +\frac{(\ell-1)(S+1-\ell)}{(S+\gd)(S+2-\ell)}
\end{align}
Note that the \rhs{} of \eqref{anna} equals, by simple algebra, 
$\pi\STV\party\ls=\pi\Q\gd\party\ls$ given by \eqref{tquot}.

We have shown that if the number of elected + remaining candidates from
$\cA$ drops below $\ell$, then \eqref{anna} must hold. 
Suppose now that there is a bad outcome such that this does not happen.
Then there is at least one remaining candidate from $\cA$ throughout the
counting, so votes from $\cW$ are only counted for candidates in $\cA$.
Furthermore, all $S$ elected have to reach the quota, since otherwise at the
end all remaining candidates would have been elected. In particular, at
least $S+1-\ell$ candidates from $\ccA$ are elected with $Q$ votes each, all
coming from $\cVW$, and thus
\begin{align}
  \label{vw}
|\cVW|\ge(S+1-\ell)Q.
\end{align}
Hence,
\begin{align}\label{vx}
  \WWW&=V-|\cVW|\le V-(S+1-\ell)Q=(S+\gd)Q-(S+1-\ell)Q
\notag\\&
=(\ell+\gd-1)Q \le\ell Q.
\end{align}
It follows from \eqref{vx} and \eqref{vw} that \eqref{apa} and \eqref{apb}
hold with $x=Q$; thus \eqref{apc} and \eqref{anna} hold in this case too.

We have shown that \eqref{anna} holds for every bad outcome. 
Consequently,
$\pi\STV\PSC\ls\le\pi\party\STV\ls$, which completes the proof of
\eqref{tstv}.

For $\pix\STV\tactic$,
consider first the case $\ell=1$, where by \eqref{tactic-same}
and \eqref{tstv},
\begin{align}\label{elf}
  \pi\STV\tactic\is \le \pi\STV\same\is=\frac{1}{S+1}.
\end{align}
Hence, \eqref{tstvtactic} follows by \refT{Tsplit}\ref{Tsplit=x}.

We immediately obtain as special cases, taking $\ell=1$ or $\gd=1$, 
\eqref{tstv1} and \eqref{tstv2}
with $\pi\tactic$ replaced by $\pix\tactic$. The equalities for 
$\pi\tactic$ in \eqref{tstv1} and \eqref{tstv2} then follow from
\eqref{pixpi} and \eqref{tactic-same}, \cf{} \eqref{elf}.

Finally, if $Q=V/(S+1)+O(1)$, define $\gd:=V/Q-S$, so that $Q=V/(S+\gd)$,
and note that $\gd\le1$ and $\gd=1+o(1)$ as \Vtoo.
Hence, \eqref{tstv} implies, for any of the scenarios there, that for a bad
outcome, $\WWW/V\le \ell/(S+1)+o(1)$ as \Vtoo, and thus
$\pix\STV\ls\le\ell/(S+1)$. A corresponding lower bound follows easily by the
party case, and thus \eqref{tstvx}  follows.
\end{proof}

\begin{proof}[Proof of \refT{TSTVD}]
  Let $\gd:=1/r-S<1$, so $r=1/(S+\gd)$.
Then, \eqref{tstv} holds and implies \eqref{tstvD}.
\end{proof}

\section{\phragmen's and Thiele's  ordered methods}
\label{SPhrTho}

We continue the study of the thresholds
$  \pi\party,
\pi\tactic,
\pi\same,
 \pi\wPSC,
 \pi\PSC
$
for election methods with ordered ballots.
In this section we consider \phragmen's and Thiele's election methods,
see \refApps{APhro} and \ref{ATho};
in \refS{SBorda} we study Borda methods.

\begin{problem}
  Further election methods for ordered ballots are described and studied in
  \eg{} \cite{ElkindEtal2017,FSST,AzizLee}.
Study   
$  \pi\party,
\pi\tactic,
\pi\same,
 \pi\wPSC,
 \pi\PSC
$
for them!
\xfootnote{
Some inequalities for $\pi\PSC$ and $\pi\wPSC$
follow from results in 
\cite{ElkindEtal2017,AzizLee}.
showing whether or not 
certain methods satisfy (weak) PSC and some related criteria, see \refR{RPSC}.
}
\end{problem}

\subsection{Two simple cases}
We note first that
in the party list case, we have the same result as for \phragmen's and
Thiele's methods for unordered ballots, \refT{TPhruparty}.

\begin{theorem}
  \label{TPhroparty}
For \phragmen's and Thiele's ordered methods\kol
\begin{align}\label{tPhroparty}
  \pi\party\Phro\ls=
  \pi\party\Tho\ls=
\pi\party\DH\ls=\frac{\ell}{S+1}.
\qquad\lele.\end{align}
\end{theorem}
\begin{proof}
It is easy to see 
that in the party list case, 
both methods reduce to D'Hondt's method
\cite[Theorem 11.1]{SJV9}.
Hence, the result follows from \eqref{party-DH}.
\end{proof}

Next we note that, as said in \refS{SPSC}, \phragmen's and Thiele's methods
do not satisfy any PSC condition, since they do not eliminate candidates.

\begin{theorem}\label{TPhrPSCbad}
For \phragmen's and Thiele's ordered methods\kol
  \begin{align}
    \pi\Phro\PSC\ls= 
    \pi\Tho\PSC\ls=1,
\qquad\lele.
  \end{align}
\end{theorem}

\begin{proof}
Consider an election where each voter in $\cW$ also is a candidate, and
votes for herself  first, followed by all others in $\cW$ (in any order).
Suppose also that there are $S$ other candidates $\cB$, and $2S$ other
voters $\cVW$, with each candidate in $\cB$ being the first and only name on
2 ballots from $\cVW$. This fits the scenario $\PSCx$, provided $\WWW\ge\ell$,
and obviously both \phragmen's and Thiele's methods will elect $\cB$, which
is a bad outcome. Hence,
\begin{equation}
  \pi\PSC\Phro\ls,
  \pi\PSC\Tho\ls\ge \frac{\WWW}{\WWW+2S}
\end{equation}
for any $\WWW\ge\ell$, and the result follows.
\end{proof}

We will see below (\refT{TPhrowPSC} and \refC{CTho})
that the weak PSC $\pi\wPSC\ls<\ell/S$
is satisfied by \phragmen's method but not
by Thiele's.

\subsection{\phragmen's ordered method}

\phragmen's method has the optimal result
for weak PSC, as well as for the scenarios $\samex$ and $\tacticx$.

\begin{theorem}[\cite{SJV9}]\label{TPhrowPSC}
For \phragmen's ordered method\kol
For $1\le \ell\le S$,
\begin{align}\label{tphrowpsc}
  \pi\party\Phro\ls
&
=  \pi\tactic\Phro\ls
=  \pi\same\Phro\ls
=  \pi\wPSC\Phro\ls
=\frac{\ell}{S+1}.
\end{align}
\end{theorem}

\begin{proof}
The result for $\pi\Phro\party$ is in \refT{TPhroparty}.

We next show the upper bound for $\pi\Phro\wPSC$.
This follows by almost the same proof as for \refT{TsamePhru}.

Let $\cW$ be a set of
voters and $\cA$ a set of candidates as in \refD{DwPSC}.
Thus $|\cA|=\ell$ and each voter in $\cW$ votes for the set $\cA$ in some
order, possibly followed by some other candidates.
Suppose that the outcome is bad, \ie, $k:=|\cA\cap\cE|<\ell$;
in other words, at least one candidate in $\cA$ is not elected.

We use the formulation with loads
in  \refApp{APhro},
and let $t=t\xx \MM$ be the final
maximum load of a ballot.
Let $x_i$ be the final load on ballot $i$, and let, 
as in the proof of \refT{TsamePhru}, the free voting power of the ballot be
$t-x_i$. 

The $k$ elected candidates in $\cA\cap\cE$ together give load $k$. Voters in
$\cW$ have not contributed to the election of any other candidate, and thus
has no load from any other candidate.
Thus \eqref{aga} holds.

Since at least one candidate in $\cA$ is not elected, each ballot in $\cA$
has when the election finishes a current top
candidate that belongs to $\cA\setminus\cE$. 
The total free voting power assigned to each candidate is at most 1, since
otherwise this candidate would have been elected in the last step (if not
earlier), with a lower $t\xx\MM$. Hence, we have, instead of \eqref{agb},
\begin{align}
  \sum_{i\in\cW}(t-x_i)\le|\cA\setminus\cE|=\ell-k,
\end{align}
but this is, combined with \eqref{aga}, enough to yield \eqref{agc}.

The same argument as in the proof of \refT{TsamePhru} now gives
\eqref{agd}--\eqref{agh}, and thus
\begin{align}\label{bgh}
\pi\wPSC\Phro\ls\le\frac{\ell}{S+1}.  
\end{align}

Next, \eqref{bgh}, 
the result for $\pi\party\Phro$, and \eqref{same-psc} 
imply the 
result in \eqref{tphrowpsc} for $\pi\Phro\same$ and $\pi\Phro\wPSC$.

Finally, this shows
by \eqref{tactic-same} that $\pi\Phro\tactic\ls\le\ell/(S+1)$.
Equality holds by \refT{Tsplit}\ref{Tsplit=}.
\end{proof}

\subsection{Thiele's ordered method}\label{SSTho}

Although \refT{TPhroparty} for party lists holds for Thiele's method as well
as for \phragmen's, \refT{TPhrowPSC} does not; the reason is that Thiele's
method invites to tactical voting, where a party may gain seats by a
splitting their votes on different (carefully chosen)  lists.
Cf.\ \refS{SSTha} where the same phenomen is seen for \Than{} for unordered
ballots.
We give first one example from the commission report \cite{bet1913}.

\begin{example} \label{Etactic-o} 
Thiele's ordered method with
$S=2$ seats, and 100 voters voting
\begin{val}
\item [61]$AB$
\item [39]$CD$
\end{val}
It is easy to see that $AC$ are elected. (This is a party list case.)

However, suppose that instead 
the larger party split their votes as follows:  
\begin{val}
  \item [41] $AB$
  \item [20] $B$\quad   (or $BA$)
  \item [39] $CD$
  \end{val}
Then,
the first seat goes to $A$;
for the second seat, $B$  has $41/2+20=40.5$ votes, and beats $C$. Elected:
$AB$.

Thus, if $\cW$ is the set of 39 $CD$ voters, $\cW$ gets no candidate
elected. This is an instance of $\samex$ (and thus of $\wPSCx$) 
for $\ell=1$ and $S=2$
with a bad outcome.
Hence,
\begin{align}
  \pi\Tho\wPSC(1,2) \ge   \pi\Tho\same(1,2)
\ge \frac{39}{100}=0.39 >\frac{\ell}{S+1}.
\end{align}
\end{example}

We can calculate the thresholds  
exactly for Thiele's
ordered method. We begin with some preliminaries.
For convenience, we let in the remainder of this section, as in \refS{SSTha},
the ``number of votes'' be arbitrary positive real
numbers; see \refRs{Rreal} and \ref{Rhomo}, 
and note that this does not affect the results since the method is homogeneous.

When, as in the next lemma, we do not specify the number of seats $S$, we
just assume that it is large enough.

\begin{lemma}\label{LA}
  Let $\tha_n$, $n\ge1$, 
be the smallest total number of (real-values) votes in an election
  by Thiele's ordered method where there are $n$ candidates $A_1,\dots,A_n$, 
and no others, and each of them is
  elected with score at least $1$.

  \begin{romenumerate}
  \item \label{LAW}
Consider an election by Thiele's ordered method where 
there are $n$ candidates $A_1,\dots,A_n$, 
and perhaps others, and each $A_i$ is 
  elected with score at least $t\ge0$.
Let $W$ be the number of voters that have voted for at least one of
$A_1,\dots,A_n$. 
Then $W\ge \tha_n t$.
\item \label{LAab}
Define also a sequence $\thb_n$, $n\ge1$, by the recursion
\begin{equation}
  \label{b1}
  \thb_n:=1-\sum_{i=1}^{n-1}\frac{\thb_i}{n+1-i},
\qquad n\ge1,
\end{equation}
or, equivalently,
\begin{equation}
  \label{b2}
  \sum_{i=1}^{n}\frac{\thb_i}{n+1-i}=1,
\qquad n\ge1.
\end{equation}
Then 
\begin{align}\label{ab}
  \tha_n=\sumin \thb_i.
\end{align}
Furthermore, $\tha_n$ and $\thb_n$ are given by the Taylor coefficients
\begin{align}
  \tha_n&=[z^n] \frac{z^2}{-(1-z)^2\log(1-z)},\label{taylora}
\\
  \thb_n&=[z^n] \frac{z^2}{-(1-z)\log(1-z)}.\label{taylorb}
\end{align}
Asymptotically, as \ntoo,
\begin{align}
  \tha_n &\sim \frac{n}{\log n}, \label{fja}
\\
  \thb_n &\sim \frac{1}{\log n}.\label{fjb}
\end{align}
  \end{romenumerate}
\end{lemma}

The first numbers $a_n$ and $b_n$ are given in \refTab{tab:thabc}.
\begin{table}[h!]
  \centering
\begin{tabular}{c|cccccc}
  $n$ & 1 & 2 & 3 & 4 & 5 & 6\\
\hline
\rule{0pt}{2.7ex}  
$\tha_n$ 
& 1 & $\frac{3}{2}$ & $\frac{23}{12}$ & $\frac{55}{24}$
& $\frac{1901}{720}$ & $\frac{4277}{1440}$
\\ \rule{0pt}{3ex}  
$\thb_n$ & 1 & $\frac{1}2$ & $\frac{5}{12}$ & $\frac{3}{8}$ & $\frac{251}{720}$
& $\frac{95}{288}$ 
\\ \rule{0pt}{3ex}  
$\thc_n$ & 1 & 2 & 4 & 6 & 9 & 12
\\[2pt]
\end{tabular}
  \caption{The numbers $\tha_n$, $\thb_n$ and $\thc_n$ in \refLs{LA} and
   \ref{LC}.} 
  \label{tab:thabc}
\end{table}

Before proving \refL{LA}, we prove a technical result.

\begin{lemma}\label{LB}
  The numbers $\thb_n$ defined by \eqref{b1} form a strictly decreasing
  sequence 
with $0<\thb_n\le 1$.
\end{lemma}

\begin{proof}
  We first give a simple proof by induction that $\thb_n>0$ (although this also
  follows from the argument below). Thus, assume $n\ge1$ and
  $\thb_1,\dots,\thb_n>0$. 
Then, by \eqref{b2},
\begin{equation}
    \sum_{i=1}^{n}\frac{\thb_i}{n+2-i}<
  \sum_{i=1}^{n}\frac{\thb_i}{n+1-i}=1,
\end{equation}
and thus \eqref{b1} yields $\thb_{n+1}>0$, completing the induction proof.

Thus $\thb_n>0$, and hence $\thb_n\le1$ by \eqref{b1}.

To show that $\thb_n$ is decreasing is (as far as we know) less elementary.
Let 
\begin{equation}\label{bz}
  B(z):=\sum_{n=1}^\infty \thb_nz^n
\end{equation}
be the generating function of $\thb_n$. Since $-\log(1-z)=\sum_{m\ge1} z^m/m$,
\eqref{b2} is equivalent to
\begin{equation}
  -\log(1-z)B(z)=\sum_{n=2}^\infty z^n = \frac{z^2}{1-z}
\end{equation}
and thus
\begin{equation}\label{b4}
  B(z)=-\frac{z^2}{(1-z)\log(1-z)}.
\end{equation}
Let $\thb_0:=0$ and
\begin{equation}
  C(z):=\sum_{n=1}^\infty \bigpar{\thb_{n-1}-\thb_n}z^n 
= (z-1)B(z)=\frac{z^2}{\log(1-z)}.
\end{equation}
Note that $C(z)$ is (extends to) an analytic function in
$\cD:=\bbC\setminus[1,\infty)$. 
Thus, Cauchy's integral formula yields, for any simple closed curve $\gamma$ in
$\cD$ that encircles 0 in the positive direction,
\begin{align}
  \thb_{n-1}-\thb_n=\frac{1}{2\pi\ii}\oint_\gam C(z)\frac{\dd z}{z^{n+1}}
=\frac{1}{2\pi\ii}\oint_\gam\frac{z^{1-n}}{\log(1-z)}\dd z.
\end{align}
Let $\eps>0$ and $R>1$, and let $\gam$ consist of: an arc on the circle
$|z|=\sqrt{R^2+\eps^2}$, going from $R+\eps\ii$ to $R-\eps\ii$ in the
positive direction;
the straight line from $R-\eps\ii$ to $1-\eps\ii$;
the semicircle $\set{1+\eps e^{-\ii t}: \frac{\pi}2\le t\le \frac{3\pi}2}$;
the straight line from $1+\eps\ii$ to $R+\eps\ii$.
For $n\ge2$, we then let first $\eps\to0$ and then $R\to\infty$, and obtain,
using simple estimates for the circular parts, 
\begin{align}
  \thb_{n-1}-\thb_n
&=\frac{1}{2\pi\ii}
\biggpar{\int_\infty^1\frac{x^{1-n}\dd x}{\log(x-1)+\pi\ii}
+ \int_1^\infty\frac{x^{1-n}\dd x}{\log(x-1)-\pi\ii}}
\notag\\&=
\int_1^\infty\frac{x^{1-n}\dd x}{\log^2(x-1)+\pi^2}.
\label{otk}
\end{align}
This proves that $\thb_{n-1}-\thb_n>0$ for $n\ge2$, which completes the proof.
\end{proof}

\begin{proof}[Proof of \refL{LA}.]
We begin by giving a good strategy; we will later show that it is optimal.
(The strategy used by the $AB$ party
in  \refE{Etactic-o} is an example, apart from rounding to
integers.)

Consider an election with $b_i$ votes for $A_iA_{i+1}\dotsm A_n$, for each
$i\le n$. (Note that $b_i>0$ by \refL{LB}.)
The total number of votes is $\thax_n:=\sumin b_i$.
We claim that, by induction, $A_1,\dots,A_n$ will be elected in this order
with score 1 each. (Assuming $S\ge n$.)
In fact, if $A_1,\dots,A_{k-1}$ already have been elected, then $A_k$ has
score, using \eqref{b2},
\begin{align}
  \sum_{i=1}^k \frac{b_i}{k-i+1}=1,
\end{align}
while if $j>k$, then $A_j$ has score $b_j<1$. Hence $A_k$ is elected
next, which verifies the induction step.
Thus $\tha_n\le\thax_n$.
 
Conversely, consider an arbitrary election as in \ref{LAW}.
Let $\cW$ 
be the set of voters that have voted for some $A_i$.
Consider a ballot $\gs$. Suppose that there are $m$ candidates 
$A_i\in\cA$ that are elected with the help of $\gs$, \ie{} the candidates
$A_i\in\gs\cap\cE$ 
such that all candidates before $A_i$ on $\gs$ are elected before $A_i$.
Let these candidates by $A_{i_1}, \dots,A_{i_m}$, with $1\le i_1<i_2<\dots<
i_m$, and let their positions on the ballot $\gs$ be $j_1,\dots,j_m$. Note
that $1\le j_1<\dots<j_m$, since $A_1,\dots,A_n$ are elected in order.
Define the weight $w(\gs,i)$ for $i\in\nn$ by
\begin{align}
  w(\gs,i_k)&:=\frac{1}{j_k}\thb_{n+1-i_k},
\\
w(\gs,i)&:=0, \qquad i\notin\set{i_1,\dots,i_k}.
\end{align}
If $i=i_k$, then the ballot $\gs$ is worth $1/j_k$ votes when $A_{i}$ is
elected.
Thus, for each $i\in\nn$, $w(\gs,i)$ equals $\thb_{n+1-i}$ times the score
contributed by $\gs$ to the election of $A_i$;
hence, $\sum_{\gs}w(\gs,i)$ is $\thb_{n+1-i}$ times the score of $A_i$
when elected, and thus,
noting that $w(\gs,i)=0$ unless $\gs\in\cW$,
\begin{align}\label{kia}
  \sum_{\gs\in\cW}w(\gs,i) 
= \sum_{\gs\in\cV}w(\gs,i)\ge \thb_{n+1-i}.
\end{align}
On the other hand, returning to the ballot $\gs$, since $i_1<\dots<i_m\le
n$, 
we have $i_k\le n-m+k$ and thus $n+1-i_k\ge m+1-k$. Hence, using \refL{LB},
\begin{align}
  \thb_{n+1-i_k}\le \thb_{m+1-k}.
\end{align}
Furthermore, $j_k\ge k$.
Hence, the total weight on $\gs $ is, using also \eqref{b2},
\begin{align}\label{gustaf}
  \sum_{i=1}^n w(\gs,i) 
=  \sum_{k=1}^m w(\gs,i_k) 
=  \sum_{k=1}^m \frac{1}{j_k}\thb_{n+1-i_k}
\le  \sum_{k=1}^m \frac{1}{k}\thb_{m+1-k}=1.
\end{align}
Consequently, summing \eqref{gustaf} over $\gs\in\cW$ 
and \eqref{kia} over $i$,
\begin{align}\label{tom}
 W=\WWW=\sum_{\gs\in\cW}1 \ge \sum_{\gs\in\cW}\sumin w(\gs,i)\ge \sumin  \thb_{n+1-i}
=\thax_n.
\end{align}
In particular, the definition of $\tha_n$ shows $\tha_n\ge\thax_n$, and thus
$\tha_n=\thax_n$, which is \eqref{ab}.
(Thus, the strategy given above is optimal.)
Hence, \eqref{tom} also shows \ref{LAW}, using homogeneity.

Finally, \eqref{taylorb} follows by \eqref{bz} and \eqref{b4}, and
\eqref{taylora} by \eqref{ab} and \eqref{b4}.
The asymptotic formulas \eqref{fja}--\eqref{fjb} follow by
\eqref{taylora}--\eqref{taylorb} and singularity analysis, see
\cite[Theorem VI.2]{FlajoletS}. 
\end{proof}

\begin{remark}
  By summing \eqref{otk} for $n\ge i+1$, and then summing again for
  $i=1,\dots k$, we obain the integral formulas
\begin{align}
\thb_n
&=
\int_1^\infty\frac{x^{1-n}\dd x}{(x-1)\bigpar{\log^2(x-1)+\pi^2}},
\\
\tha_n
&=
\int_1^\infty\frac{x(1-x^{-n})\dd x}{(x-1)^2\bigpar{\log^2(x-1)+\pi^2}}.
\end{align}
It is also possible to derive \eqref{fja}--\eqref{fjb} from these.
\end{remark}

\begin{lemma}
  \label{LC}
  Let $\thc_n$, $n\ge1$, 
be the smallest real number such that
in any  election
  by Thiele's ordered method where 
there are $n$ candidates $A_1,\dots,A_n$, 
and each voter votes for all of them in some order,
each $A_i$ is  elected with score at least $1$.
Then
\begin{align}\label{lc}
  \thc_n=
\Bigl\lfloor{\frac{n+1}{2}}\Bigr\rfloor\Bigl\lceil{\frac{n+1}{2}}\Bigr\rceil
=
\begin{cases}
  m^2, & n=2m-1,
\\
m(m+1), & n=2m.
\end{cases}
\end{align}
\end{lemma}
The first numbers $\thc_n$ are given in \refTab{tab:thabc}.

\begin{proof}
  Consider any election with $n$ candidates  $A_1,\dots,A_n$, 
and each of the $V$ voters voting for all of them in some order.
 When $k\ge0$ of the candidates
have been elected,
  there are $n-k$ candidates left, and each ballot is worth at least
  $1/(k+1)$, so at least one candidate has score $\ge V/\bigpar{(k+1)(n-k)}$.

Conversely, for a given $k \in \set{0,\dots,n-1}$, consider an election
where each voter votes for $A_1,\dots,A_k$ in order followed by the
remaining candidates in some order, with 
each of the remaining candidates in place $k+1$ by 
$V/(n-k)$ voters.
Then, in round $k+1$, each remaining candidate has score
$V/\bigpar{(k+1)(n-k)}$. 

It follows that $\thc_n$ is given by
\begin{align}
  \min_{0\le k<n} \frac{\thc_n}{(k+1)(n-k)}=1,
\end{align}
and thus
\begin{align}
\thc_n=  \max_{0\le k<n} (k+1)(n-k).
\end{align}
The maximum is attained for $k=\floor{(n-1)/2}$, and \eqref{lc} follows.
\end{proof}

\begin{lemma}\label{LD}
  Consider an election by Thiele's ordered method where 
there are $n$ candidates $A_1,\dots,A_n$, 
and perhaps others, and each $A_i$ is 
  elected with score at least $t\ge0$.
Add an arbitrary set on new votes (and perhaps new candidates).
Then, throughout the counting, as long as not  all candidates $A_i$
have been elected, at least one of them has score $\ge t$.
\end{lemma}
\begin{proof}
This is perhaps not completely obvious
since  the order in which $A_1,\dots,A_n$ are elected might be changed.
Thus, 
consider some round in the new election (after adding ballots).
If $A_j$ is the first candidate in $\cA:=\set{A_i}$ not yet elected, 
so $A_1,\dots,A_{j-1}$ (and possibly some others) already are elected,
then $A_j$ has at least the same score as when $A_j$ was elected in the
original election (without extra ballots), which is $\ge t$.
Hence, until all of $\cA$ have been elected,
there is always at least one of them with score $\ge t$.
\end{proof}

\begin{theorem}\label{TTho}
Let\/ $\tha_n$ be as in \refL{LA} and $\thc_n$ as in \refL{LC}. 
Then, 
for Thiele's ordered method\kol
For $1\le \ell\le S$,
  \begin{align}
\pix\tactic\Tho\ls&= \frac{\tha_{\ell}}{\tha_\ell+\tha_{S+1-\ell}},
\label{tthotactic}
\\
\pi\same\Tho\ls&= \frac{\ell}{\ell+\tha_{S+1-\ell}},\label{tthosame}
\\
\pi\wPSC\Tho\ls&= \frac{\thc_{\ell}}{\thc_\ell+\tha_{S+1-\ell}}.\label{tthowpsc}
  \end{align}
\end{theorem}
Some numerical values
are given in \refTabs{tab:Thotactic}--\ref{tab:Thowpsc} in \refApp{Anum}.

\begin{proof}
We show the three equation using separate but similar arguments.
Let $\cA=\set{A_1,\dots,A_\ell}$ be a set of $\ell$ candidates,
and let $\cW$ be a set of $W:=\WWW$ voters.

$\pi\same$:
Suppose that the set $\cW$ of
voters vote on the same list $\gs$, beginning with 
$A_1\dotsm A_\ell$.
Suppose also that the outcome is bad; then not all of $A_1,\dots,A_\ell$ are
elected. We may then assume that $\gs=A_1\dotsm A_\ell$, since deleting
later names will not affect the counting.
Furthermore, at least $S-\ell+1$ other candidates
$B_1,\dots,B_{S+1-\ell}$ are elected.
Throughout the counting, since not all of $A_1,\dots,A_\ell$
are elected, 
at least one of them (the first not yet elected) 
has a score of at least $W/\ell$. Consequently,
$B_1,\dots,B_{S+1-\ell}$ are all elected with a score of $\ge W/\ell$.
Since only voters from $\cVW$ vote for any $B_j$,
\refL{LA}\ref{LAW} applied to $\set{B_j}_1^{S+1-\ell}$ yields
\begin{equation}\label{cp1}
V-W=  |\cVW|\ge \tha_{S+1-\ell} \frac{W}\ell,
\end{equation}
which is equivalent to
\begin{equation}\label{sinbad}
  \frac{W}{V} \le \frac{\ell}{\ell+\tha_{S+1-\ell}}.
\end{equation}

Conversely, consider an election with 
$\ell$ votes on the list $A_1\dotsm A_\ell$, and
$\tha_{S+1-\ell}$ votes on (only) $B_1,\dots,B_{S+1-\ell}$, distributed 
(using the strategy in the proof of \refL{LA}) such that
that each  $B_j$ is elected with score $\ge1$ if
enough seats are distributed.
Then, as long as not all
$B_1,\dots,B_{S+1-\ell}$ have been elected, some $B_j$ will have
score $\ge1$ (in fact, exactly 1).
If $A_1,\dots,A_{\ell-1}$ have been elected, then $A_\ell$ will have score
$1$; thus there will be a tie for each of the remaining seats
and it is possible that $B_1,\dots,B_{S+1-\ell}$ are
elected and $A_\ell$ not, a bad outcome.
This example shows that equality may hold in \eqref{sinbad} for a bad
outcome, and thus 
\eqref{tthosame} holds.

 $\pix\tactic$:
Let the $W$ voters in $\cW$ vote according to
the strategy in  \refL{LA} (and its proof), scaling the number of votes by
$W/\tha_\ell$.
Then, using \refL{LD}, as long as not all $A_i$ are elected,
at least one of them has score $\ge W/\tha_\ell$.
Suppose that the outcome is bad.
Then, at least $S-\ell+1$ other candidates
$B_1,\dots,B_{S+1-\ell}$ are elected.
Furthermore,  throughout the counting, some $A_i$ has score $\ge W/\tha_\ell$;
thus every elected candidate, and in particular every $B_j$, 
is elected with score $\ge W/\tha_\ell$.
Similarly to \eqref{cp1}, it now follows by
\refL{LA}\ref{LAW} that
\begin{equation}\label{cp2}
V-W=  |\cVW|\ge \tha_{S+1-\ell} \frac{W}{\tha_\ell},
\end{equation}
which is equivalent to
\begin{equation}\label{sinbadd}
  \frac{W}{V} \le \frac{\tha_\ell}{\tha_\ell+\tha_{S+1-\ell}}.
\end{equation}

Conversely, suppose that $W<\tha_\ell$ and $V=\tha_\ell+\tha_{S+1-\ell}$, so
$|\cVW|>\tha_{S+1-\ell}$. Suppose that the voters in $\cVW$ vote on
$B_1,\dots,B_{S+1-\ell}$ (and no others)
using the strategy in \refL{LA} and its proof; thus, using also 
\refL{LD}, no matter
how the voters in $\cW$ vote, as long as not $B_1,\dots,B_{S+1-\ell}$ are
not all elected,
at least one of them will have score $\ge1$.
Hence, if the outcome is good, so $A_1,\dots,A_\ell$ all are elected, then 
each of them has to be elected with score $\ge1$.
However, this contradicts \refL{LA}\ref{LAW}, since $W<\tha_\ell$.
Consequently, for such $W$, the outcome may be bad for any strategy chosen
by $\cW$.
This and \eqref{sinbadd} yield \eqref{tthotactic}.

$\pi\wPSC$:
Suppose that all voters in $V$ vote on lists beginning with
$\cA$ in some order. 
Then, by \refLs{LC} and \ref{LD} (and homogeneity),
as long as not all $A_i$ are elected, at least one of them has score $\ge
W/\thc_\ell$. By the same argument as just given for $\pix\tactic$ (with
$\tha_\ell$ replaced by $\thc_\ell$),
it follows that if the outcome is bad, then
\begin{equation}\label{sinbad3}
  \frac{W}{V} \le \frac{\thc_\ell}{\thc_\ell+\tha_{S+1-\ell}}.
\end{equation}

Conversely, consider an election with $V=\thc_\ell+\tha_{S+1-\ell}$ 
votes; 
$\thc_\ell$ votes on
$A_1,\dots,A_\ell$ distributed as in the extremal case in \refL{LC} and its
proof, and 
$\tha_{S+1-\ell}$ votes on $B_1,\dots,B_{S+1-\ell}$ distributed 
according to the strategy in \refL{LA} and its proof.
Then, in some round during the counting, not all $A_i$ have been elected and
the leading $A_i$ has score $1$; hence all remaining candidates in $\cA$ have
scores $\le1$. Furthermore, unless all $B_j$ already are elected (in which
case the outcome is bad), at least one of them has score $\ge1$.
Hence, assuming that $\cA$ loses every time there is a tie, 
a possible outcome is that all $B_1,\dots,B_{S+1-\ell}$ are elected before
any further $A_i$, and thus the outcome is bad.
Hence, equality may hold in \eqref{sinbad3} for a bad
outcome, and thus 
\eqref{tthowpsc} holds.
\end{proof}

\begin{corollary}\label{CTho}
For Thiele's ordered method and $\ell=1$\kol
With\/ $\tha_S$ as in \refL{LA},
  \begin{align}
\pix\tactic\Tho\is=
\pi\tactic\Tho\is=
\pi\same\Tho\is=
\pi\wPSC\Tho\is= \frac{1}{1+\tha_{S}},
\qquad S\ge1. 
\label{cthowpsc}
  \end{align}
Asymptotically, this is $\sim\log S/S$ as \Stoo.
\end{corollary}
\begin{proof}
  By \eqref{tthotactic}--\eqref{tthowpsc}, since $\tha_1=\thc_1=1$;
the result for $\pi\tactic$ follows using \eqref{pixpi} and \eqref{tactic-same}.
The asymptotic formula follows from \eqref{fja}.
\end{proof}

In particular, Thiele's ordered method does not satisfy the weak PSC
$\pi\wPSC\ls<\ell/S$, see \refS{SPSC}, not even for $\ell=1$ and $S$ large.
(In fact, not for $S\ge3$, as easily follows from \eqref{tthowpsc} and
\refL{LA}.)
Consequently, Thiele's method does not satisfy
``Solid coalitions''  in \cite{ElkindEtal2017}, see \refR{RPSC}.

\begin{example}\label{ETho}
Since $\tha_1=1$ and $\tha_2=3/2$, \refT{TTho} yields
$\pix\Tho\tactic(1,2) =\xfrac{2}{5}$ and   $\pix\Tho\tactic(2,2) =\xfrac{3}{5}$.
See \refE{Etactic-o}, which
is the extremal case with a perturbation to avoid ties.
This also shows that no strategy can help the smaller party in
\refE{Etactic-o}.
\end{example}

\begin{example}\label{ETho2}
\refT{TTho} yields
$\pi\same\Tho(2,3) =\xfrac{4}{7}\doteq0.571>0.5$.
Hence, in an election with 3 seats, even if a majority votes for the same
list, they can fail to obtain a majority of the seats.
A concrete example (modified from \cite{SJV9}), 
which essentially uses the strategy above, is:
\begin{val}
  \item [55]$ABC$
  \item [30]$XYZ$
  \item [15]$YZX$
  \end{val}
$A$ gets the first seat, but then the second
and third seats go to $X$ and $Y$ with 30 votes each.
Thus $A,X,Y$ are elected, and the $ABC$ party gets only 1 seat,
in spite of a majority of the votes.
\end{example}


\section{Borda (scoring) methods}\label{SBorda}

Finally, we consider Borda methods, \refApp{ABorda}. 
Recall than $H_n:=\sumin 1/i$, the harmonic number. 

\begin{theorem}\label{TBorda}
Let\/ $\Bordax{\ww}$ be the Borda method with weights $\ww=(w_k)_1^\infty$,
and let
\begin{align}
  \bw_k:=\frac{1}{k}\sum_{i=1}^k w_i.
\end{align}
Then, for $\lele$,
  \begin{align}
\pix\tactic\Borda{\ww}\ls&= \frac{\bw_{S+1-\ell}}{\bw_\ell+\bw_{S+1-\ell}},
\label{tbordatactic}
\\
\pi\same\Borda{\ww}\ls&= 
\pi\wPSC\Borda{\ww}\ls=  \frac{\bw_{S+1-\ell}}{w_\ell+\bw_{S+1-\ell}}.
\label{tbordasame}
  \end{align}
In particular, for
the harmonic Borda method $\Bordax{1/k}$:
 For $\lele$,
  \begin{align}
\pix\Borda{1/k}\tactic\ls&
= \frac{\ell H_{S+1-\ell}}{(S+1-\ell)H_\ell+\ell H_{S+1-\ell}},
\label{tbordatactich}
\\
\pi\same\Borda{1/k}\ls&= 
\pi\Borda{1/k}\wPSC\ls
= \frac{\ell H_{S+1-\ell}}{S+1-\ell+\ell H_{S+1-\ell}}.
\label{tbordasameh}
  \end{align}
Furthermore,
\begin{align}
  \pi\Borda{1/k}\party\ls
= \frac{\ell }{S+1}.
\label{tbordaparty}
\end{align}
\end{theorem}
Some numerical values for the harmonic Borda method
are given in \refTabs{tab:Btactic} and \ref{tab:Bsame} in \refApp{Anum}.

\begin{proof}
It is easy (more or less trivial) to see that
\refLs{LA}\ref{LAW}, \ref{LC} and \ref{LD} hold  also for $\Bordax{\ww}$
with $\tha_n$ replaced by $1/\bw_n$
and $\thc_n$ is replaced by $1/w_n$.

The result \eqref{tbordatactic}--\eqref{tbordasame} 
then follows by the same proof as for \refT{TTho}, where for
$\pi\same$ we also replace $\ell$ by $1/w_\ell$.

For the harmonic case $w_k=1/k$, we have $\bw_k=H_k/k$, and
\eqref{tbordatactich}--\eqref{tbordasameh} follow.
Finally, in the party list case, the harmonic Borda method is equivalent to
\DHn's method, and thus \eqref{tbordaparty} follows from \refT{Tdiv}.
\end{proof}

\begin{corollary}  \label{CBorda}
For the harmonic Borda method with $\ell=1$\kol
  \begin{align}
\pix\tactic\Borda{1/k}\is&
=
\pi\tactic\Borda{1/k}\is
=
\pi\same\Borda{1/k}\is= 
\pi\wPSC\Borda{1/k}\is
\nonumber\\&
= \frac{H_{S}}{S+ H_{S}}
,\qquad S\ge1.
\label{cborda}
  \end{align}
Asymptotically, this is $\sim\log S/S$ as \Stoo.
\end{corollary}

\begin{proof}
  By \eqref{tbordatactich}--\eqref{tbordasameh}, since $H_1=1$;
the result for $\pi\tactic$ follows using \eqref{pixpi} and \eqref{tactic-same}.
The asymptotic formula follows since $H_S\sim \log S$ as \Stoo.
\end{proof}

Note that $\tha_n\sim H_n/n$ as \ntoo, and thus the asymptotics in, for
example, \refCs{CTho} and \ref{CBorda} are the same, but exact values 
differ.

\begin{example}\label{EBorda}
Since $H_1=1$ and $H_2=3/2$, 
$\pix\Borda{1/k}\tactic(1,2) =\xfrac{3}{7}$ and   
$\pix\Borda{1/k}\tactic(2,2) =\xfrac{4}{7}$.
Cf.\  \refE{ETho}.
\end{example}


Note that the Borda method with weights $\ww=(1,0,0,\dots)$ is equivalent to
SNTV; more generally, 
the Borda method with weights $\ww=(1,\dots,1,0,\dots)$ with $L$ 1's is
equivalent to LV($L$); in particular, $L=S$ yields BV.
Consequently, the results in \refTs{TBV}, \ref{TSNTV} and \ref{TLV} follow from
\refT{TBorda}.

\begin{acks}
This work was partly carried out in spare time during a visit to 
Churchill College in Cambridge and
the 
Isaac Newton Institute for Mathematical Sciences
(EPSCR Grant Number EP/K032208/1), 
partially supported by a grant from the Simons foundation, and
a grant from
the Knut and Alice Wallenberg Foundation.

I thank Markus Brill for interesting discussions about, in particular,
EJR, JR and PJR, 
which inspired the present paper.
  \end{acks}

\appendix

\section{Election methods}\label{Amethods}

We give here brief descriptions of the election methods discussed above.
For further details, other election methods, political aspects and examples
of actual use, see \eg{} 
\cite{Farrell,IPU,SJV6,SJV9,Politics,Pukelsheim}.
We give also the abbreviation used in formulas in the present paper.

Note that many election methods are known under several different names. We
give a few synonyms, but omit many others.

\subsection{Election methods with party lists}\label{Alist}
Among the proportional election methods, the ones that
are most often used are party list methods, where the voter votes for a
party and the seats are distributed among the parties according to their
numbers of votes. (The seats obtained by a party then are distributed to
candidates by some method; many versions are used in practice.)

Conversely, most (but not all) party list methods that are used in practice,
including the ones below, are proportional, so that each party gets a
proportion of the  seats that approximates its proportion of votes.
\xfootnote{\label{fnEst}%
Examples that are not proportional are the divisor methods with
  divisors $d(n)=n^{0.9}$ used in Estonia, and $d(n)=2^{n-1}$ used in Macau,
  see \cite{SJV6}.}

Many different list methods are described in detail in \eg{} \citet{Pukelsheim};
see also \citet{BY} for the 
mathematically (but not politically) equivalent problem of
allocating the seats in the US House of Representatives proportionally among
the states.
\xfootnote{There, until 1941, the method was decided after each census,
so the choice of method was heavily influenced by its result.
Moreover, the number of seats was not fixed in advance and thus
also open to negotiations.
}

There are two major types of party list methods: \emph{divisor methods} and
\emph{quota methods}.

\subsubsection{Divisor methods}\label{Adivisor}
A divisor method is defined by a sequence $d(1),d(2),\allowbreak d(3)\allowbreak\dots$
of \emph{divisors}.
\xfootnote{Warning: 
Some authors denote our $d(n+1)$ by $d(n)$, thus starting the  
sequence with $d(0)$.} 
In the traditional formulation,
seats are allocated
sequentially. 
For each seat, a party that so far has obtained $s$ seats gets its number
of votes divided by $d(s+1)$; 
the party with the highest quotient gets the seat.
\xfootnote{\label{fDivisor}%
A different, but equivalent, formulation is that
the number of seats for each party is obtained by selecting a number $D$
and then giving a party with $v_i$ votes 
$s_i\ge0$ seats where
$d(s_i)\le v_i/D \le d(s_i+1)$ (with $d(0):=0$),
where $D$ is chosen such that the total number of seats given to the parties
is $\MM$. (This can be regarded as a special rounding
rule used to round $v_i/D$ to
an integer $s_i$.)
See \citet{Pukelsheim} for a details and examples.
Similar formulations have also been used in the United States for
allocating the seats in the House of Representatives to the states
\cite{BY}.
In this version, the number $D$ is called \emph{divisor}.
}
The two most important divisor methods are described in the next
subsubsections.
They are both of the linear form 
\begin{equation}
  \label{dlinear}
d(n)=n-1+\gam
\end{equation}
 for some
constant $\gam$; several other choices of $\gam$ 
(usually but not always with $\gam\in\oi$) have also been 
used or suggested, for example \emph{Adams's method} with $\gam=0$,
and there are also divisor methods with other (non-linear) sequences of
divisors, see \eg{} \refFn{fnEst}.
We let $\Divx\gam$ denote the linear divisor method given by
\eqref{dlinear};
for simplicity we do not consider other divisor methods 
(except in \refR{Rjamkad}).

\subsubsection{D'Hondt's method\abbrev\DHx}\label{ADHondt}

This is the divisor method with the sequence of divisors $1,2,3,\dots$,
i.e. $d(n)=n$; thus this is the method \eqref{dlinear} with $\gam=1$.
It was
proposed by Victor D'Hondt in 1878 \cite{DHondt1878,DHondt1882}.
D'Hondt's method is (almost)
equivalent to \emph{Jefferson's method},
proposed by Thomas Jefferson in 1792 for 
allocating the seats in the US House of Representatives, 
see \cite{BY}.
\xfootnote{\label{fUSA}%
Jefferson's method is formulated
as in Footnote \ref{fDivisor}, with rounding downwards, which is the same as
\DHn's method, and nowadays \citat{Jefferson's method} is used for this divisor
method, equivalent to \DHn's.
However, historically this is not quite accurate:
for the allocation of seats in the House of Representatives, 
the total number was not fixed in
advance but decided by the outcome once a suitable divisor $D$ was chosen
by Congress,
which makes the method as used by Jefferson different from \DHn's, where the
number of seats is given in advance.
}

\subsubsection{Sainte-Lagu\"e's method\abbrev\StLx}\label{AStL}
This is the divisor method with the sequence of divisors $1, 3, 5,\dots$,
\ie, $d(n)=2n-1$. Equivalently, we may take $d(n)=n-1/2$;
thus this is the method \eqref{dlinear} with $\gam=1/2$.
The method was proposed in 1910 by \citet{StL}.
It is (almost) equivalent to \emph{Webster's method},
proposed by Daniel Webster in 1832 
for
allocating the seats in the US House of Representatives, 
see \cite{BY}.
\xfootnote{
Webster's method is formulated
as in Footnote \ref{fDivisor}, with rounding to the nearest integer.
However, Footnote \ref{fUSA} applies here too.}

\subsubsection{Quota methods}\label{Aquota}
In a quota method, first a \emph{quota} $Q$ is calculated; this is roughly the
number of votes required for each seat, and different quota methods differ
in the choice of $Q$. The two main quotas that are used
are  the \emph{Hare quota} $V/\MM$ and the \emph{Droop quota} $V/(\MM+1)$
\cite{Droop}, where $V$
is the total number of (valid) votes and $\MM$ is the number of seats
(\cf{} \refR{Rquota});
in practice,
the quota is often rounded to an integer, either up, or  down, or to the nearest
integer (see \eg{} \cite{Pukelsheim} and \cite{SJV6} for various examples with
different, or no,  rounding).
Rounding is mathematically a nuisance, leading to various problems, but has
in practice usually no effect. In the present paper, we consider the ideal
versions without rounding, but see \refR{RQrounding}.

The quota method  gives a party with $v_i$ votes first $\floor{v_i/Q}$
seats; the remaining seats, if any, are given
to the parties with largest remainder in these divisions.
In other words, we find a threshold $t\in\oi$ such that
we round $v_i/Q$ upwards if the remainder is larger than $t$.
Mathematically, this means that we find $t$ and integers $s_i$ (the number
of seats for party $i$)  such that
\begin{equation}\label{qm}
  s_i-1+t \le \frac{v_i}Q  \le s_i+t.
\end{equation}
The (traditional) description above with remainders assumes that 
\begin{equation}
\sum_i\floor{v_i/Q}\le S\le \sum_i\ceil{v_i/Q},   
\end{equation}
and thus that we can find
$t\in\oi$ such that \eqref{qm} holds and $\sum_i s_i=S$.
This is always the case with the Hare and Droop quotas, 
and more generally with any quota in the
interval $[V/(S+1),V/S]$. (Except a case with
ties for the Droop quota, but then \eqref{qm} still works with $t=1$.)
Mathematically, the method works with any quota $Q>0$, taking \eqref{qm} as
the definition together with $\sum_is_i=S$ and letting $t$ be
arbitrary real (and requiring $s_i\ge0$ with a modification of \eqref{qm} if
$t>1$). It is easily shown that if $Q\ge V/(S+1)$, then we can take $t\le
1$, and if $Q\le V/S$, then we can take $t\ge0$.

As said above, there are two main quotas used, and thus two main quota methods. 

\subsubsection{Method of largest remainder  (Hare's method)\abbrev\vkx}
\label{Avalkvot}
The quota method using the Hare quota $V/S$.

\subsubsection{Droop's method\abbrev\Droopx}\label{ADroop}
The quota method using the Droop quota $V/(S+1)$.

\subsection{Election methods with unordered ballots}\label{Au}

In these methods, each ballot contains a list of names, but their order is
ignored. 
The number of candidates on each ballot is, depending on the method,
either arbitrary or limited to at most (or exactly) a given number, for
example $\MM$, the number to be elected.

\subsubsection{Block Vote\abbrev\BVx}\label{ABV}
\relax{Each voter votes for at most $\MM$ candidates. The $\MM$ candidates
  with the largest numbers of votes are elected.}
\xfootnote{One version  requires each voter to vote for exactly $\MM$
  candidates. For our purposes this makes no difference, see \refR{Rdummy}.
}
Also called \eg{}
\emph{Multi-Member Plurality}; \emph{Plurality-at-Large}.

The special case $\MM=1$ is the widely used \emph{Single-Member Plurality}
method, also called \emph{First-Past-The-Post}, where each voter votes for
one candidate, and the candidate with the largest number of votes wins;
this simple method has probably been used since pre-historic times.
\BVn{} usually means the multi-member case ($\MM\ge2$); this too
is an ancient method, and it is widely used
in non-political elections.
 Block Vote is well-known for not being proportional; when there are
organized parties, the largest party will get all seats,
\cf{} \refT{TBV}.



\subsubsection{\AVn\abbrev\AVx}\label{AAV}
\relax{Each voter votes for an arbitrary number of candidates. The $\MM$
  candidates with the largest numbers of votes are elected.}

This is also an old method, 
and has been used, even for public elections, since the 19th century,
although much less frequently than  Block Vote.

The method is one of the methods
proposed by \citet{Thiele} in 1895,
viz.~ his ``strong method'', see \refApp{AThopt} below.
The method was reinvented again (and given the name Approval Voting) by Weber 
c.~1976; see further \eg{} \cite{Weber,BramsFishburn:AV,Kilgour}.

For our problems, \AVn{} behaves often similarly to Block
Vote, see \eg{} \refT{TAV}, but note also the difference in
\refTs{TEJRBV} and \ref{TEJRAV}.

\subsubsection{Single Non-Transferable Vote\abbrev\SNTVx}\label{ASNTV}
\relax{Each voter votes for one candidate.
The $\MM$ candidates with the
largest numbers of votes are elected.}

The idea is that a set $\cW$ of voters can concentrate their votes on a
suitable number of candidates (one or several), 
according to the size of $\cW$, and thus get these candidates elected.
In an ideal situation where all voters belong to parties 
and vote as instructed by their party, 
and, moreover,
the parties accurately know  in advance  the number of votes
for each party,
optimal strategies lead to the same result as D'Hondt's method,
see \refT{TSNTV} and \cite[Appendix A.8]{SJV6}.

However, in practice, there are several practical problems.
In particular, the outcome depends heavily on how the votes are distributed
inside the parties; note that a party can get hurt by too much concentration
of votes as well as by too little concentration. Hence, parties are more or
less forced to use schemes of tactical voting.


\subsubsection{Limited Vote\abbrev{\LVxx,\;\LVx{L}}}\label{ALV}
\relax{Each voter votes for at most $\lv$ candidates, where $\lv$ is some
  given   number.
The $\MM$ candidates with the
largest numbers of votes are elected.}
\xfootnote{As for Block Vote, one version  requires each voter to
  vote for exactly $\lv$   candidates.}

Here $\lv\le\MM$. (Although, formally, this is not necessary, and Approval
Voting in \refApp{AAV} could be seen as the case $\lv=\infty$.)
The case $\lv=\MM$ is Block Vote, and $\lv=1$ is SNTV,
so usually only the case $1<\lv<\MM$ is called Limited Vote.

The idea behind Limited Vote is the same as for SNTV, and its problems are
more or less the same, see \eg{}
\cite[p.~193]{Carstairs}, 
\cite{Dodgson},
\cite[p.~27]{Farrell}. 

\subsubsection{\CVn\abbrev\CVx}\label{ACV}
Each voter can choose between voting for a single
candidate or splitting the vote between two or several. 
Usually there are, for practical
reasons, limitations on how the votes may be split.
\xfootnote{
An ideal mathematical version, probably never used,  
is that each voter may give each candidate $i$ a vote with weight
$w_i\ge0$, where these weights are arbitrary real numbers with 
$\sum_iw_i=1$.
}
In one common version, each
voter has a fixed number $\lv$ votes, and can give distribute them
arbitrarily between one or several (up to $\lv$) candidates. 
In another version, 
sometimes called \emph{Equal and Even Cumulative Voting},
the voter has one vote that may be given to a single candidate or 
split between two or several candidates;
this version uses unordered ballots, and a ballot with $n$ names is counted
as $1/n$ vote for each of them.
We let $\CVqx$ denote the mathematically ideal version of this,
where there is no limitation on the number of candidates on each ballot.
(This version has been reinvented and
called \emph{Satisfaction Approval Voting (SAV)} 
\cite{BramsKilgour}.)
Another version of {Equal and Even Cumulative Voting}, used in practice, 
allows at most $S$ names on each ballot.

\subsubsection{\phragmen's unordered method\abbrev{\Phrux}}\label{APhru}
\phragmen's method for unordered ballots
was proposed in 1894 by the Swedish mathematician
\phragmen
\cite{Phragmen1894, Phragmen1895, Phragmen1896,Phragmen1899},
using several different (but equivalent) formulations.
See \cite{SJV9} for a detailed discussion, including an algorithmic definition.
We use here a description from \cite{Phragmen1899}.

The candidates are elected sequentially. 
When a candidate is elected, the participating ballots 
(i.e., the ballots that include this candidate)
incur a total \emph{load} of 1 unit, which has to be 
distributed between them. 
This load is distributed such that
the maximum load on each ballot (inluding loads from earlier elected
candidates) is a small as possible;
furthermore, 
in each round, the candidate is elected 
that makes the resulting maximal load on a ballot as small as possible.
(For example, the first elected candidate is the one that appears on the
largest number of ballots; if she appears on $m$ ballots, these ballots get
a load $1/m$ each.)

The total load when $k$ candidates have been elected is thus $k$.
Note that the load on a given ballot increases when someone on the ballot is
elected, and otherwise stays the same; it never decreases.
We denote the maximum load 
when $k$ candidates have been elected by $t\xx k$.
Thus $0=t\xx0<t\xx1\le t\xx2\dots$ 
(equalities are possible when there are ties).
When $k$ candidates have been elected, say $A_1,\dots,A_k$ in this order,
the load on each
ballot that includes $A_k$ is exactly $t\xx k$.
More generally, the load on a ballot $\gs$ is $t\xx j$, where 
$j:=\max \set{i\le k:A_i\in\gs}$.

\subsubsection{Thiele's \opt{} method\abbrev{\Thoptx}}\label{AThopt}
Let $\psi(n)$ be a given function, which represents the ``satisfication'' of a
voter that sees $n$ of her candidates elected. Thus the total satisfication
if a set $\cE$ is elected is
\begin{equation}\label{sat}
\Psi(\cE):=\sum_{\gs\in\cV} \psi\bigpar{|\gs\cap\cE|},  
\end{equation}
The elected set $\cE$ is the set with $|\cE|=S$ that maximizes 
$\Psi(\cE)$. 

Of course, the method depends on the choice of the satisfaction function $f$.
It is convenient to write $\psi(n)=\sumkn w_k$ for a non-negative sequence
$\ww=(w_k)_1^\infty$ of weights. 
The standard choice is $w_k=1/k$ and thus $\psi(n)=H_n$; this is assumed 
unless another choice is explicitly mentioned.
In particular, $\Thoptx$ always denotes the standard version; we denote
the version defined by another weight sequence $\ww$ by $\Thoptwx{\ww}$.

This method was proposed in 1895 by \citet{Thiele}.
Thiele considered apart from the standard choice $w_k=1/k$ also the
\emph{strong method} with $w_k=1$ and thus $\psi(n)=n$
(which is the same as appoval voting, see \refApp{AAV}) 
and the \emph{weak method} with 
$\psi(n)=1$, $n\ge1$, and thus $w_k=0$ for $k\ge2$.
We denote the weak method by $\Thoptwx\weakw$.

\Thoptn{} was reinvented in 2001 
under the name 
\emph{Proportional Approval Voting (PAV)}
\cite{Kilgour}.

\citet{Thiele} realized
that is was computationally difficult to find the optimal
set $\cE$, and thus proposed also two sequential methods as approximations,
the \emph{addition method} and the \emph{elimination method}, 
\xfootnote{
In Danish:
\emph{Tilf{\o}jelsesreglen} and 
\emph{Udskydelsesreglen}, respectively.
\citet{Thiele} gave also the French names
 \emph{règle d'addition} and \emph{règle de rejet}.
}
described in
the following subsubsections.

\subsubsection{Thiele's addition method\abbrev{\Thax}}\label{ATha}

Candidates are elected one by one. In each round, the candidate $A$ is elected
that maximizes the increase of the satisfaction \eqref{sat},
i.e., $\Psi(\cE\cup \set A)-\Psi(\cE)$, where $\cE$ is the set of already
elected.
(Equivalently, $\Psi(\cE\cup\set A)$ is maximized.)

This is thus the greedy version of Thiele's \opt{} method.
As in \refApp{AThopt}, we assume the standard choice $w_k=1/k$ unless 
we say otherwise. Thus $\Thax$ denotes the standard version,
the weak version is denoted $\Thawx\weakw$,
and the general
version with weights $\ww$ is denoted $\Thawx{\ww}$.

The standard version can, equivalently but more concretely, be described by
the following rule. 
(For general weights, $1/(k+1)$ is replaced by $w_{k+1}$.)

{\em
Seats are awarded one by one.
For each seat, a ballot where $k\ge0$ names already have been
elected is counted as $1/(k+1)$ votes for each remaining candidate on
the ballot. 
The \emph{score} of a (not yet elected) candidate is 
the number of votes for her
counted with these weights.
The candidate with the largest score is elected.
}

Note that the score may change (decrease) from one round to another.

Thiele's addition method was used in Sweden 1909--1921 for the distribution
of seats within parties, see \cite[Appendix D]{SJV9} for details.


\subsubsection{Thiele's elimination method\abbrev{\Thex}}\label{AThe}

Candidates are eliminated one by one until only $S$ remain; these are
elected. In each round, the candidate is eliminated that minimized the
decrease in total satisfaction \eqref{sat}. In other words, if the set of
candidates is $\cC$, we define by backwards recursion 
a sequence of subsets $\cE_i$ for $i=|\cC|,\dots,S$ (in decreasing order)
by $\cE_{|\cC|}=\cC$ and then letting $\cE_i$ be the subset of $\cE_{i+1}$
with $|\cE_i|=i$ that maximizes $\Psi(\cE_i)$.

We consider only the standard version with weights $w_k=1/k$. 
We then have an equivalent, more explicit, description:

{\em
Candidates are eliminated one by one until only $S$ remain; these are
elected.
In each round, a ballot where $k\ge1$ names remain
is counted as $1/k$ votes for each remaining candidate on
the ballot. 
The \emph{score} of a candidate is 
the number of votes for her counted with these weights.
The candidate with the smallest score is eliminated.
}

The elimination method seems to have been ignored
after Thiele's paper.
It was recently reinvented 
under the name \emph{Harmonic Weighting}
as a method for ordering alternatives  for the electronic voting
system \emph{LiquidFeedback}  \cite{LiquidFeedback}.

\subsection{Election methods with ordered ballots}\label{Ao}

In these method, each ballot contains an ordered list of names.
The number of candidates on each ballot is
usually arbitrary, but may be limited to at most, or at least, or exactly
some given number, for 
example $\MM$, the number to be elected;
some versions require each voter to rank all candidates.
(In the latter case, 
each ballot can be seen as a permutation of the set of candidates.)

\subsubsection{Single Transferable Vote (STV) \abbrev\STVx }\label{ASTV}
First, a quota is calculated, nowadays almost always the Droop quota
$V/(S+1)$  (usually but not always rounded to the nearest higher integer), 
\cf{} \refApp{Aquota}.
Each ballot is counted for its first name only
(at later stages, ignoring candidates that have been elected or eliminated).
A candidate whose number of
votes is
at least the quota is elected; the surplus, \ie, the ballots
exceeding the quota, are transferred to
the next (remaining) name on the ballot.
This is repeated as long as some unelected candidate reaches the quota.
If there is no such candidate, and not enough candidates have been elected,
then the candidate with the least number of votes is eliminated, and the 
votes for that candidate are transferred to the next name on the ballot.
At the end, if the number of remaining candidates equals the number of
remaining seats, then all these are elected, even if they have not reached
the quota.

This description is far from a complete definition; many details are
omitted, and can be filled in in different ways, 
see \eg{} 
\cite{Tideman},
\cite{SJV6}, \cite[Appendix E.2]{SJV9}.
In particular, several quite different methods 
are used to transfer the
surplus, 
\xfootnote{
Some versions even involve randomness, \eg{} 
the version of STV used for the lower house,  Dáil,
of the Irish parliament,
where the outcome may depend on the random order in
which ballots are counted.}
but some other details can also vary (which might influence the outcome).
Thus, STV is a
family of election methods rather than a single method. 


\subsubsection{\phragmen's ordered method\abbrev{\Phrox}}\label{APhro}
\phragmen's method for ordered ballots
was proposed in 1913 by 
a Swedish Royal Commission (where \phragmen{} was one of the members)
\cite{bet1913}.  
\xfootnote{
\phragmen's ordered method has since 1921
been part of the Swedish election law (for
distribution of seats within parties; nowadays only in a minor role).
}
As the version for unordered ballots,
\refApp{APhru}, it can be defined in several different, equivalent ways;
see again \cite{SJV9} for a detailed discussion, including an algorithmic
definition. 

We use here, as for the unordered version, a definition using \emph{loads};
the method is defined as \phragmen's unordered method in \refApp{APhru},
with
the added rule that
in each round, a ballot is only counted for the first name on it, ignoring
candidates already elected. 
(For example, the first elected candidate is the one that is first
 on the largest number of ballots.)

\subsubsection{Thiele's ordered method\abbrev{\Thox}}\label{ATho}
Seats are distributed sequentially. In each round, each ballot is only
counted for the first name on it that has not already been elected; if this
is name number $k$ (so the $k-1$ preceding have been elected), then this
ballot is counted as a vote with weight $1/k$.
To avoid confusion, we use ``score'' for the total number of votes for a
given candidate, counted with these weights. (Thus the score may change from
one round to another.)
In each round, the candidate with the largest score is elected.

This version of Thiele's method is not due to \citet{Thiele}, 
who only considered unordered ballots (\refApp{AThopt}--\ref{ATha}).
\xfootnote{Thus our name for the method is historically incorrect, but we 
regard it as a version of Thiele's addition method for unordered ballots.}
The method was proposed in 1912 in the Swedish parliament 
by  Nilson in \"Orebro,
for distribution of seats within parties in parliamentary elections;
it is discussed (and rejected) in \cite{bet1913}.
\xfootnote{
The method was not adopted for general elections, but it was later chosen, 
still within parties, for
elections of committees inside city and county councils, 
and it is still (formally, at least) used for that purpose, see \cite{SJV6,SJV9}.
}

\begin{remark}
 As for Thiele's unordered methods (\refApp{AThopt}--\ref{ATha}), 
an arbitrary given sequence of weights
$w_k$ could be used instead of $1/k$. 
We do not consider this possibility in the present paper.
\end{remark}

\subsubsection{Borda (scoring) methods\abbrev{\Bordax{\ww}}}
\label{ABorda}
Each ballot is counted as $w_1$ votes for the first name, $w_2$ for
  the second name, and so on, according to some given non-increasing
  sequence $\ww=(w_k)_{k=1}^\infty$ of non-negative numbers.
The $S$ candidates with highest scores are elected.

Several different sequences $(w_k)_k$ are used.
The method was suggested 
(for the election of a single person)
by Jean-Charles de Borda in 1770 \cite{Borda} 
with 
$w_k=n-k+1$ where $n$ is the number of candidates, i.e., the sequence
$n,n-1,\dots,1$. (His proposal required each voter to rank all candidates.)
\xfootnote{
The same method was proposed already in 1433 by Nicolas Cusanus for election
of the king (and future emperor, after being crowned by the pope) of the Holy
Roman Empire \cite{McLean-Borda,Pukelsheim:Cusa}.
} 
This scoring rule
is called the \emph{Borda method}; more generally, any scoring rule is
sometimes called a Borda method.

Another common choice of scoring rule uses the \emph{harmonic series},
$w_k=1/k$; we call this the \emph{harmonic Borda method}.
This was
proposed in 1857 by Thomas Hare (who later instead developed and proposed
STV, \refS{ASTV}).
The harmonic Borda method gives in the
party-list case (see \refSS{SSpartylist}) 
the same result as D'Hondt's method.
However, tactical voting where
the voters of a party vote on the same name in different orders may give
quite different results, \cf{} \refT{TBorda}.

The choice $\ww=(1,\dots,1,0,\dots)$ with $w_k=1$ for $k\le L$ and $w_k=0$
for $k>L$, for some given $L\ge1$, 
means that only the first $L$ names on each ballot matter, and that their
order is irrelevant; this Borda method is equivalent to  \LVn{} with $L$
names on each ballot, see \refApp{ALV}.
In particular, the weights $(1,0,0,\dots)$ yield SNTV, \refApp{ASNTV}.

\newpage
\section{Some numerical values}\label{Anum} 

The tables below give,
as illustrations  and to faciliate comparisons,
numerical values for $1\le S\le 5$ for many of the thresholds discussed in
the paper.

\newcommand\hvc{\hfilneg}
\newcommand\tabstrut{\rule{0pt}{3ex}}  
\newcommand\tabbottom{\\[6pt]}  
\newcommand\Sl{\relax}
\newcommand\topell{$\ell=1$}
\newcommand\leftS{$S=1$}

\begin{table}[h]
\hvc
\begin{tabular}{r|lllll}
\Sl
  & \topell&2&3&4&5
\\ 
\hline \tabstrut 
\leftS & $\frac12=0.5$
\\ \tabstrut 
2 & $\frac{1}{3}\doteq0.333$ & 
$\frac23\doteq0.667$
\\ \tabstrut 
3 & 
$\frac {1}{4}=0.25$ & 
$\frac12=0.5$ & 
$\frac {3}{4}=0.75$
\\ \tabstrut 
4 & 
$\frac {1}{5}=0.2$ & 
${\frac {2}{5}}=0.4$ &
${\frac {3}{5}}=0.6$ & 
$\frac {4}{5}=0.8$
\\ \tabstrut 
5 & 
${\frac {1}{6}}\doteq0.167$ & 
${\frac {1}{3}}\doteq0.333$ &
$\frac12=0.5$ &
${\frac {2}{3}}\doteq0.667$ & 
${\frac {5}{6}}\doteq0.833$
\tabbottom
\end{tabular}
  \caption{The ``optimal'' threshold $\ell/(S+1)$, attained by \eg{}
$\pi\party\DH$, 
$\pi\party\Droop$, 
$\pix\tactic\SNTV$, 
$\pi\tactic\CV$, 
$\pi\Phru\PJR$,
$\pi\Thopt\EJR$, 
$\pi\The\same$, 
$\pi\STV\PSC$ (with the Droop quota),
$\pi\Phro\wPSC$.
(\refTs{Tdiv}, \ref{Tquot}, \ref{TSNTV}, \ref{TCV}, \ref{TsamePhru},
\ref{TEJRTh}, \ref{TThe}, \ref{TSTV}, \ref{TPhrowPSC})
}
  \label{tab:optimal}
\end{table}

\begin{table}[htbp]
\hvc
\begin{tabular}{r|lllll}
\Sl
  & \topell&2&3&4&5
\\ 
\hline \tabstrut 
\leftS & $\frac12=0.5$
\\ \tabstrut 
2 & $\frac{1}{3}\doteq0.333$ & 
$\frac34=0.75$
\\ \tabstrut 
3 & 
$\frac {1}{4}=0.25$ & 
$\frac35=0.6$ & 
${\frac {5}{6}}\doteq0.833$
\\ \tabstrut 
4 & 
$\frac {1}{5}=0.2$ & 
$\frac12=0.5$ &
${\frac {5}{7}}\doteq0.714$ &
$\frac {7}{8}=0.875$
\\ \tabstrut 
5 & 
${\frac {1}{6}}\doteq0.167$ & 
${\frac {3}{7}}\doteq0.429$ &
$\frac {5}{8}=0.625$ &
${\frac {7}{9}}\doteq0.778$ & 
${\frac {9}{10}}=0.9$
\tabbottom
\end{tabular}
  \caption{
$\pi\party\StL\ls$ for $1\le\ell\le S\le5$.
(\refT{Tdiv})
}
  \label{tab:StL}
\end{table}

\begin{table}[htbp]
\hvc
\begin{tabular}{r|lllll}
\Sl
  & \topell&2&3&4&5
\\ 
\hline \tabstrut 
\leftS & $\frac12=0.5$
\\ \tabstrut 
2 & $\frac{1}{3}\doteq0.333$ & 
$\frac34=0.75$
\\ \tabstrut 
3 & 
$\frac {1}{4}=0.25$ & 
${\frac {5}{9}}\doteq0.556$ &
${\frac {5}{6}}\doteq0.833$
\\ \tabstrut 
4 & 
$\frac {1}{5}=0.2$ & 
$\frac{7}{16}\doteq0.438$ &
${\frac {2}{3}}\doteq0.667$ & 
$\frac {7}{8}=0.875$
\\ \tabstrut 
5 & 
${\frac {1}{6}}\doteq0.167$ & 
${\frac {9}{25}}=0.36$ &
$\frac{11}{20}=0.55$ &
${\frac {11}{15}}\doteq0.733$ & 
${\frac {9}{10}}=0.9$
\tabbottom
\end{tabular}
  \caption{
$\pi\party\vk\ls$ for $1\le\ell\le S\le5$.
(\refT{Tquot})
}
  \label{tab:LR}
\end{table}

\begin{table}[htbp]
\hvc
\begin{tabular}{r|lllll}
\Sl
  & \topell&2&3&4&5
\\ 
\hline \tabstrut 
\leftS & $\frac12=0.5$
\\ \tabstrut 
2 & 
$\frac {1}{2}=0.5$ & 
$\frac {1}{2}=0.5$
\\ \tabstrut 
3 & 
$\frac {1}{2}=0.5$ & 
$\frac35=0.6$ & 
$\frac {1}{2}=0.5$
\\ \tabstrut 
4 & 
$\frac12=0.5$ &
${\frac {4}{7}}\doteq0.571$ &
${\frac {3}{5}}=0.6$ & 
$\frac12=0.5$ 
\\ \tabstrut 
5 & 
$\frac12=0.5$ &
${\frac {5}{9}}\doteq0.555$ &
${\frac {5}{8}}=0.625$ & 
${\frac {3}{5}}=0.6$ &
$\frac12=0.5$ 
\tabbottom
\end{tabular}
  \caption{The non-monotone
$\pi\EJR\BV\ls$ for $1\le \ell\le S\le 5$.
(\refT{TEJRBV})
}
  \label{tab:BVEJR}
\end{table}

\begin{table}[htbp]
\hvc
\begin{tabular}{r|lllll}
\Sl
  & \topell&2&3&4&5
\\ 
\hline \tabstrut 
\leftS & $\frac12=0.5$
\\ \tabstrut 
2 & 
$\frac {1}{2}=0.5$ & 
${\frac {2}{3}}\doteq0.667$ 
\\ \tabstrut 
3 & 
$\frac {1}{2}=0.5$ & 
$\frac35=0.6$ & 
$\frac {3}{4}=0.75$
\\ \tabstrut 
4 & 
$\frac12=0.5$ &
${\frac {4}{7}}\doteq0.571$ &
${\frac {2}{3}}\doteq0.667$ &
${\frac {4}{5}}=0.8$
\\ \tabstrut 
5 & 
$\frac12=0.5$ &
${\frac {5}{9}}\doteq0.555$ &
${\frac {5}{8}}=0.625$ & 
${\frac {5}{7}}\doteq0.714$ &
${\frac {5}{6}}\doteq0.833$
\tabbottom
\end{tabular}
  \caption{
$\pi\EJR\AV\ls$ for $1\le \ell\le S\le 5$.
(\refT{TEJRAV})
}
  \label{tab:AVEJR}
\end{table}

\begin{table}[htbp]
\hvc
\begin{tabular}{r|lllll}
\Sl
  & \topell&2&3&4&5
\\ 
\hline \tabstrut 
\leftS & $\frac12=0.5$
\\ \tabstrut 
2 & $\frac{1}{3}\doteq0.333$ & 
$\frac23\doteq0.667$ ?
\\ \tabstrut 
3 & 
$\frac {3}{11}\doteq 0.273$  & 
$\frac12=0.5$ ? & 
$\frac {3}{4}=0.75$ ?
\\ \tabstrut 
4 & 
$\frac {7}{31}\doteq0.226$ & 
${\frac {3}{7}}\doteq0.429$ ? &
${\frac {3}{5}}=0.6$ ? & 
$\frac {4}{5}=0.8$ ?
\\ \tabstrut 
5 & 
$??\doteq0.193$  & 
${\frac {7}{19}}\doteq0.368$ ?&
$\frac9{17}\doteq0.529$ ?&
${\frac {2}{3}}\doteq0.667$ ?& 
${\frac {5}{6}}\doteq0.833$ ?
\tabbottom
\end{tabular}
  \caption{Known and conjectured values of $\pi\same\Tha\ls$  and
    $\pi\PJR\Tha\ls$  
for $1\le\ell \le S\le 5$. (\refT{TThaJR}, \refConj{ConjTh})}
  \label{tab:Thasame}
\end{table}

\begin{table}[htbp]
\hvc
\begin{tabular}{r|lllll}
\Sl
  & \topell&2&3&4&5
\\ 
\hline \tabstrut 
\leftS & $\frac12=0.5$
\\ \tabstrut 
2 & $\frac{2}{5}=0.4$ & 
$\frac35=0.6$
\\ \tabstrut 
3 & 
$\frac {12}{35}\doteq 0.343$ & 
$\frac12=0.5$ & 
$\frac {23}{35} \doteq 0.657$
\\ \tabstrut 
4 & 
$\frac {24}{79}\doteq 0.304$ & 
${\frac {18}{41}}\doteq 0.439$ &
${\frac {23}{41}}\doteq 0.561$ & 
$\frac {55}{79}\doteq 0.696$
\\ \tabstrut 
5 & 
${\frac {720}{2621}}\doteq0.275$ & 
${\frac {36}{91}}\doteq0.396$ &
$\frac12=0.5$ &
${\frac {55}{91}}\doteq0.604$ & 
${\frac {1901}{2621}}\doteq0.725$
\tabbottom
\end{tabular}
  \caption{$\pix\tactic\Tho\ls$ for $1\le\ell \le S\le 5$.
  (\refT{TTho})}
  \label{tab:Thotactic}
\end{table}

\begin{table}[htbp]
\hvc
\begin{tabular}{r|lllll}
\Sl
  & \topell&2&3&4&5
\\ 
\hline \tabstrut 
\leftS & $\frac12=0.5$
\\ \tabstrut 
2 & $\frac{2}{5}=0.4$ & 
$\frac23\doteq0.667$
\\ \tabstrut 
3 & $\frac {12}{35}\doteq 0.343$ & 
$\frac47\doteq0.571$ & 
$\frac {3}{4}=0.75$
\\ \tabstrut 
4 & $\frac {24}{79}\doteq 0.304$ & 
${\frac {24}{47}}\doteq0.511$ &
$\frac{2}3\doteq0.667$ &
$\frac{4}5=0.8$
\\ \tabstrut 
5 & ${\frac {720}{2621}}\doteq0.275$ & 
${\frac {48}{103}}\doteq0.466$ &
${\frac {36}{59}}\doteq0.610$ &
$\frac {8}{11}\doteq0.727$ &
$\frac{5}{6}\doteq0.833$
\tabbottom
\end{tabular}
  \caption{$\pi\same\Tho\ls$ for $1\le\ell \le S\le 5$.  (\refT{TTho})}
  \label{tab:Thosame}
\end{table}

\begin{table}[htbp]
\hvc
\begin{tabular}{r|lllll}
\Sl
  & \topell&2&3&4&5
\\ 
\hline \tabstrut 
\leftS & $\frac12=0.5$
\\ \tabstrut 
2 & $\frac{2}{5}=0.4$ & 
$\frac23\doteq0.667$
\\ \tabstrut 
3 & $\frac {12}{35}\doteq 0.343$ & 
$\frac47\doteq0.571$ & 
$\frac {4}{5}=0.8$
\\ \tabstrut 
4 & $\frac {24}{79}\doteq 0.304$ & 
${\frac {24}{47}}\doteq0.511$ &
$\frac {8}{11}\doteq0.727$ &
$\frac{6}{7}\doteq0.857$
\\ \tabstrut 
5 & ${\frac {720}{2621}}\doteq0.275$ & 
${\frac {48}{103}}\doteq0.466$ &
${\frac {48}{71}}\doteq0.676$ &
$\frac {4}{5}=0.8$ &
$\frac{9}{10}=0.9$
\tabbottom
\end{tabular}
  \caption{$\pi\wPSC\Tho\ls$ for $1\le\ell \le S\le 5$.
(\refT{TTho})}
  \label{tab:Thowpsc}
\end{table}

\begin{table}[htbp]
\hvc
\begin{tabular}{r|lllll}
\Sl
  & \topell&2&3&4&5
\\ 
\hline \tabstrut 
\leftS & $\frac12=0.5$
\\ \tabstrut 
2 & $\frac{3}{7}\doteq0.429$ & 
$\frac47\doteq0.571$
\\ \tabstrut 
3 & $\frac {11}{29}\doteq 0.379$ & 
$\frac12=0.5$ & 
$\frac {18}{29}=0.621$
\\ \tabstrut 
4 & 
$\frac {25}{73}\doteq0.342$ &
${\frac {22}{49}}\doteq0.449$&
${\frac {27}{49}}\doteq0.551$&
${\frac {48}{73}}\doteq0.657$
\\ \tabstrut 
5 & 
${\frac {137}{437}}\doteq0.314$ &
${\frac {25}{61}}\doteq0.410$ &
$\frac12=0.5$ &
${\frac {36}{61}}\doteq0.59$ &
${\frac {300}{437}}\doteq0.686$
\tabbottom
\end{tabular}
  \caption{$\pix\Borda{1/k}\tactic\ls$ for $1\le\ell \le S\le 5$.
(\refT{TBorda})}
  \label{tab:Btactic}
\end{table}

\begin{table}[htbp]
\hvc
\begin{tabular}{r|lllll}
\Sl
  & \topell&2&3&4&5
\\ 
\hline \tabstrut 
\leftS & $\frac12=0.5$
\\ \tabstrut 
2 & $\frac{3}{7}\doteq0.429$ & 
$\frac23\doteq0.667$
\\ \tabstrut 
3 & $\frac {11}{29}\doteq 0.379$ & 
$\frac35=0.6$ & 
$\frac {3}{4}=0.75$
\\ \tabstrut 
4 & 
$\frac {25}{73}\doteq0.342$ &
${\frac {11}{20}}=0.55$ &
${\frac {9}{13}}\doteq0.692$ &
$\frac{4}5=0.8$
\\ \tabstrut 
5 & 
${\frac {137}{437}}\doteq0.314$ &
${\frac {25}{49}}\doteq0.510$ &
${\frac {11}{17}}\doteq 0.647$ &
$\frac{3}4=0.75$ &
$\frac{5}{6}\doteq0.833$
\tabbottom
\end{tabular}
  \caption{$\pi\Borda{1/k}\same\ls=\pi\Borda{1/k}\wPSC\ls$ 
for $1\le\ell \le S\le 5$.
(\refT{TBorda})}
  \label{tab:Bsame}
\end{table}

\clearpage

\newcommand\vol{\textbf}
\newcommand\jour{\emph}
\newcommand\book{\emph}
\newcommand\inbook{\emph}
\def\no#1#2,{\unskip#2, no. #1,} 
\newcommand\toappear{\unskip, to appear}

\newcommand\arxiv[1]{\texttt{arXiv:#1}}
\newcommand\arXiv{\arxiv}

\def\nobibitem#1\par{}

\newcounter{bibrub}
\newcommand\bibrub[1]{\stepcounter{bibrub} \medskip 
  \centerline{\textsc{\Alph{bibrub}. #1}} 
  \medskip}
\newcounter{bibsubrub}[bibrub]
\newcommand\bibsubrub[1]{\stepcounter{bibsubrub} \medskip 
  \centerline{\textsc{\Alph{bibrub}\arabic{bibsubrub}. #1}} 
  \medskip}

\newcommand\urldag[1]{(#1)}
\newcommand\urlq[2]{\url{#1} \urldag{#2}}
\newcommand\xurl{\\\url}
\newcommand\xurlq{\\\urlq}
\newcommand\supplera[1]{#1}


\end{document}